\documentclass[a4paper,twocolumn,11pt,accepted=2021-05-18]{quantumarticle}
\pdfoutput=1
\usepackage{amsmath,amsthm,amssymb,amsfonts}

\usepackage[sort&compress]{natbib}
\setcitestyle{numbers,square,comma} 

\usepackage{setspace}

\usepackage{bm}
\usepackage{booktabs} 
\usepackage{color} 
\usepackage[usenames,dvipsnames]{xcolor}
\usepackage{multicol}

\usepackage{hyperref}
\hypersetup{colorlinks, citecolor=magenta, filecolor=blue, linkcolor=blue, urlcolor=magenta}
\usepackage[english]{babel}
\usepackage[T1]{fontenc}
\usepackage{mathtools}
\usepackage{cleveref} 
\usepackage{bm,bbm}
\usepackage[normalem]{ulem} 
\usepackage{graphicx}
\graphicspath{ {./figures/}} 
\usepackage{afterpage}
\usepackage{titletoc} 
\usepackage[page,header]{appendix} 
\usepackage{cleveref}
\usepackage[caption=false]{subfig}

\usepackage{ragged2e}
\usepackage{dblfloatfix}
\usepackage{float}

\makeatletter
\let\newfloat\newfloat@ltx
\makeatother
 
\usepackage[noend]{algpseudocode}
\usepackage{algorithm}
\newlength{\tablen}
\newcommand\tabbox[2][\tablen]{\makebox[#1][l]{#2}}

\usepackage{stackengine,scalerel}
\stackMath

\usepackage{quantinmacros}

\title{Modeling and mitigation of cross-talk effects in readout noise with applications to the Quantum Approximate Optimization Algorithm}

	\author{Filip B. Maciejewski}
	\email{filip.b.maciejewski@gmail.com}
	\affiliation{Center for Theoretical Physics, Polish Academy of Sciences, Al. Lotnik\'ow 32/46, 02-668 Warsaw, Poland}
	
	\author{Flavio Baccari}
	\affiliation{Max-Planck-Institut f\"ur Quantenoptik, Hans-Kopfermann-Stra{\ss}e 1, 85748 Garching, Germany}
	
	\author{Zolt\'an Zimbor\'as}
	\affiliation{Wigner Research Centre for Physics,
		H-1525 Budapest, P.O.Box 49, Hungary}
	\affiliation{BME-MTA Lend\"ulet Quantum Information Theory Research Group, Budapest, Hungary}
	\affiliation{Mathematical Institute, Budapest University of Technology and Economics,
		P.O.Box 91, H-1111, Budapest, Hungary}
		
	\author{Micha\l\ Oszmaniec}
	\email{oszmaniec@cft.edu.pl}
	\affiliation{Center for Theoretical Physics, Polish Academy of Sciences, Al. Lotnik\'ow 32/46, 02-668 Warsaw, Poland}
	
\begin{document} 
\twocolumn[
	\begin{@twocolumnfalse}
		\maketitle
\begin{abstract}
Measurement noise is one of the main sources of errors in currently available quantum devices based on superconducting qubits. 
At the same time, the complexity of its characterization and mitigation often exhibits exponential scaling with the system size.
In this work, we introduce a correlated measurement noise model that can be efficiently described and characterized, and which admits effective noise-mitigation on the level of marginal probability distributions.
Noise mitigation can be performed up to some error for which we derive upper bounds. 
Characterization of the model is done efficiently using Diagonal Detector Overlapping Tomography -- a generalization of the recently introduced Quantum Overlapping Tomography to the problem of reconstruction of readout noise with restricted locality. 
The procedure allows to characterize $k$-local measurement cross-talk on $N$-qubit device using $O\rbracket{k2^k\log\rbracket{N}}$ circuits containing random combinations of X and identity gates.
We perform experiments on 15 (23) qubits using IBM's (Rigetti's) devices to test both the noise model and the error-mitigation scheme, and obtain an average reduction of errors by a factor $>22$ ($>5.5$) compared to no mitigation.
Interestingly, we find that correlations in the measurement noise do not correspond to the physical layout of the device.
Furthermore, we study numerically the effects of readout noise on the performance of the Quantum Approximate Optimization Algorithm (QAOA). 
We observe in simulations that for numerous objective Hamiltonians, including random MAX-2-SAT instances and the Sherrington-Kirkpatrick model, the noise-mitigation improves the quality of the optimization.  
 Finally, we provide arguments why in the course of QAOA optimization the estimates of the local energy (or cost) terms often behave like uncorrelated variables, which greatly reduces sampling complexity of the energy estimation compared to the pessimistic error analysis.
 We also show that similar effects are expected for Haar-random quantum states and states generated by shallow-depth random circuits.
\end{abstract}
\end{@twocolumnfalse}]
	
\twocolumn
\section{Introduction}
\subsection{Motivation}
	Outstanding progress has been made in the last years on the path to development of scalable and fully functional quantum devices.
	With state of the art quantum processors reaching a scale of 50-100 qubits \cite{Arute2019quantum}, the scientific community is approaching a regime in which quantum systems cannot be modeled on modern supercomputers using known methods \cite{Villalonga2020establishing}. This situation is unarguably exciting since it opens possibilities for the first demonstrations of quantum advantage and, potentially, useful applications \cite{farhi2016quantum, Moll2018Quantum,Preskill2018}. 
	At the same time, however, it entails a plethora of non-trivial problems related to fighting effects of experimental imperfections on the performance of quantum algorithms. As quantum devices in the near-term will be unable to implement proper error correction \cite{Preskill2018}, various methods of \emph{noise mitigation} have been recently developed \cite{
	Wallman2016noise, 
	Li2017efficient, 
	Temme2017error,
	Endo2018, 
	Kandala2019error,
	sun2020mitigating,
	huggins2020virtual,
	Maciejewski2020mitigation,
	Chen2019detector,
	Bravyi2020mitigating}. 
	Those methods aim at reducing the effects of errors present in quantum gates and/or in quantum measurements. In this work, we focus on the latter, i.e., noise affecting quantum detectors.
	
Indeed, it has been found in currently available quantum devices that the noise affecting measurements is quite significant. Specifically,  errors of the order of a few percent in a single qubit measurement and non-negligible effects of cross-talk were reported  \cite{
		Arute2019quantum,
		Maciejewski2020mitigation,
		Chen2019detector, 
		Bravyi2020mitigating,
		Geller2020efficient}.
		Motivated by this, a number of methods to characterize and/or reduce measurement errors were proposed \cite{ 
		Maciejewski2020mitigation, 
		Chen2019detector,
		Geller2020efficient, 
		Bravyi2020mitigating,
		Geller2020rigorous,
		Nachman2019unfolding,
		Kwon2020hybrid,
		Hamilton2020scalable,
		lilly2020modeling,
		funcke2020measurement,
		zheng2020bayesian}.
	Readout noise mitigation usually relies on the classical post-processing of experimental statistics, preceded by some procedure of noise characterization.
	Existing techniques typically suffer from the curse of dimensionality due to characterization cost, sampling complexity, and the complexity of post-processing -- which scale exponentially with the number of qubits $N$. 
	Fortunately, some interesting problems in quantum computing do not require measurements to be performed across the whole system. 
An important class of algorithms that have this feature are the Quantum Approximate Optimization Algorithms (QAOA) \cite{farhi2014quantum}, which only require the estimation of a number of few-particle \emph{marginals}.
	QAOA is a heuristic, hybrid quantum-classical optimization technique \cite{farhi2014quantum, cerezo2020variational}, that was later generalized as a standalone ansatz \cite{hadfield2019quantum} shown to be computationally universal \cite{lloyd2018quantum,Morales2020universality}.
	In its original form, QAOA aims at finding approximate solutions for hard combinatorial problems, notable examples of which are maximum satisfiability (SAT) \cite{Hansen1990algorithms}, maximum cut (MAXCUT) \cite{Guerreschi2019QAOA}, or the Sherrington-Kirkpatrick (SK) spin-glass model \cite{Panchenko2012Sherrington, farhi2019quantum}. 
	Regardless of the underlying problem, the main QAOA subroutine is the estimation of the energy of the local classical Hamiltonians on a quantum state generated by the device. 	Those Hamiltonians are composed of a  number of a few-body commuting operators, and hence estimation of energy can be done via estimation of the local terms (which can be performed simultaneously). Since the estimation of marginal distributions is the task of low sampling and post-processing complexity, this suggests that error mitigation techniques can be efficiently applied in QAOA.
	In this work, we present a number of contributions justifying the usage of measurement error-mitigation in QAOA, even in the presence of significant cross-talk effects. 
	
\subsection{Results outline}	Our first contribution is to provide an efficiently describable measurement noise model that incorporates asymmetric errors and cross-talk effects.
	Importantly, our noise model admits efficient error-mitigation on the marginal probability distributions, which can be used, e.g., for improvement of the performance of variational quantum algorithms.
	We show how to efficiently characterize the noise using a number of circuits scaling logarithmically with the number of qubits. 
	To this aim, we generalize the techniques of the recently introduced Quantum Overlapping Tomography (QOT) \cite{Cotler2020quantum}
	to the problem of readout noise reconstruction.
	Specifically, we introduce notion of Diagonal Detector Overlapping Tomography (DDOT) which allows to reconstruct noise description with $k$-local cross-talk on $N$-qubit device using $O\rbracket{k2^k \log\rbracket{N}}$ quantum circuits consisting of single layer of $X$ and identity gates.
	Furthermore, we explain how to use that characterization to mitigate noise on the marginal distributions and provide a bound for the accuracy of the mitigation.
	Importantly, assuming that cross-talk in readout noise is of bounded locality, the sampling complexity of error-mitigation is not significantly higher than that of the starting problem of marginals estimation.
	
	We test our error-mitigation method in experiments on 15 qubits using IBM's device and on 23 qubits using Rigetti's device, both architectures based on superconducting transmon qubits \cite{Koch2007}.
	We obtain a  significant advantage by using a correlated error model for error mitigation in the task of ground state energy estimation. Interestingly, the locality structure of the reconstructed errors in these devices does not match the spatial locality of qubits in these systems.

We also study statistical errors that appear in the simultaneous estimation of multiple local Hamiltonian terms that appear frequently in QAOAs. 
	In particular, we provide arguments why one can expect that the estimated energies of local terms behave effectively behave as uncorrelated variables,
	for the quantum states appearing at the beginning and near the end of the QAOA algorithm. This allows to prove significant reductions in sampling complexity of the total energy estimation (compared to the worst-case upper bounds).

	Finally, we present a numerical study and detailed discussion of the possible effects that measurement noise can have on the performance of the Quantum Approximate Optimization Algorithm.  This includes the study of how noise distorts the quantum-classical optimization in QAOAs. We simulate the QAOA protocol on an 8-qubit system affected by correlated measurement noise inspired by the results of IBM's device characterization.
    For a number of random Hamiltonians, we conclude that the error-mitigation highly improves the accuracy of energy estimates, and can help the optimizer to converge faster than when no error-mitigation is used.

\subsection{Related works}
The effects of simple, uncorrelated noisy quantum channels on QAOAs were analyzed in Refs.~\cite{xue2019effects, marshall2020characterizing, alam2019analysis}.
In particular, the symmetric bitflip noise analyzed in \cite{xue2019effects} can be used also to model symmetric, uncorrelated measurement noise - a type of readout noise which has been demonstrated to be not very accurate in currently available devices based on transmon qubits \cite{Maciejewski2020mitigation,Geller2020efficient}. 
A simple readout noise-mitigation technique on the level of marginals for QAOAs was recently experimentally implemented on up to 23 qubits in a work by Google AI Quantum team and collaborators \cite{arute2020quantum}. 
Importantly, the authors assumed uncorrelated noise on each qubit -- we believe that our approach which accounts for correlations could prove beneficial in those kinds of experiments. 
In a similar context, readout noise-mitigation using classical post-processing on global probability distributions was implemented in QAOA and Variational Quantum Eigensolver (VQE) experiments on up to 6 qubits \cite{montanaro2020compressed, gokhale2020optimized}. 
While for such small system sizes it is possible to efficiently perform global error mitigation, we emphasize that our approach based on the noise-mitigation on the marginals could be performed also in larger experiments of this type. 

Alternative methods of characterization of correlated measurement noise were recently proposed in Refs.~\cite{Geller2020efficient,Hamilton2020scalable,Bravyi2020mitigating}. 
Out of the above-mentioned works, perhaps the most related to ours in terms of studied problems is Ref.~\cite{Bravyi2020mitigating}, therefore we will now comment on it thoroughly.
The authors introduced a correlated readout noise model based on Continuous Time Markov Processes (CTMP). 
They provide both error-characterization and error-mitigation methods.
The CTMP noise model assumes that the noisy stochastic process in a measurement device can described by a set of two-qubit generators with corresponding error rate parameters. 
The total number of parameters required to describe a generic form of such noise for $N$-qubit device is $2N^2$.
The authors propose a method of noise characterization that requires preparation of a set of suitably chosen classical states and performing a post-selection on the noise-free outcomes (i.e., correct outcomes given known input classical state) on the subset of $\rbracket{N-2}$ qubits.
Since probability of noise-free outcomes can be exponentially small even for the uncorrelated readout noise, this method can prove relatively costly in practice.
We believe that our DDOT characterization technique could prove useful in reducing the number of parameters needed to describe CTMP model (by showing which pairs of qubits are correlated, and for which the cross-talk can be neglected), hence allowing for much more efficient version of noise reconstruction presented in Ref.~\cite{Bravyi2020mitigating}.
The authors also present a novel noise-mitigation method that allows to estimate the noise-reduced expected values of observables.
The method is based on decomposition of inverse noise matrix into the linear combination of stochastic matrices, and constructing a random variable (with the aid of classical randomness and post-processing) that agrees in the expected value of noise-free observable.
In general, both sampling complexity and classical post-processing cost scale exponentially with the number of qubits.
Since the authors aim to correct observables with arbitrary locality, the problem they consider is different to our approach that aims to correct only \emph{local} observables (and therefore does not exhibit exponential scalings).
It is an interesting problem to see whether the CTMP noise model can also be interpreted on the level of marginal probability distributions in a way that allows for mitigation analogous to ours in terms of complexity.

\section{Correlated readout noise model}\label{sec:noise_model}
   \subsection{Preliminaries}  
    In quantum mechanics, the most general measurement $\M$ that can be performed is represented by a Positive Operator-Valued Measure (POVM) \cite{Peres2006}. 
    A POVM $\M$ with $r$ outcomes on a $d$-dimensional Hilbert space is a set of positive-semidefinite operators that sum up to identity, i.e., 
	\begin{align}	\label{eq:povm_def}
	\M&=\cbracket{M_i}_i^r\ ,		
	&\forall i\ M_{i}&\geqslant\ 0\ ,&\sum_{i=1}^{r}\ M_{i}=\iden ,
	\end{align}
	where $\iden$ is the identity operator.
	When one performs a measurement $\M$ on the quantum state $\rho$, the probability of obtaining outcome $i$ is given by Born's rule: $\text{Pr}\rbracket{i|\M,\ \rho} = \Tr{\rho M_i}$. 
	
	In quantum computing a perfect measurement is often modeled as a so-called  \emph{projective} measurement $\P = \cbracket{P_i}_i^r$ for which the measurement operators, in addition to the requirements from Eq.~\eqref{eq:povm_def}, are also projectors, i.e.,  $\forall_i P_i^2=P_i$.
    While non-projective measurements have many applications in certain quantum information processing tasks, in this work we are interested in using them to model imperfections in measurement devices.
    The relationship between an ideal measurement $\P$ and its noisy implementation $\M$ can be modeled by the application of a generic quantum channel.
    Recently it has been experimentally demonstrated for superconducting quantum devices that in practice those channels, to a good approximation, belong to a restricted class of stochastic maps \cite{Chen2019detector,Maciejewski2020mitigation}, which we will refer to as 'classical measurement noise'.
	In such a model, the relation between $\M$ and $\P$ is given by some stochastic transformation $\Lambda$. 
	Namely, we have $\M=\Lambda \P$, i.e., $M_i = \sum_j \Lambda_{ij}P_j$. 
	Due to the linearity of Born's rule, it follows that probabilities $\p^{\textit{noisy}}$ from which noisy detector samples are related to the noiseless probabilities $\p^{\textit{ideal}}$ via the same stochastic map, hence \cite{Maciejewski2020mitigation}
	\begin{align}\label{eq:classical_noise}
	    \p^{\textit{noisy}} = \Lambda \p^{\textit{ideal}} \ .
	\end{align}
	Specifically, in the convention where probability vectors are columns, the noise matrix $\Lambda$ is left-stochastic, meaning that each of its columns contains non-negative numbers that sum up to 1.
	Such noise is thus equivalent to a stochastic process, in which an outcome from a perfect device probabilistically changes to another (possibly erroneous) one.
	Equation~\eqref{eq:classical_noise}  suggests a simple way to mitigate errors on the noisy device -- via left-multiplying the estimated statistics by the inverse of noise matrix $\Lambda^{-1}$ \cite{Chen2019detector,Maciejewski2020mitigation}. 
	From stochasticity of $\Lambda$ it follows that its inverse does preserve the sum of the elements of probability vectors, however it may introduce some unphysical (i.e., lower than 0 or higher than 1) terms in corrected vector. 
	A common practice in such a scenario is to solve optimization problem
	\begin{align}
	    \p=\mathrm{argmin}_{\mathbf{q}} ||\mathbf{q}-\Lambda^{-1}\p^{\text{noisy}}||_{2}^2 \ ,
	\end{align}
	where minimization goes over all proper probability distributions. 
	This introduces additional errory in the final estimations which can be easily upper-bounded \cite{Maciejewski2020mitigation}.
	
	Before going further, we note that while multiplication by $\Lambda^{-1}$ is perhaps the most natural (and simple) method to reduce the noise and has been shown to be useful in practical situations \cite{Chen2019detector,Maciejewski2020mitigation}, there exist more sophisticated techniques of noise-mitigation that do not exhibit this problem.
	For example, Iterative Bayesian Unfolding \cite{Nachman2019unfolding} always returns physical probability vectors.
	We note that our noise model and noise-mitigation on marginal probability distributions (to be introduced in following sections) is consistent with using methods different than $\Lambda^{-1}$ correction, and we intend to test them in the future.

 To finish this introduction, we recall that, as mentioned above,  Eq.~\eqref{eq:classical_noise}  does not present the most general model of quantum measurement noise. Specifically, \textit{coherent} errors might occur, and reduce the effectiveness of error-mitigation (a detailed analysis of this effect was presented in \cite{Maciejewski2020mitigation}).
	However, it was validated experimentally on multiple occasions, that coherent errors in superconducting quantum devices are small, and the error-mitigation by classical post-processing was shown to work very well in few-qubit scenarios
	\cite{Maciejewski2020mitigation,Chen2019detector}.

\subsection{Correlations in readout noise} The size of the matrix $\Lambda$ scales exponentially with the number of qubits. Thus, if one wants to estimate such a generic $\Lambda$ using standard methods, both the number of circuits and sampling complexity scale exponentially. Indeed, the standard method of reconstructing $\Lambda$ is to create all the $2^n$ computational basis and estimate the resulting probability distributions (which constitute the columns of $\Lambda$). 	We refer to such characterization as Diagonal Detector Tomography (DDT), since it probes the diagonal elements of the measurement operators describing the detector.
This is restricted version of more general Quantum Detector Tomography (QDT) \cite{Lundeen2008,Hradil2004,Fiurasek2001,Gianani2020compressively}.

	These complexity issues can be circumvented if one assumes some \emph{locality} structure in the measurement errors. 
	For example, in the simplest model with completely uncorrelated readout noise, the $\Lambda$ matrix is a simple tensor product of single-qubit error matrices $\Lambda_{Q_i}$
	\begin{align}\label{eq:uncorrelated_noise}
	    \Lambda = \bigotimes_i \Lambda_{Q_i} \qquad \text{(uncorrelated noise)} .
	\end{align}
    In this model, we need to estimate only single-qubit matrices, which renders the complexity of the problem to be linear in the number of qubits. 
    However, for contemporary quantum devices based on superconducting qubits, it was demonstrated that such a noise model is not very accurate due to the cross-talk effects \cite{Arute2019quantum,Maciejewski2020mitigation,Chen2019detector,Geller2020efficient}.
    At the same time, the completely correlated noise is not realistic as well, which motivates the search for a model that can account for correlations in readout errors while still giving an efficient description of $\Lambda$.

	In this work, we propose such a model and give a method to characterize it.  Let us lay out the basic concepts of our model.
    Consider the correlated errors between some group of qubits $\cluster_{\cindex}$. 
    The most general way of describing those errors is to treat the qubits in $\cluster_{\cindex}$ as a single object, i.e., to always consider their measurement outcomes together. 
    In terms of the noise matrix description, this means that the noise matrix on $\cluster_{\cindex}$ is some generic $\Lambda_{\cluster_{\cindex}}$ acting on $\cluster_{\cindex}$.
    This gives rise to the first basic object in our model -- the \emph{clusters} of qubits.
    The cluster $\cluster_{\cindex}$ is a group of qubits with correlations between them so strong, that one can not consider outcomes of their measurements separately.
    At the same time, it is unlikely that in actual devices the correlations between all the qubits will be so strong that one should assign them all to a single cluster.
    This motivates the introduction of another, milder possibility.
    Consider a measurement performed on qubits in cluster $\cluster_{\cindex}$ and some other qubits $\N_{\cindex}$ (not being in that cluster).
    It is conceivable to imagine some complicated physical process, which results in the situation in which the noise matrix $\Lambda_{\cluster_{\cindex}}$ on cluster $\cluster_{\chi}$ slightly depends on the state of the qubits in $\N_{\cindex}$.
    To account for that, we introduce the second basic object of our model -- the \emph{neighborhood} of the cluster.
    The neighborhood $\N_{\cindex}$ of a cluster ${\cluster_{\cindex}}$ is a group of qubits the state of which $\emph{just before the measurement}$ affects slightly the noise matrix acting on the cluster ${\cluster_{\cindex}}$.

    For example, if $\cluster_{\cindex}$ contains only a single qubit, say $Q_0$, it is possible that due to some effective ferromagnetic-type interaction, the probability of erroneously detecting state ``0'' of $Q_0$ as ``1'' rises when the neighboring qubit $Q_1$ is in state ``1'' (compared to when it is in state $'0'$).
    
     A notion related to our ``neighborhood'' has appeared in recent literature. Specifically in the context of measurement error characterization, 'spectator qubits' are the qubits that affect measurement noise on other qubits \cite{Geller2020efficient,Hamilton2020scalable}.
       However, so far this effect was treated rather as an undesired complication, while here it is an inherent element of the proposed noise model.

\begin{figure}[!t]
\begin{center}
\captionsetup[subfigure]{format=default,singlelinecheck=on,justification=RaggedRight}
\subfloat[\label{fig:generic_device}]
        {\includegraphics[width=0.24\textwidth]{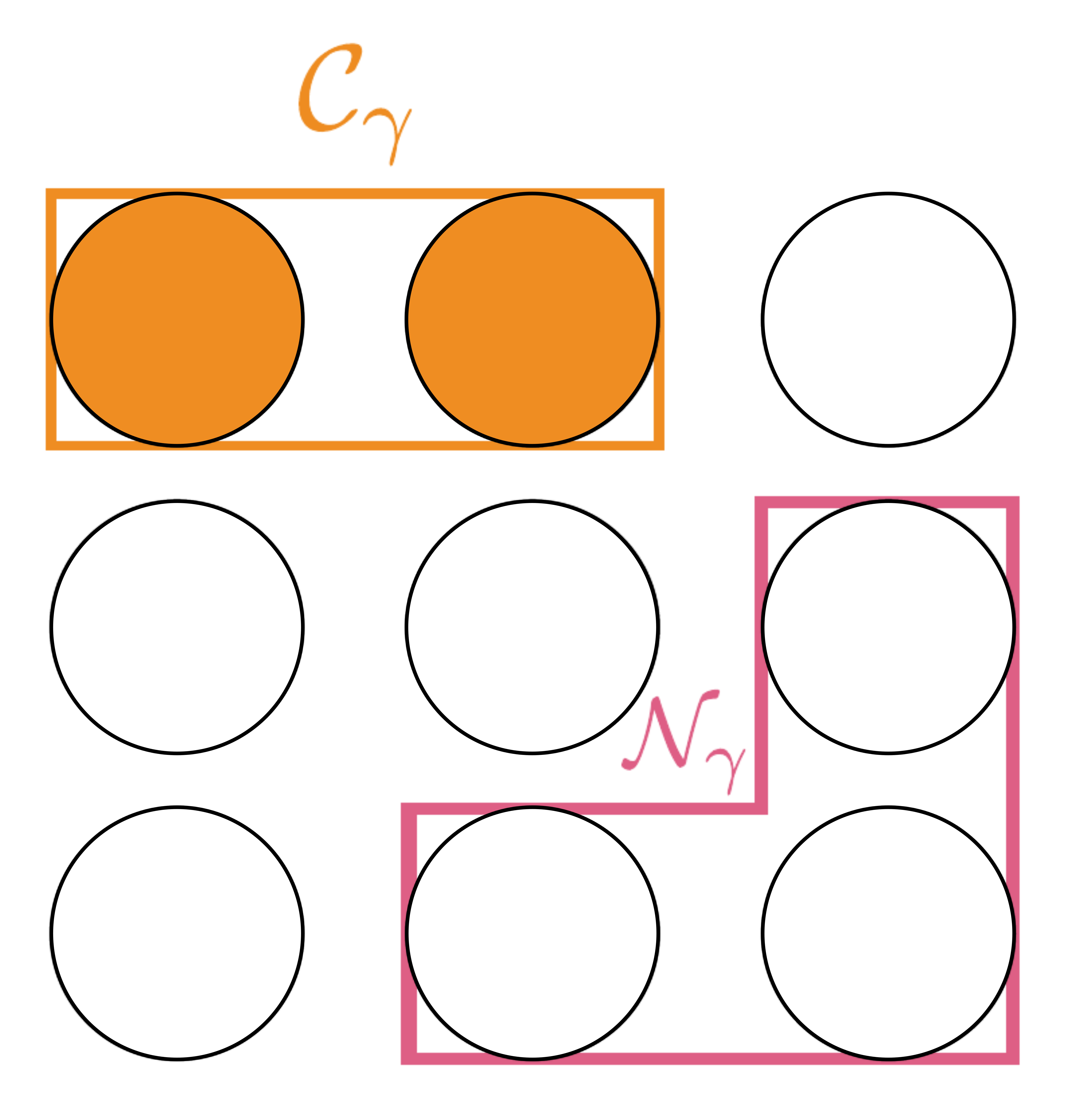}}
\subfloat[\label{fig:exemplary_correlations4q}]
        {\includegraphics[width=0.21\textwidth]{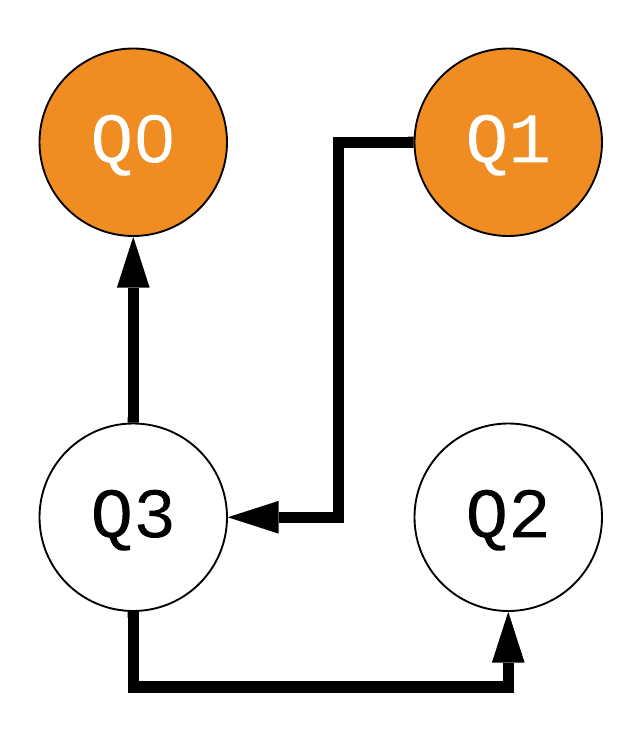}}
    \caption{
Exemplary correlations in the measurement noise that can be captured by our model. Each circle represents a qubit. 
Subfigure a) represents a 9-qubit device with some cluster $\mathcal{C}_{\cindex}$ (orange envelope) consisting of two qubits.
The noise on the qubits from that cluster is dependent on the state of qubits from its neighborhood $\N_{\cindex}$ (magenta envelope). 
Note that the ``neighborhood'' does not have to correspond to spatial arrangement of qubits in the device.
Subfigure b) shows a more detailed example of a four-qubit device.
Qubits $Q0$ and $Q1$ are in one cluster, which is indicated by coloring them with the same color.
Qubits $Q2$ and $Q3$ are white, meaning they do not belong to any cluster.
 Qubits at the beginning of the arrow are the neighbors of the qubits at the end of the arrow. 
 Explicitly, the clusters in the example are $\mathcal{C}_1 = \cbracket{0, 1},\ \mathcal{C}_2 = \cbracket{2},$ and $\mathcal{C}_3 = \cbracket{3}$, while their neighborhoods are $\N_1 = \cbracket{3},\ \N_2 = \cbracket{3}$, and $\N_3 = \cbracket{1}$. 
 Correlations in the readout errors for qubits $Q_0$ and $Q_1$ can be arbitrary, while for the rest of the qubits the dependencies are restricted by the structure of clusters and the neighbors.
 In particular, the noise matrix on $Q_3$ depends on the state of $Q_1$, while the state of $Q_3$ affects the noise on qubits $Q_0$ and $Q_2$.  At the same time, qubit $2$ does not affect the noise matrix on any other qubit. See the description in the main text.}
 \end{center}
\end{figure}

Now we are ready to provide an efficient noise model.
We construct a global noise matrix $\Lambda$ with matrix elements of the following form 
\begin{align}\label{eq:noise_model_correlated}
	\Lambda_{X_1\dots X_N|Y_1 \dots Y_N} = \prod_{\chi} \Lambda^{ \Y_{\N_{\chi}}}_{\X_{\cluster_\chi}|\Y_{\cluster_\chi}} \ .
\end{align}
\noindent 
In the next subsection we give some illustrative examples, but first  let us thoroughly describe the notation used in the above equation.
	A collection $\cbracket{\cluster_{\cindex}}_{\cindex}$ gives us the partitioning of the set of all qubits. 
	Explicitly, $\cluster_{\cindex} \cap \cluster_{\cindex'} = \emptyset$ if $\cindex \neq \cindex'$ and $\cup_{\cindex}\cluster_{\cindex} = \sbracket{N}$, where $N$ is the number of qubits. To each cluster $\cluster_{\cindex}$ the model associate its \emph{neighborhood} $\N_{\cindex}$. 	Equation \eqref{eq:noise_model_correlated} can be now understood in the following way. The noise matrix $\Lambda^{\Y_{\N_{\cindex}}}$ describing the measurement noise occurring in cluster $\cluster_{\cindex}$ depends on the state \emph{just before measurement} of the qubits in the neighborhood $\N_{\cindex}$ of that cluster (hence the superscript $\Y_{\N_{\cindex}}$ denoting that state). 
	Importantly, each $\Lambda^{\Y_{\N_{\cindex}}}$ is left-stochastic for any state of the neighbors.
	By $\X_{\cluster_{\cindex}}$ (or $\Y_{\cluster_{\cindex}}$) we denote bit-strings of qubits belonging to cluster $\cluster_{\cindex}$ which were measured (or put inside the device just before measurement).
	Finally,  $\Y_{\N_{\cindex}}$ indicates the bit-string denoting the state just before the measurement of the qubits from the neighborhood $\N_{\cindex}$ of the cluster $\cluster_{\cindex}$ (see Fig.~\ref{fig:generic_device} for illustration).   
	Note that in general the correlations in measurement errors (expressed by the structure of the clusters and neighborhoods) do not need to be directly correlated with the physical layout of the device. 
	In general $\Lambda$ in Eq.\eqref{eq:noise_model_correlated} is specified by only $\sum_{\cindex}(2^{|\cluster_{\cindex}|})2^{|\N_{\cindex}|}$ parameters, where $|\cluster_{\cindex}|$ and $|\N_{\cindex}|$ are sizes of $\cindex$'th cluster and its neighbourhood respectively.
	Therefore this description is \emph{efficient} provided sizes of the clusters and their neighborhoods are bounded by a constant.

\subsection{Illustrative examples}   
In what follows present  examples of readout correlation structures that can be described with our model.
It is instructive to start with a simple example of a hypothetical four-qubit device depicted in  Fig.~\ref{fig:exemplary_correlations4q}.
    Note that in this example we have only one non-trivial (i.e., with size $\geq2$) cluster. 
    Let us write explicitly the matrix elements of the global noise matrix acting on that exemplary 4-qubit device
    \begin{align}\label{eq:example4q}
        \Lambda_{X_0X_1X_2X_3|Y_0Y_1Y_2Y_3} = \Lambda^{Y_3}_{X_0X_1|Y_0Y_1}\Lambda^{Y_3}_{X_2|Y_2}\Lambda^{Y_1}_{X_3|Y_3} \ .
    \end{align}
    Note that on the RHS of Eq.~\eqref{eq:example4q}, the superscript $Y_3$ appears two times, indicating that noise matrices on the cluster $\cluster_1$ and on the cluster $\cluster_2$ both depend on the state just before measurement of the qubit 3.
    At the same time, there is no superscript $Y_2$, which follows from the fact that qubit 2 does not affect the noise on any other qubits.
    Note that while a generic noise matrix on $4$ qubits would require reconstruction of $16 \times 16$ matrix, here we need a number of smaller dimensional matrices to fully describe the noise.

We now move to a more general readout error model, which is particularly inspired by current superconducting qubits implementations of quantum computing devices. Consider a collection of qubits arranged on a device with a limited connectivity. This can be schematically represented as a graph $\mathcal{G}(V,E)$, where each vertex in $V$ represents a qubit and each edge in $E$ connects two qubits that can interact in the device.
In such a scenario, a natural first step beyond an uncorrelated readout noise model can be a nearest-neighbour correlated model, where readout errors on each qubit are assumed to be influenced at most by the state of the neighbouring ones. By using the notation introduced in the previous section, we can represent such a model by associating a single-qubit cluster to each vertex, $ C_i = \lbrace i \rbrace $, for ${i \in V} $, and defining the neighbourhoods according to the graph structure, namely $\mathcal{N}_i = \lbrace j | (i,j) \in E \rbrace$. The global noise matrix then reads
\begin{equation}
	\Lambda_{X_1\dots X_{|V|}|Y_1 \dots Y_{|V|}} = \prod_{i \in V} \Lambda^{ \Y_{\N_{i}}}_{\X_{i}|\Y_{i}}   \, .  
\end{equation}
If we specialise this to the case of a $2D$ rectangular lattice of size $L$, the neighbourhood of generic (i.e. not belonging to the boundary) vertex becomes $\mathcal{N}_i = \lbrace i+1,i-1,i+L,i-L  \rbrace$. It follows that each $\Lambda^{ \Y_{\N_{i}}}_{\X_{i}|\Y_{i}} $ can be represented by a collection of $2^4 = 16$ matrices of size $2 \times 2$, which is an exponential improvement with respect to a general $2^{L^2} \times 2^{L^2}$. 

Although the above correlated noise model seems a very natural one, we will see in the following Section that it does not encompass all the correlated readout errors in current superconducting devices, for which it will be more convenient to resort to models \eqref{eq:noise_model_correlated} with more general cluster and neighbourhood structures that do not necessarily correspond to the physical layout of the devices.

\section{Characterization of readout noise}\label{sec:characterization}

Here we outline a strategy to determine a noise matrix in the form \eqref{eq:noise_model_correlated} which closely represents the readout noise of a given device. We proceed in two steps: at first we infer the structure of clusters ($\lbrace \mathcal{C}_{\chi}\rbrace $ and neighbourhoods  ($\lbrace{\mathcal{N}_{\chi}\rbrace}$ by making use of Diagonal Detector Tomography (DDT); then we proceed to experimentally determine noise matrices $\lbrace \Lambda^{ \Y_{\N_{\chi}}}_{\X_{\cluster_\chi}|\Y_{\cluster_\chi}} \rbrace$

For the first step, we propose to reconstruct all two-qubit noise matrices (averaged over all other qubits) by means of DDT and calculate the following quantities
\begin{align}\label{eq:correlations_pair}
c_{j\rightarrow i} = \frac{1}{2}||\Lambda_{Q_i}^{Y_j = '0'}-\Lambda_{Q_i}^{Y_j = '1'}||_{1\rightarrow1}\, 
\end{align}
where $||A||_{1\rightarrow1} \coloneqq \sup_{||v||_1=1} ||Av||_{1}= \max_j \sum_i |A_{ij}|$. The above quantity has a strong operational motivation in terms of Total-Variation Distance (TVD). 
This distance quantifies statistical distinguishability of probability distributions $\p$ and $\mathbf{q}$ and can be defined by 
\begin{align}\label{eq:TVD}
    \text{TVD}\rbracket{\p,\mathbf{q}} = \frac{1}{2} ||\p-\mathbf{q}||_{1} =  \frac{1}{2}  \sum_i |p_i-q_i| \ .
\end{align}

We can give the following, intuitive interpretation of the quantity from Eq.~\eqref{eq:correlations_pair}: $c_{j\rightarrow i}$ represents the maximal TVD for which the output probability distributions on qubit $Q_i$ differs due to the impact of the state of the qubit $Q_j$ on the readout noise on $Q_i$.
Note that in general $c_{j\rightarrow i}\neq c_{i\rightarrow j}$, which encapsulates the asymmetry in the correlations which is built into our noise model.

We propose to use the values of $c_{j\rightarrow i}$ to decide whether the qubits should belong to the same clusters, to the neighborhoods, or should be considered uncorrelated 
(a simple, intuitive algorithm for this procedure is presented in Appendix~\ref{sec:app:ddot_inference}., Algorithm~\ref{alg:clusters} --  in the future, we intend to extend those methods).
After doing so, the noise matrices $\lbrace \Lambda^{ \Y_{\N_{\chi}}}_{\X_{\cluster_\chi}|\Y_{\cluster_\chi}} \rbrace$ can be reconstructed by means of joint DDT over the sets of qubits $\lbrace \mathcal{C}_\chi \cup \mathcal{N}_\chi  \rbrace_\chi$. 

In the above construction we assumed that the joint size of a cluster and its neighborhood is at most $k$. This makes it so that one has to gather DDT data on subsystems of fixed size, implying a number of different circuits that scales at most as $O(N^k)$. However, for any characterization procedure, it is expedient to utilize as few resources as possible. In order to reduce the number of circuits even further, in the next Section we generalize the recently introduced Quantum Overlapping Tomography (QOT) \cite{Cotler2020quantum} (see also recent followups \cite{evans2019scalable,yu2020sample}) to the context of Diagonal Detector Tomography. We will refer to our method as Diagonal Detector Overlapping Tomography (DDOT).

\subsection{\\Diagonal Detector Overlapping Tomography}
Quantum Overlapping Tomography is a technique that was introduced for a problem of efficient estimation of all $k$-particle marginal quantum states.
The main result of Ref.~\cite{Cotler2020quantum} was to use the concept of hashing functions \cite{Majewski1996family,Stinson2000perfect,blackburn2000perfect,alon2008balanced} to reduce the number of circuits needed to reconstruct all $k$-qubit marginal states.
Specifically, it was shown there that  $O\rbracket{\log\rbracket{N}e^k}$ circuits suffice for this purpose.
Here we propose to use an analogous technique to estimate all noise matrices corresponding to $k$-particle subsets of qubits. 
Specifically, we propose to 
construct a collection of circuits consisting of certain combinations of $\iden$ and $X$ gates in order to initialize qubits in states $\ket{0}$ or $\ket{1}$. 
With fixed $k$, the collection of quantum circuits for DDOT must have the following property -- for each subset of qubits of size $k$, all of the computational-basis states on that subset must appear at least once in the whole collection of circuits. 
Intuitively, if a collection has this property, then the implementation of all circuits in the collection allows us to perform tomographic reconstruction (via standard DDT) of noise matrices on all $k$-qubit subsets.
One can think about DDOT as a method of parallelizing multiple local DDTs in order to minimize number of circuits needed to obtain description of all local $k$-qubit noise processes.
In Appendix~\ref{sec:app:ddot_perfect} we show that it suffices to implement $O\rbracket{k2^k\log\rbracket{N}}$ quantum circuits consisting of random combinations of X and identity gates in order to construct a DDOT circuits collection that allows to capture all $k$-qubit correlations in readout errors (see Algorithm~\ref{alg:random_generation} and Algorithm~\ref{alg:random_generation_circuits}).
It is an exponential improvement over standard technique of performing local DDTs separately (which, as mentioned above, requires $O\rbracket{N^k}$ circuits).
For example, if one chooses $k=5$ for $N=15$-qubit device, the naive estimation of all $5$-qubit marginals would require the implementation of $2^5 \binom{15}{5}\approx 10^5$ quantum circuits, while DDOT allows doing so using $\approx 350$ circuits.
We note that this efficiency, however, comes with a price. 
Namely, since different circuits are sampled with different frequencies, some false-positive correlations might appear. 
This may cause some correlations in the reconstructed noise model to be overestimated (see Appendix~\ref{sec:app:ddot_overestimation} for a detailed explanation of this effect).
This effect can be mitigated either by certain post-processing of experimental results (see Appendix~\ref{sec:app:ddot_overestimation}), or by constructing DDOT collections that sample each term the same number of times.
Using probabilistic arguments in Appendix~\ref{sec:app:ddot_balanced} we show that still the number of circuits exponential in $k$ and logarithmic in $N$ suffices if we want to have all $k$ particle subsets sampled with \emph{approximately} equal frequency.

\subsection{Experimental results} We implemented the above procedure with $k=5$ for IBM's 15q \textit{Melbourne} device and 23-qubit subset of Rigetti's \textit{Aspen-8} device (the details of the experiments are moved to Appendices).
The obtained correlation models are depicted in Fig.~\ref{fig:correlations_models}. 
In the case of Rigetti's device, our procedure reports a very complicated structure of multiple correlations in readout noise, while in the case of IBM's device the correlations are fairly simple.
We discuss this issue in detail in further sections while presenting results of noise-mitigation benchmarks.
Here we conclude by making an observation that, despite common intuition, the structure of the correlations in the readout noise can not be directly inferred from the physical layout of the device.

\begin{@twocolumnfalse}
\begin{figure*}
\begin{center}
\captionsetup[subfigure]{format=default,singlelinecheck=on,justification=RaggedRight}
\subfloat[IBM's \textit{Melbourne} device, 15 qubits.]
        {\includegraphics[width=0.95\textwidth]{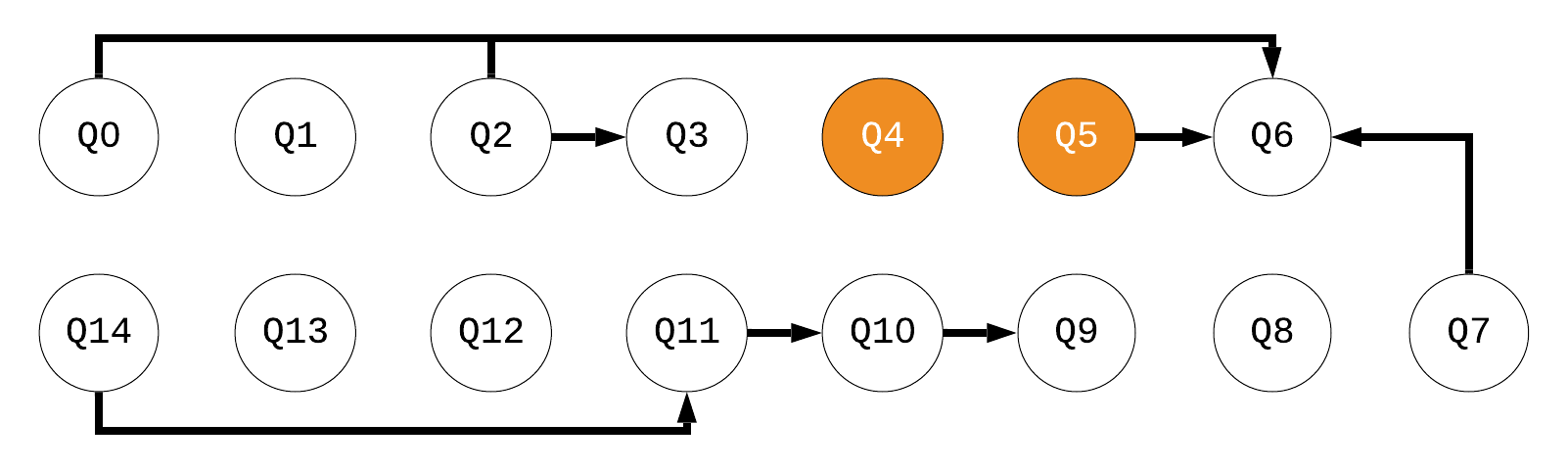}\label{fig:IBM_model15q}}\\
\subfloat[Rigetti's \textit{Aspen-8} device. 
Arrows indicate qubits which affect the measurement noise on the left half of the device.]
        {\includegraphics[width=0.95\textwidth]{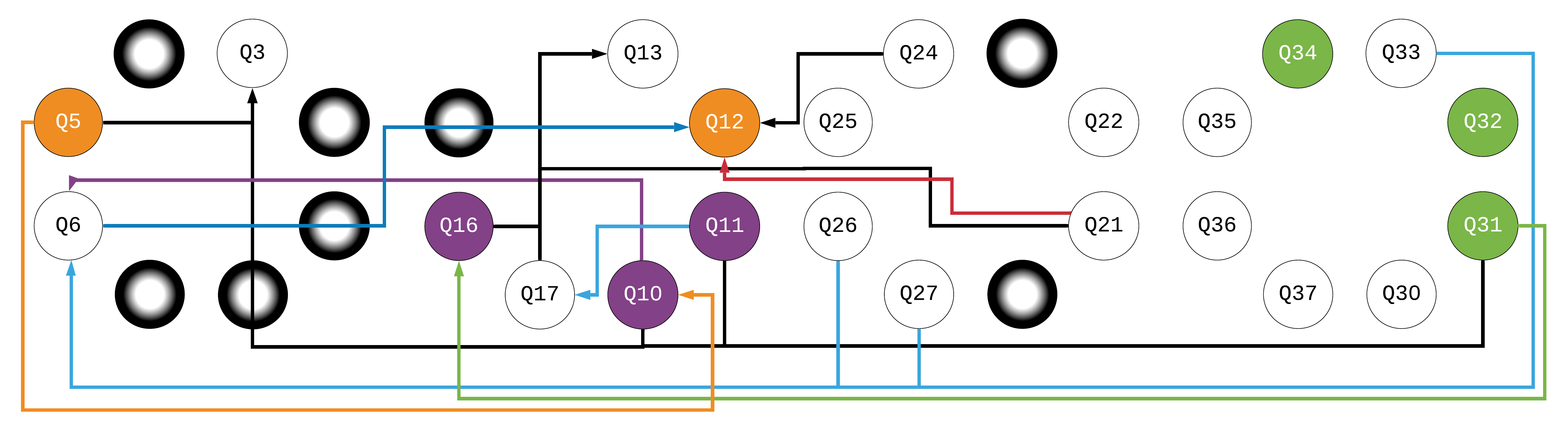}\label{fig:Rigetti_model_a}}\\
\subfloat[Rigetti's \textit{Aspen-8} device. 
Arrows indicate qubits which affect the measurement noise on the right half of the device.]
        {\includegraphics[width=0.95\textwidth]{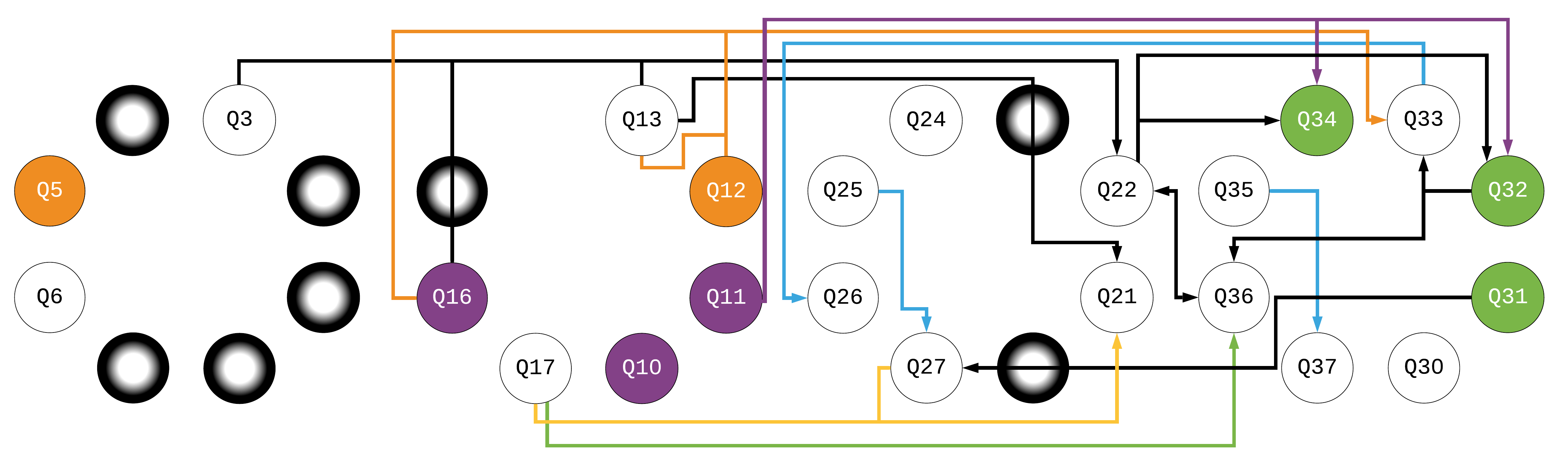}\label{fig:Rigetti_model_b}}
    \caption{ \label{fig:correlations_models} 
    Depiction of the correlation model obtained with Diagonal Detector Overlapping Tomogaphy on a) IBM's 15-qubit  \textit{Melbourne} device and b,c) a 23-qubit subset of Rigetti's \textit{Aspen-8} device.
    Due to the complicated structure of Rigetti's correlations, for clarity we divided the plots into two parts which show correlations on the left and right halves of the device (the merged plot can be found in the Appendix~\ref{sec:app:misc_experiments}).
    The $8$ black-and-white circles without label represent qubits which were not included in the characterization due to poor fidelity of the single-qubit gates (below $98\%$).
    The meaning of the rest of the symbols is described in the caption of Fig.~\ref{fig:exemplary_correlations4q}.  
  The colors of the lines connecting the neighbors on b,c) are provided such that the crossings of the lines are unambiguous (and have no other meaning otherwise).
    For the layout of the graphs, we used the qubits actual connectivity in the devices (i.e., it is possible to physically implement two-qubit entangling gates on all nearest-neighbors in the graph).
  For IBM's device, we included the qubits in the cluster if the correlations given by Eq.~\eqref{eq:correlations_pair} were higher than $4\%$ in any direction, while we marked qubits as neighbors if the correlations were higher than $1\%$. 
  In the case of Rigetti, the respective thresholds were chosen to be $6\%$ and $2\%$. 
  Moreover, for Rigetti we imposed locality constraints by forcing the joint size of the cluster and the neighborhood to be at most 5 by disregarding the smallest correlations.
  In practice the correlations within clusters were significantly higher than the chosen thresholds -- heatmaps of all correlations can be found in Appendix~\ref{sec:app:misc_experiments}.
}
\end{center}
\end{figure*}
\end{@twocolumnfalse}

\section{Noise mitigation on marginals}\label{sec:mitigation}

Now we will analyze the estimation and correction of marginal probability distributions affected by the kind of correlated noise model introduced in the previous section. 
In the next subsection, we will use those findings to propose a benchmark of the adopted noise model.

\subsection{Noise on marginal probability distributions}
Let us denote by $\p^{noisy}$ a global probability distribution generated by measurement of arbitrary quantum state on the noisy detector for which Eq.~\eqref{eq:noise_model_correlated} holds. 
As mentioned previously, for many interesting problems, such as QAOA or VQE algorithms, one is interested not in the estimation of $\p$ itself (which is an exponentially hard task), but instead in the estimation of multiple marginal probability distributions obtained from $\p$.
Let us say that we are interested in the marginal on a subset $\S$ formed by clusters $\cluster_\cindex$ indexed by a set $\mathcal{A}$, $\S\coloneqq \cup_{\cindex \in \mathcal{A} }{\cluster_{\cindex}}$,  where each $\cluster_{\cindex}$ is some cluster of qubits (see green envelope on Fig.~\ref{fig:example_marginal_subset} for illustration).
Our goal is to perform error mitigation on $\S$. 
To achieve this, we need to understand how our model of noise affects marginal distribution on $\S$.

From the definition of the noise model in Eq.~\eqref{eq:noise_model_correlated} we get that the marginal probability distribution $\p^{noisy}_\S$ on $\S$, is a function of the local noise matrices acting on the qubits from $\S$ and the ``joint neighborhood'' of $\S$, $\N\rbracket{\S} \coloneqq  \cup_{\cindex\in\mathcal{A}}\N_{\cindex} \setminus \S$. 
The set $\N\rbracket{\S}$ consists of qubits which are neighbors of points from $\S$ but are not in $\S$.

Because of this one can not simply use the standard mitigation strategy: i.e., estimate $\p^{noisy}_\S$ and reconstruct probability distribution $\p^{ideal}_\S$ by applying the inverse of $\Lambda$ (more discussion of this matter is given in Appendices).

To circumvent the above problem we propose to use the following natural ansatz for the construction of an \textit{approximate} effective noise model on the marginal $\S$
\begin{align}\label{eq:marginal_lambda_average}
 \Lambda^{\S}_{av}\coloneqq \frac{1}{2^{|\N\rbracket{\S}|}} \sum_{\Y_{\N\rbracket{\S}}} \Lambda^{\Y_{{\N(\S)}}} \ ,
\end{align}
where summation is over states of qubits in the joint neighborhood $\N(\S)$ defined above. 
In other words, it is a noise matrix averaged over all states of the neighbors of the clusters in $\S$, \emph{excluding} potential neighbors which themselves belong to the clusters in $\S$. 
Indeed, note that it might happen that a qubit from one cluster is a neighbor of a qubit from another cluster -- in that case, one does not average over it but includes it in a noise model.
Importantly, the average matrices $\Lambda^{\S}_{av}$ can be calculated explicitly using data obtained in the characterization of the readout noise.

\subsection{Approximate noise mitigation} The noise matrix  $\Lambda^{\S}_{av}$ can be used to  construct the corresponding \emph{effective correction matrix}
\begin{align}\label{eq:marginal_correction_average}
\corr^{\S}_{av} \coloneqq \rbracket{\Lambda^{\S}_{av}}^{-1}\ .
\end{align}	
Correcting the marginal distribution via left-multiplication by the above ansatz matrix is not perfect and can introduce error in the mitigation. 
In the following Proposition~\ref{lem:error_mitigation} we provide an upper bound on that error measured in TV distance.
\begin{proposition}\label{lem:error_mitigation}
Let $\p^{noisy}$ be a probability distribution on $N$ qubits obtained from the $N$ qubit probability distribution $\p^{ideal}$ via stochastic transformation $\Lambda$ of the form given in Eq.~\eqref{eq:noise_model_correlated}. 
Consider the subset of qubits $\S=\cup_{\cindex\in \mathcal{A}}{\cluster_{\cindex}}$. 
Let $\p^{\text{corr}}_\S=\corr^\S_{av} \p^{noisy}_\S$ be the result of the application to the marginal distribution $\p^{noisy}_\S$ of the correction procedure using the effective correction matrix $\corr^{\S}_{av}$ from Eq. \eqref{eq:marginal_correction_average}. 
We then have the following inequality 
\begin{align*}\label{eq:mitigation_error_bound}
&\text{TVD}\rbracket{\p^{\text{corr}}_{\S},\p^{\text{ideal}}_{\S}} \leq\\ &\qquad \quad\leq \frac{1}{2} ||\corr^{\S}_{av}||_{1\rightarrow1} \max_{\Y_{\N\rbracket{\S}}}||\Lambda^{\S}_{av}-\Lambda^{\Y_{\N\rbracket{\S}}}||_{1\rightarrow 1} \ , \numberthis
\end{align*}
where the maximization goes over all possible states of the neighbors of $\S$.
\end{proposition}

The proof of the above Proposition is given in Appendix~\ref{sec:app:marginals_proof1} -- it uses the convexity of the set of stochastic matrices, together with standard properties of matrix norms and with a triangle inequality (similar methods were used for providing error bounds on mitigated statistics in Ref.~\cite{Maciejewski2020mitigation}). 
Note that the quantity on RHS of Eq.~\eqref{eq:mitigation_error_bound} shows resemblance to $c_{i\rightarrow j}$ in  Eq.~\eqref{eq:correlations_pair} (which we used to quantify correlations). 
Hence $\frac{1}{2}\max_{\Y_{\N\rbracket{\S}}}||\Lambda^{\S}_{av}-\Lambda^{\Y_{\N\rbracket{\S}}}||_{1\rightarrow 1}$ can be interpreted as the maximal TVD between  states on $\S$ generated by $\Lambda^{\S}_{av}$ and states generated by $\Lambda^{\Y_{\N_{\cindex}}}$ (which appear in the description of the noise model). 
This can be also interpreted as a measure of dependence of noise between qubits in $\S$ and the state of their neighbors just before measurement. 
Indeed, if the true noise does not depend on the state of the neighbors, the RHS of inequality Eq.~\eqref{eq:mitigation_error_bound} yields $0$, and it grows when the noise matrices $\cbracket{\Lambda^{\Y_{\N\rbracket{\cindex}}}}$ increasingly differ.

\begin{figure}[!t]
\begin{center}
		\includegraphics[width=0.25\textwidth]{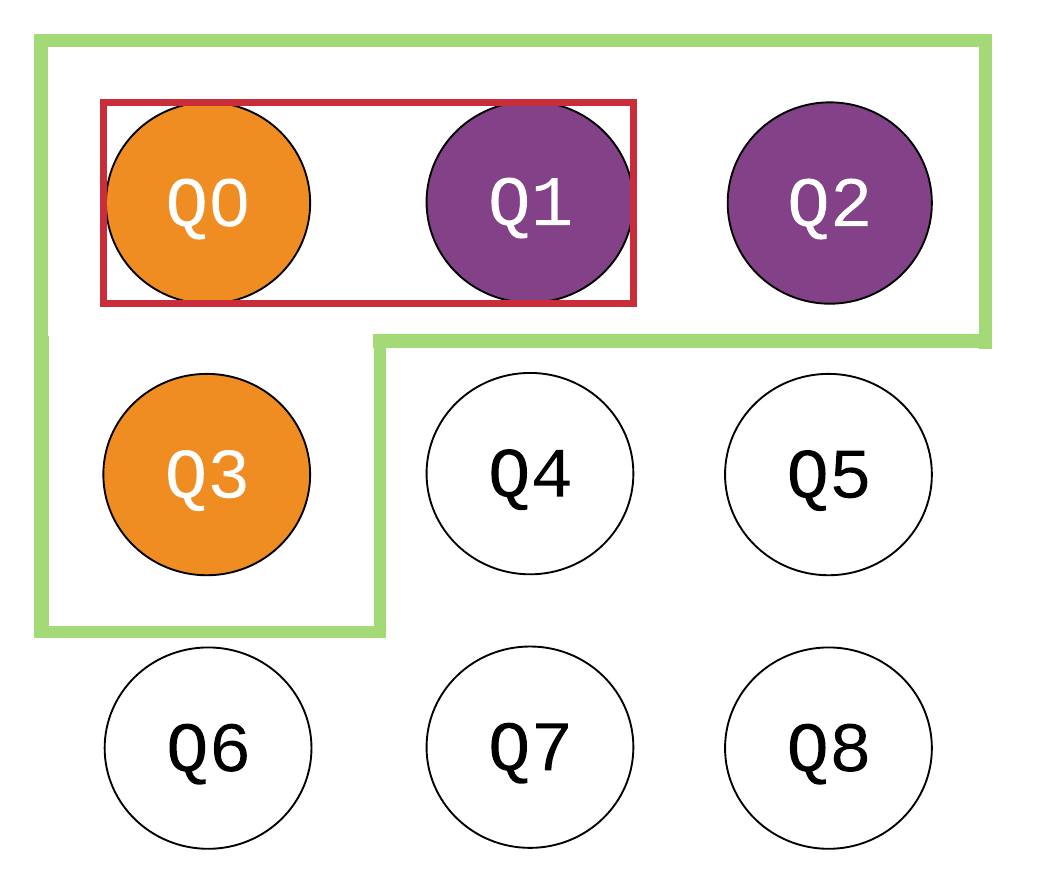}
		\caption{\label{fig:example_marginal_subset}
		Illustration of the cluster structure of an exemplary 9-qubit device. 
		There are two non-trivial clusters present ($\cbracket{0,3}$ and $\cbracket{1,2}$). 
		For clarity, no neighborhood dependencies are shown, though in reality noise on the clusters can be dependent on the qubits outside the clusters.
		When one measures all the qubits but the goal is to estimate four-qubit marginal distribution on $\mathcal{S} = \cbracket{0,3}\cup \cbracket{1,2}= \cbracket{0,1,2,3}$ (green envelope), noise-mitigation should be performed based on the noise model for the whole set of qubits $\S$.
		On the other hand, when the goal is to estimate the two-qubit marginal on qubits $\S_k=\cbracket{0,1}$ (red envelope), it is still preferable to first perform error-mitigation on the four-qubit marginal on qubits $\S=\cbracket{0,1,2,3}$, and then take marginals over qubits $2$ and $3$ to obtain corrected marginal on $\S_k$. See the description in the text.
}
	\end{center}
\end{figure}

In practice, it might happen that one is interested in the marginal distribution on the qubits from some subset $\S_k \subset \S$ (red envelope in Fig.~\ref{fig:example_marginal_subset}).
In principle, one could then consider a coarse-graining of noise-model \textit{within} $\S$ (i.e., construction of noise model averaged over qubits from $\S$ that do not belong to $\S_k$, treating those qubits like neighbors) and perform error-mitigation on the coarse-grained subset $\S_k$.
However, due to the high level of correlations within clusters, we expect such a strategy to work worse than performing error-mitigation on $\S$, and then taking marginal to $\S_k$. Indeed, we observed numerous times that the latter strategy works better in practice. 
Yet, it is also more costly (since, by definition, $\S$ is bigger than $\S_k$), hence in actual implementations with restricted resources one may also consider implementing a coarse-grained strategy.
In the following sections, we will focus on error-mitigation on the set $\S$.
All those considerations can be easily generalized to the case of $\S_k\subset \S$. 

\subsection{Sampling complexity of error-mitigation}
Let us now briefly comment on the sampling complexity of this error-mitigation scheme (the detailed discussion is postponed to Section~\ref{sec:covariances}).
If one is interested in estimating an expected value of local Hamiltonian, a standard strategy is to estimate the local marginals and calculate the expected values of local Hamiltonian terms on those marginals.
Without any error-mitigation, this has exponential sampling complexity in locality of marginals (which for local Hamiltonians is small), and logarithmic complexity in the number of local terms (hence, for typically considered Hamiltonians, also logarithmic in the number of qubits) -- see Eq.~\eqref{eq:statistical_bound} and its derivation in Appendix~\ref{sec:app:marginals_statistical}.
Now, if one adds to this picture error-mitigation \emph{on marginals}, this, under reasonable assumptions, does not significantly change the scaling of the sampling complexity.
We identify here two sources of sampling complexity increase (as compared to the non-mitigated marginal estimation).
First, the noise mitigation via inverse of noise matrix does propagate statistical deviations -- the bound on this quantitatively depends on the norm of the correction matrix (see Ref.~\cite{Maciejewski2020mitigation} and detailed discussion around Eq.~\eqref{eq:statistical_bound} in Section~\ref{sec:covariances}). 
Assuming that the local noise matrices are not singular (which is anyway required for error-mitigation to work), this increases sampling complexity by a constant factor .
Second, the additional errors can come from the fact that, as described above, sometimes it is desirable to perform noise-mitigation on higher-dimensional marginals (if some qubits are highly correlated). 
However, assuming that readout noise has bounded locality, this can increase a sampling complexity only by a constant factor (this factor is proportional to the increase of the marginal size as compared to estimation without error-mitigation). 
In both cases, for a fixed size of marginals (as is the case for local Hamiltonians), it does not change the \emph{scaling} of the sampling complexity with the number of qubits, which remains logarithmic.

\begin{@twocolumnfalse}
\begin{figure*}
\begin{center}
\includegraphics[width=0.98\textwidth]{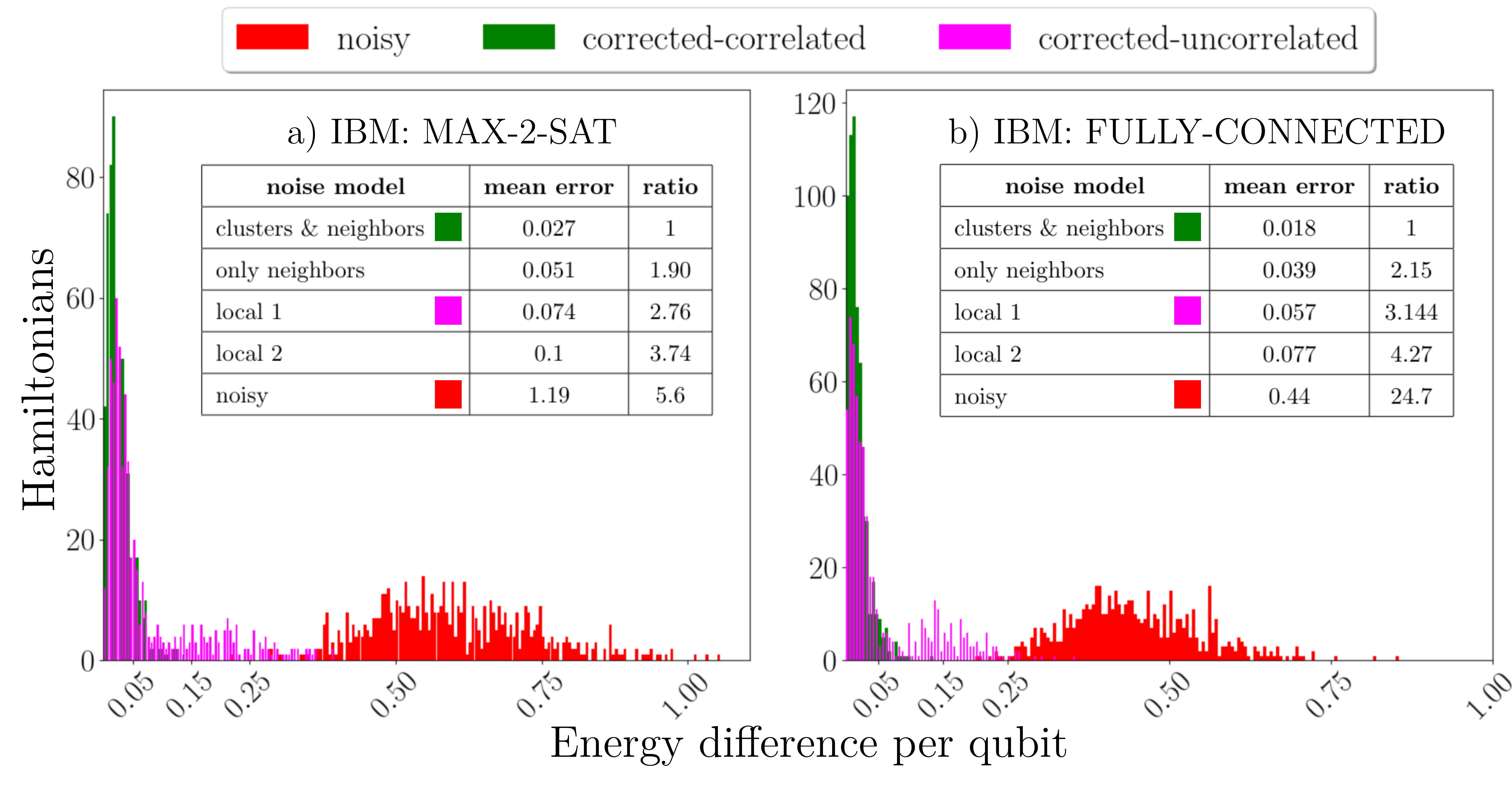}\\
\includegraphics[width=0.75\textwidth]{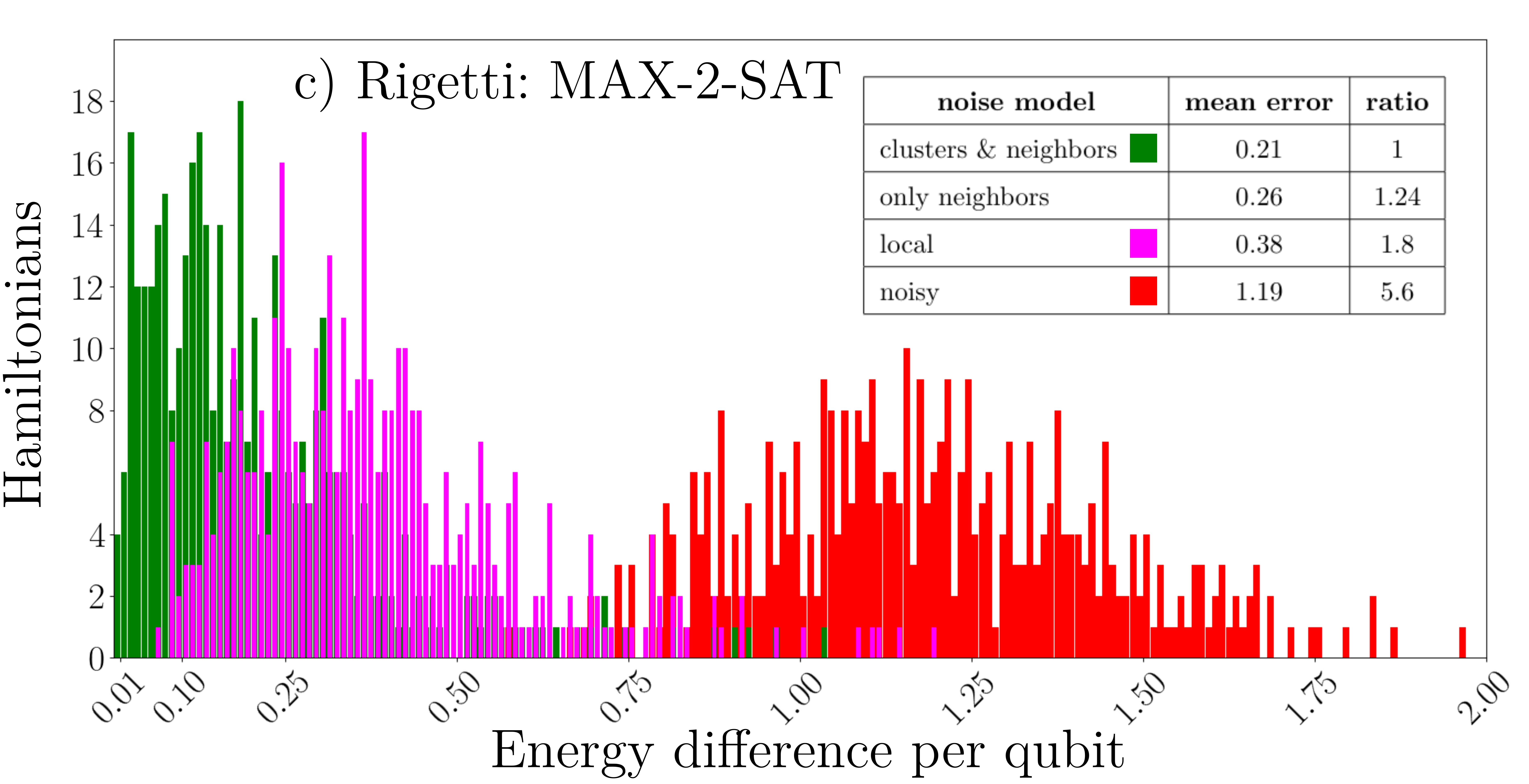}
\caption{Results of an experimental benchmark of the readout noise mitigation on (a-b) IBM's 15-qubit \textit{Melbourne} device, and (c) a subset of 23 qubits on Rigetti's \textit{Aspen-8} device. 
Each histogram shows data for 600 (IBM's) or 399 (Rigetti) different random Hamiltonians -- (a,c) random MAX-2-SAT and (b) fully-connected graph with random interactions and local fields. 
The horizontal axis shows the absolute energy difference (between estimated and theoretical) divided by the number of qubits.
The histogram comparison is done with no mitigation and with uncorrelated noise model characterization. 
The embedded tables show average errors depending on the adopted noise model. 
Here "ratio" refers to ratio of means.
Additional second row in each figure shows data for noise model with only trivial (single-qubit) clusters and their neighborhoods.
In case of IBM data, the additional fifth row illustrates memory effects. See description in the main text.} 
\label{fig:benchmark}
\end{center}
\end{figure*}

\end{@twocolumnfalse}

\section{Benchmark of the noise model}

\subsection{Energy estimation of local Hamiltonians}
After having characterized the noise model, how to assess whether it is accurate? 
To answer this question we propose the following, application-driven heuristic benchmark. 
The main idea is to test whether the error-mitigation of local marginals based on the adopted noise model is accurate. To check this we propose to to consider the problem of estimation of the expectation value $\tbracket{\H}$  of a local classical Hamiltonian 
	\begin{align}\label{eq:hamiltonian_local}
    \H = \sum_{\alpha}H_{\alpha}\ 
    \end{align} 
    measured on its ground state $\ket{\psi_{0}\rbracket{\H}}$. 
   Here by ``local'' we mean that the maximal number of qubits on which each $H_{\alpha}$ acts non-trivially does not scale with the system size. 
   Classicality of the Hamiltonian means that every $\H_\alpha$ is a linear combination of products of  $\sigma_z$ Pauli matrices. 
   In turn the ground state  $\ket{\psi_{0}\rbracket{\H}}$ can be chosen as classical i.e.,
    \begin{align}
        \ket{\psi_0\rbracket{\H}} = \ket{\X\rbracket{\H}} \ ,
    \end{align}
    for some bit-string $\X\rbracket{\H}$ representing one of the states from the computational basis. This problem is a natural candidate for error-mitigation benchmark due to at least three reasons.
    First, a variety of interesting optimization problems can be mapped to Ising-type Hamiltonians from Eq.~\eqref{eq:hamiltonian_local}.
    Indeed, this is the type of Hamiltonians appearing in the Quantum Approximate Optimization Algorithm. 
    The goal of the QAOA is to get as close as possible to the ground state $\ket{\psi_{0}\rbracket{\H}}$.
    Second, the estimation of $\langle \H \rangle  $ can be solved by the estimation of energy of local terms $\tbracket{H_{\alpha}}$ and therefore error-mitigation on marginals can be efficiently applied. Finally, the preparation of the classical ground state $\ket{\psi_0}$, once it is known, is very easy and requires only the application of local $\sigma_x$ (NOT) gates. This works in our favor because we want to extract the effects of the readout noise, and single-qubit gates are usually of high quality in existing devices.

    To perform the benchmark we propose to implement quantum circuits preparing ground states of many different local classical Hamiltonians, measure them on the noisy device, and perform two estimations of the energy - first from the raw data, and second with error-mitigation based on our characterization.    Naturally, it is also desirable to compare both with the error-mitigation based on a completely uncorrelated noise model (cf. Eq.~\eqref{eq:uncorrelated_noise}).
	We propose that if the mitigation based on a particular noise model works well on average (over some number of Hamiltonians), one can infer that the model is more accurate as well.

\subsection{Experimental results}	
We applied the benchmark strategy described above on the 15-qubit IBM's \textit{Melbourne} device and 23-qubit subset of Rigetti's \textit{Aspen-8} device.
We implemented Hamiltonians encoding random MAX-2-SAT instances with clause density 4 (600 on IBM and 399 on Rigetti) and fully-connected graphs with random interactions and local fields (600 on IBM).
MAX-2-SAT instances were generated by considering all possible sets of $4*15=60$ clauses with 8 variables, choosing one randomly and mapping it to Ising Hamiltonian (see, e.g., Ref.~\cite{Santra2014max}).
Random interactions and local fields for fully-connected graphs were chosen uniformly at random from $\sbracket{-1,1}$.
 Figure~\ref{fig:benchmark} presents the results of our experiments, together with a comparison with the uncorrelated noise model.
 
 Let us first analyze the results of experiments performed on a 15-qubit IBM's device.
Here it is clear that the error-mitigation based on our noise model performs well, often reducing errors in estimation by as much as one order of magnitude. 
 We further note that the uncorrelated noise model performs quite well (yet being visibly worse than ours). 
 To compare the accuracy by using a single number (as opposed to looking at the whole histogram), we take the ratio of the mean deviations from the ideal energy for the error-mitigated data based on two models.    
 The results are presented in tables embedded into Fig.~\ref{fig:benchmark}.
 In those tables, we provide also additional experimental data. 
 Namely, for each figure there the second row shows data for noise model labeled as "only neighbors". 
 This corresponds to noise model in which each qubit is a trivial, single-qubit cluster, and correlations are included only via neighborhoods.
 The worse results of error-mitigation for such model as compared to full clusters-neighborhoods model motivates the introduction of non-trivial clusters.
 Furthermore, we found experimentally that the characterization of the uncorrelated noise model exhibits significant memory effects (see, for example, Ref.~\cite{Rudinger2019probing}).
     Particularly, if one performs uncorrelated noise characterization in a standard way, i.e., by performing characterization in a separate job request to a provider, without any other preceding experiments (``local 2'' in tables), the accuracy (measured by the error in energy after mitigation based on a given noise model) is much lower than for the characterization with some other experiments performed \textit{prior} to the characterization of the uncorrelated model (``local 1'' in tables).
     Indeed, the difference in mean accuracy can be as big as $\approx 26\%$.

Clearly, the overall performance of Riggeti's 23 qubit device is lower than that for IBM's device.
 First, the effects of noise (measured in energy error per qubit) are stronger.
 Second, the mean error with error-mitigation is only around $\approx 5.6$ times smaller than the error without error-mitigation (as opposed to factor over $22$ for IBM's device).
 Third, the comparison to the uncorrelated noise model shows that the uncorrelated model performs not much worse than the correlated one.
 
Here we provide possible explanations of this poorer quality of experiments performed on Rigetti's device.
  Due to the limited availability of the Rigetti’s device, we used a much lower number of samples to estimate Hamiltonian's energies in those experiments. 
  Specifically, each energy estimator for Rigetti's experiments was calculated using only $1000$ samples, while for IBM's experiments the number of samples was $40960$. 
  This should lead to statistical errors higher by a factor of roughly $\sqrt{41}\approx 6.4$ (and note that the errors in error-mitigated energy estimation in Rigetti's device are around $7.8$ times higher than corresponding errors for IBM's device for the same class of Hamiltonians).
  Similarly, we used fewer measurements to perform DDOT -- on Rigetti's device, we implemented 504 DDOT circuits sampled $1000$ each, while on IBM's device we performed $749$ circuits sampled $8192$ each.
  Less DDOT circuits imply less balanced collection, hence, as already mentioned, some correlations might have been overestimated.
  In summary, our characterization of this device was in general less accurate than on IBM's device.
    This might be further amplified by the fact that single-qubit gates (which are used to implement DDOT circuits) were of lower quality for Rigetti's device.
 Finally, as illustrated in Fig.~\ref{fig:correlations_models}, we observed that correlations in measurement noise for Rigetti's device are much more complex than in the case of IBM's.
As mentioned in the Figure's description, to work around this we imposed locality constraints in the constructed noise model by disregarding the lowest correlations between qubits, which made the model less accurate.

  We note that due to the limited availability of Rigetti's device, we did not perform the study of memory effects similar to that performed in IBM. The local noise model presented for this device originates from a separate uncorrelated noise characterization performed prior to the rest of the experiments (hence it is analogous to the ``local 2'' model in IBM's case).

To summarize, presented results suggest that in experiments on near-term quantum devices it will be indispensable to account for cross-talk effects in measurement noise.
For both studied quantum devices we provided proof-of-principle experiments showing significant improvements in ground state energy estimation on the systems of sizes in the NISQ regime.
Motivated by those results, we hope that the framework developed in this work will prove useful in the future, more complex experiments on even larger systems.

\section{Error analysis for QAOA}\label{sec:covariances}
In this section, we analyze the magnitude of errors resulting from our noise-mitigation scheme when applied to an energy estimation problem. Those errors result as a combination of two independent sources. First, from the fact that we use approximate correction matrices instead of the exact ones (see in Proposition~\ref{lem:error_mitigation}). Secondly, by statistical errors due to the common practice of measuring \textit{multiple} marginals \textit{simultaneously} in a single run of the experiment. In the following, we will analyze the first and second effects separately and then provide a bound that takes both into account. We restrict our analysis to local Hamiltonians diagonal in the computational basis. A detailed derivation of the results below can be found in Appendices.

\subsection{QAOA overview} Before starting, let us provide a short overview of the QAOA algorithm. 
In standard implementation \cite{farhi2014quantum}, one initializes quantum system to be in $\ket{+}^{\otimes \noq}$ state, where $\ket{+}=\frac{1}{\sqrt{2}}\rbracket{\ket{0}+\ket{1}}$.
Then $p$-layer QAOA is performed via implementation of unitaries of the form
\begin{align}
U_p\rbracket{\mathbf{\alpha},\mathbf{\beta}} = \prod_{j}^p U_{\alpha_j}U_{\beta_j} , \ 
\end{align}
where $\mathbf{\alpha},\ \mathbf{\beta}$ are the angles to-be-optimized.
Unitary matrices are given by $U_{\alpha_j} \coloneqq \exp\rbracket{-i\ \alpha_j \H_{D}}$, and 
$U_{\beta_j} \coloneqq \exp\rbracket{-i\ \beta_j \H_{O}}$,
where $\H_{D}$ is driver Hamiltonian (which we take to be $\H_{D} = \sum_{k}^{\noq} \sigma_{x}^{k}$), and $\H_{O}$ is objective Hamiltonian that one wishes to optimize (i.e., to find approximation for its ground state energy). 
The quantum state after $p$-th layer is $\ket{\psi_p} = U_p \ket{+}^{\otimes \noq}$
and the function which is passed to classical optimizer is the estimator of the expected value $\bra{\psi_p}\H_{O}\ket{\psi_p}$ (note that this makes those estimators to effectively be a function of parameters $\cbracket{\alpha_j},\cbracket{\beta_j}$).
The estimator is obtained by sampling from the distribution defined by the measurement of $\ket{\psi_p}$ in the computational basis, taking the relevant marginals, and calculating the expected value of $\H_{O}$ using values of those estimated marginals.
Let us now proceed to the analysis of possible sources of errors while performing noise-mitigation on the level of marginals to estimate the energy of local Hamiltonians, such as those present in QAOA.

\subsection{Approximation errors} We start by recalling that performing noise mitigation with the average noise matrix instead of the exact one subjects the estimation of each marginal to an error upper bounded by Eq.~\eqref{eq:mitigation_error_bound}. 
It follows that the correction of multiple marginal distributions can lead to the accumulation of errors which for each marginal $\alpha$ (we label subset of qubits by $\alpha$ so that local term $H_{\alpha}$ acts non-trivially on qubits from $\alpha$) take the form
\begin{align}\label{eq:approx_error_definition}
        \delta^{\alpha} \coloneqq  \frac{1}{2} ||C^{\S_{\alpha}}_{av}||_{1\rightarrow1} \max_{\Y_{\N\rbracket{\S_{\alpha}}}}||\Lambda^{\S_{\alpha}}_{av}-\Lambda^{\Y_{\N\rbracket{\S_{\alpha}}}}||_{1\rightarrow 1}\,
\end{align}
where set $\S_{\alpha}=\cup_{\gamma\in\mathcal{A}}\cluster_{\gamma}$, $\mathcal{A}=\{\cindex\ |\ \cluster_\cindex  \cap \alpha \neq \emptyset\}$, consists of clusters to which qubits from $\alpha$ belong, and $C^{\S_{\alpha}}_{av}$ is the average correction matrix for the marginal on that set.	
It is straightforward to show that the total possible deviation between the error-mitigated expected value $\tbracket{\H^{\text{corr}}}$ and the noiseless one $\tbracket{\H}$ is upper bounded by
\begin{align*}
	    |\tbracket{\H^{\text{corr}}}-\tbracket{\H}| \leq 2\sum_{\alpha}\delta^{\alpha}||H_{\alpha}|| \\
	    \quad\rbracket{\text{additive approximation bound}}\ . \numberthis
\end{align*}

\subsection{Additive statistical bound} Moving to the effect of measuring several marginals simultaneously, let us start by considering the simplest bound on the propagation of statistical deviations under our error-mitigation.
In Appendix~\ref{sec:app:marginals} we derive that the TVD (Eq.~\eqref{eq:TVD}) between the estimator $\p^{\text{est}}_{\alpha}$ and the actual local marginal $\p_{\alpha}$ is upper bounded by
\begin{align*}\label{eq:statistical_bound}
    &\text{TVD}\rbracket{\p^{\text{est}}_{\alpha},\p_{\alpha}} 
    \leq\\ &\leq \epsilon^{*} \coloneqq \sqrt{\frac{\log\rbracket{2^n-2}+\log{\rbracket{\frac{1}{P_{ \text{err}}}}+\log{\rbracket{K}}}}{2s}} \ , \numberthis
\end{align*}
where $n$ is the number of of qubits in the support of each local term (for simplicity we assume it to be the same for all $H_{\alpha}$), $K$ is the total number of local terms, $s$ is the number of samples, and $1-P_{\text{err}}$ is the confidence with which the above bound is stated.
Importantly, the above bound is satisfied for each marginal \emph{simultaneously}, hence the logarithmic overhead $log\rbracket{K}$.
Using Eq.~\eqref{eq:statistical_bound} together with standard norm inequalities one obtains the following bound for the total energy estimation
\begin{align*}\label{eq:additive_bound}
        |{\H^{\text{est}}_{\text{corr}}}-\tbracket{\H}| \leq \sum_{\alpha} ||H^{\alpha}||\ ||\corr_{\alpha}||_{1\rightarrow1}\ \epsilon^{*} \\\quad \rbracket{\text{additive statistical bound}}.\qquad \numberthis
\end{align*}
Here ${\H^{\text{est}}_{\text{corr}}}$ denotes the estimator of the total energy with error-mitigation performed on each local term independently and $\corr_{\alpha}$ is the exact (not approximate)  correction matrix on marginal $\alpha$. 

\subsection{Joint approximation and statistical bound} 
The two bounds provided above took into account the two considered sources of errors independently.
By using the triangle inequality (see Appendix~\ref{sec:app:marginals_energy}), we can now combine them to obtain
\begin{align*}\label{eq:additive_bound_total}
  &|{\H^{\text{est}}_{\text{corr}}}-\tbracket{\H}| \leq \\
  &2\sum_{\alpha}||H^{\alpha}||\ \rbracket{\ \underbrace{\epsilon^{*}||\corr^{\S_{\alpha}}_{av}||_{1\rightarrow1}}_{\text{statistical errors}} +\underbrace{  \; \delta_{\alpha}}_{\text{approximation errors}} \ }\ . \numberthis
\end{align*}
It follows that the dominant scaling in the overall error are linear in the number of terms $K$ caused by summing over all of them and the logarithmic overhead in $\epsilon^{*}$ added by the statistical errors.

\subsection{Sampling complexity of energy estimation} While the additive bound from Eq.~\eqref{eq:additive_bound} could be tight \emph{in principle}, we observed numerically on many occasions that in practice the statistical errors are much smaller 
(see Fig.~\ref{fig:statistical_illustration} for exemplary results).

Here we will provide arguments that show that natural estimators of local energy terms $H_\alpha$ effectively behave as uncorrelated for a broad class of quantum states, hence leading to a significantly smaller total error than that obtained from an additive bound.

We start by describing in detail the natural strategy for energy  estimation in the considered scenarios. In this work we are concerned with classical local Hamiltonians. This means that all local terms $H_\alpha$ can be measured simultaneously via a single computational basis measurement. The natural estimation procedure amounts to repeating $s$ independent computational basis measurements on a quantum state $\rho$ of interests. Outcomes of these measurements are then used to  obtain local energy estimators $E^{\text{est}}_\alpha=\frac{1}{s}\sum_{i=1}^s E_\alpha^i$ , where $E^i_\alpha$ are values of the local energy terms obtained in the $i$'th experimental run. Now to perform estimation of expected value of energy, $\langle \H \rangle$, we simply sum the local estimators $E^{\text{est}}_\alpha$  
\begin{equation}
    \H^{\text{est}} =\sum_{\alpha}H_{\alpha}^{\text{est}} =\frac{1}{s} \sum_{i=1}^s \sum_{\alpha} E^i_\alpha\ .
\end{equation}
It is clear that $\H^{\text{est}}$ is an unbiased estimator of $\langle \H \rangle$. Likewise $E^{\text{est}}_\alpha$ are unbiased estimators of $\langle H_\alpha \rangle$.  

We would like to understand statistical properties of     $\H^{\text{est}}$ (specifically its variance) as as a function of number of experimental runs (samples) $s$ and the number of local terms in the Hamiltonian $K$. To this and we observe that random variables $E^i_\alpha, E^j_\beta$ are independent unless $i=j$ and therefore  
\begin{align}
    \mathrm{Var}\rbracket{\H^{\text{est}}}  = \frac{1}{s}  \sum_{\alpha,\beta} \mathrm{Cov}(E_\alpha^i,E_\beta^i)  \ , 
\end{align}
Assuming that measurements of the energy $E^i_\alpha$ are distributed according to the probability compatible with the Born rule allows us to write
\begin{equation}
    \mathrm{Cov}(E_\alpha^i,E_\beta^i)=\mathrm{Cov}\rbracket{H_{\alpha},H_{\beta}} = \tbracket{H_{\alpha}H_{\beta}}-\tbracket{H_{\alpha}}\tbracket{H_{\beta}}\ .
\end{equation}
Consequently we have
\begin{equation}\label{eq:varianceSUM}
    \mathrm{Var}\rbracket{\H^{\text{est}}}= \frac{1}{s} \mathrm{Var}(\H) = \sum_{\alpha,\beta} \mathrm{Cov}(H_\alpha,H_\beta) .
\end{equation}

The variance of $\H^{\text{est}}$ can be related to the  \emph{sample complexity} of the energy estimation.  Let $\Delta E>0$ be some positive number. Then using Chebyshev inequality we get
\begin{equation}
\text{Prob}\rbracket{|\H^{\mathrm{est}}-\tbracket{\H}|\geq \Delta E} \leq  \frac{\mathrm{Var}\rbracket{\H^{\mathrm{est}}}}{(\Delta E)^2}\ .
\end{equation}
Choosing some parameter $P_f$ as an upper bound bound on the RHS of this inequality, i.e., $\frac{\mathrm{Var}\rbracket{\H^{\mathrm{est}}}}{(\Delta E)^2}\leq P_f$.
Then we obtain using \eqref{eq:varianceSUM} that for the number of samples $s$ satisfying 
\begin{align}\label{eq:chebyshev_confidence}
 \frac{\mathrm{Var}\rbracket{\H}}{(\Delta E)^2 P_f} \leq s \ .
\end{align}
the estimator $\H^{\mathrm{est}}$ will be within accuracy form the expectation value $\langle \H \rangle$ with probability at least $1-P_f$. Now, if different Hamiltonian terms $\H_\alpha$ are correlated then according to the above bound the sample complexity grows like $K^2$, where $K$ is the total number of terms in  $\H$. Conversely, if $\mathrm{Cov}(H_\alpha,H_\beta)\approx0$ (for $\alpha\neq\beta$) then get sample complexity scaling linearly with $K$.

The above consideration can be equivalently translated to the estimates of the confidence intervals associated with estimator $\H^{\mathrm{est}}$ for a fixed value of samples $s$.  Specifically, if local terms are uncorrelated, then the confidence interval (statistical error) will scale as square-root $\sqrt{K}$ of the number of Hamiltonian terms, contrary to the pessimistic bound in Eq.~\eqref{eq:additive_bound_total} which is linear in $K$. We now show that one should expect that the sub-linear scaling of energy errors described above holds for a variety of quantum states, which in turn greatly reduces the sampling complexity compared to the pessimistic (linear) bound. We want to emphasize that our results are of immense practical importance for near-term devices. Our findings indicate a reduction (compared to the naive bounds)  of sample complexity of energy estimation by orders of magnitude, even for relatively small systems ($K\approx 100$  and larger).

\subsection{Generic 2-local Hamiltonians in QAOA}
We start by considering states that appear at the beginning and at the end of QAOA.
In recent work \cite{farhi2020quantum} it was shown that after $p$-th layer of QAOA optimization, for given two local terms $H_{\alpha}$ and $H_{\beta}$, there is no entanglement between qubits from $\alpha$ and qubits from $\beta$ if they are further away from each other than $2p$ (on a graph corresponding to interactions present in a Hamiltonian).
Therefore, for generic QAOA optimization, one can expect that for low $p$, i.e., at the beginning of the QAOA, the local Hamiltonian terms will be uncorrelated variables.
On the other hand, it is is well known that the ground states of local Hamiltonians are product states. 
Therefore, provided that QAOA worked well and converged to the state close to the ground state of objective Hamiltonian, the same arguments can be applied for the reversed QAOA circuit, and one expects local terms to be uncorrelated as well.

The following Proposition provides a more quantitative description for generic 2-local Hamiltonians corresponding to random graphs.
\begin{proposition}\label{lem:random_graphs}
Consider Hamiltonian with connectivity given by Erd\"os-R\'enyi random graph in which each edge of the graph is added independently at random with some fixed probability.
Assume that the probability of adding edge is chosen so the average degree of a node is equal to $q=\frac{K}{N}$, hence that a random graph has on average $N$ nodes and $K$ edges.
For QAOA starting from product state, if the number of layers satisfies
\begin{align}\label{eq:p_regime}
    p <   \frac{  w\log(N)}{ {8} \log(2q/ \ln(2))} - 1
\end{align}
with $w<1$, then with probability $1 - e^{-N^{a/2}}$ the variance of the Hamiltonian is bounded by 
 \begin{align}\label{eq:variances_generic_bound}
    \mathrm{Var}\rbracket{\H} \le f_{\H} \ q\ N^{{A}+1}
\end{align}
where
\begin{align*}
 f_{\H} &= \max_{\alpha,\beta} ||H_{\alpha}||\ ||H_{\beta}||  \ ,\\
   A &= w \, \frac{(2 + |\log_{2q}(\ln(2))|) }{(1+ |\log_{2q}(\ln(2))|)}  \ , \\
   a &= \frac{w}{3(1 + |\log_{2q}(\ln(2))|}  \ , \numberthis
\end{align*}
where maximization in first definition goes over all two-qubit local terms acting on subsets of qubits $\alpha$ and $\beta$. 
Since we can always choose $w <1 $ such that $A < 1$, the variance thus scales sub-quadratically for shallow depth QAOA. 
Importantly,  the parameter $w$  can be chosen in such a way that $A$ goes asymptotically to $0$ (an exemplary choice is $w=O(\log(N)^{-\frac{1}{2}})$, which in turn means that in the $N$ regime the variance $\mathrm{Var}(\H)$ will scale almost linearly with the number of terms.

\end{proposition}
The proof of the above Proposition~\ref{lem:random_graphs} uses insights from Ref.~\cite{farhi2020quantum} and is delegated to Appendix~\ref{sec:app:sampling_proof2}.

\subsection{Random quantum states}
A simple argument can be made to show that for generic Haar-random pure sates, as well as random states appearing in random local quantum circuits, the variance of a local Hamiltonian  $\H$ behaves as if different energy terms were independent. Let $\ket{\psi}$ be a pure state on $(\mathbb{C}^2)^{\otimes N}$.  For a subset of qubits $\gamma$,  let $\rho_\gamma$ denote the reduced density matrix $\ketbra{\psi}{\psi}$ on qubits in $\gamma$, and let $\id_\gamma$ denote the normalized maximally mixed state on qubits in $\gamma$.  Assume now that local Hamiltonian terms $H_\alpha,H_\beta$ have disjoint supports. It can be then shown (see Proposition \ref{prop:covMIXEDBOUND} in Appendices for the proof) that 
 \begin{equation}
     \mathrm{Cov}(H_\alpha,H_\beta) \leq 3 \|H_\alpha\|\|H_\beta\|\,  \|\rho_{\alpha\cup\beta}-\id_{\alpha\cup \beta} \|_1\ .
 \end{equation}
Now, it is a well-known fact \cite{entFOUNDstatMech} that with overwhelming probability \emph{all} few-body marginals $\rho_\gamma$ of a Haar-random multiqubit states $\ket{\psi}$ are exponentially close to maximally mixed states. Therefore, assuming that $\max_{\alpha,\beta} $ does not scale with the system size, we have that high probability over the choice of $\ket{\psi}$, for every disjoint terms $H_\alpha,H_\beta$ in a local Hamiltonian $\H$ 
\begin{equation}
    \mathrm{Cov}(H_\alpha,H_\beta) \approx 0\ . 
\end{equation}
The above reasoning mimics the computation done in Theorem 1 of \cite{OszmaniecMetro}, where it was used to establish that generic Haar-random pure states attain only the so-called standard quantum limit in the paradigmatic interferometric scenarios (again the underlying argument was based on the fact that \emph{all} few-body reduced density matrices of a generic pure state $\ket{\psi}$ are very close to maximally mixed states with overwhelming probability).

Analogous analysis can be carried out for typical states generated by local random quantum circuits. Such circuits are known to form approximate $t$-designs, i.e., capture properties of typical Haar-random unitaries captured by low-degree moments  \cite{BHH2016}. Specifically,  a recent paper \cite{Nick2020} considered evolution of local entropies for pure states $\ket{\psi}$ generated by shallow local random quantum circuits. From Theorem 1 of that work it directly follows that with probability greater than $1-\delta$ over the choice random states $\ket{\psi}$ generated by random local circuits in the brick-wall architecture of depth $r$, all marginals $\rho_\gamma$ of size $|\gamma|=k$,  satisfy $\|\rho_\gamma-\id_\gamma\|_1 \leq \ep$, where 
 \begin{equation}
     \delta \leq  \frac{\binom{N}{k}^2 }{\ep^2} \left( 2^{2k-N} +2^k\left(\frac{4}{5}\right)^{2(r-1)} \right)\ . 
\end{equation}
Clearly, if the size of the marginals $k$ is fixed, setting  $r=c \log(N/ \ep)$, for a suitable constant $c$, allows us to conclude that for all $\gamma$ such that $|\gamma|\leq k$ one has $\|\rho_\gamma -\id_\gamma\|_1 \leq \ep $ with probability approaching $1$ with the increasing system size.

\begin{center}
 \begin{figure}[!t]
       \includegraphics[width=0.495\textwidth]{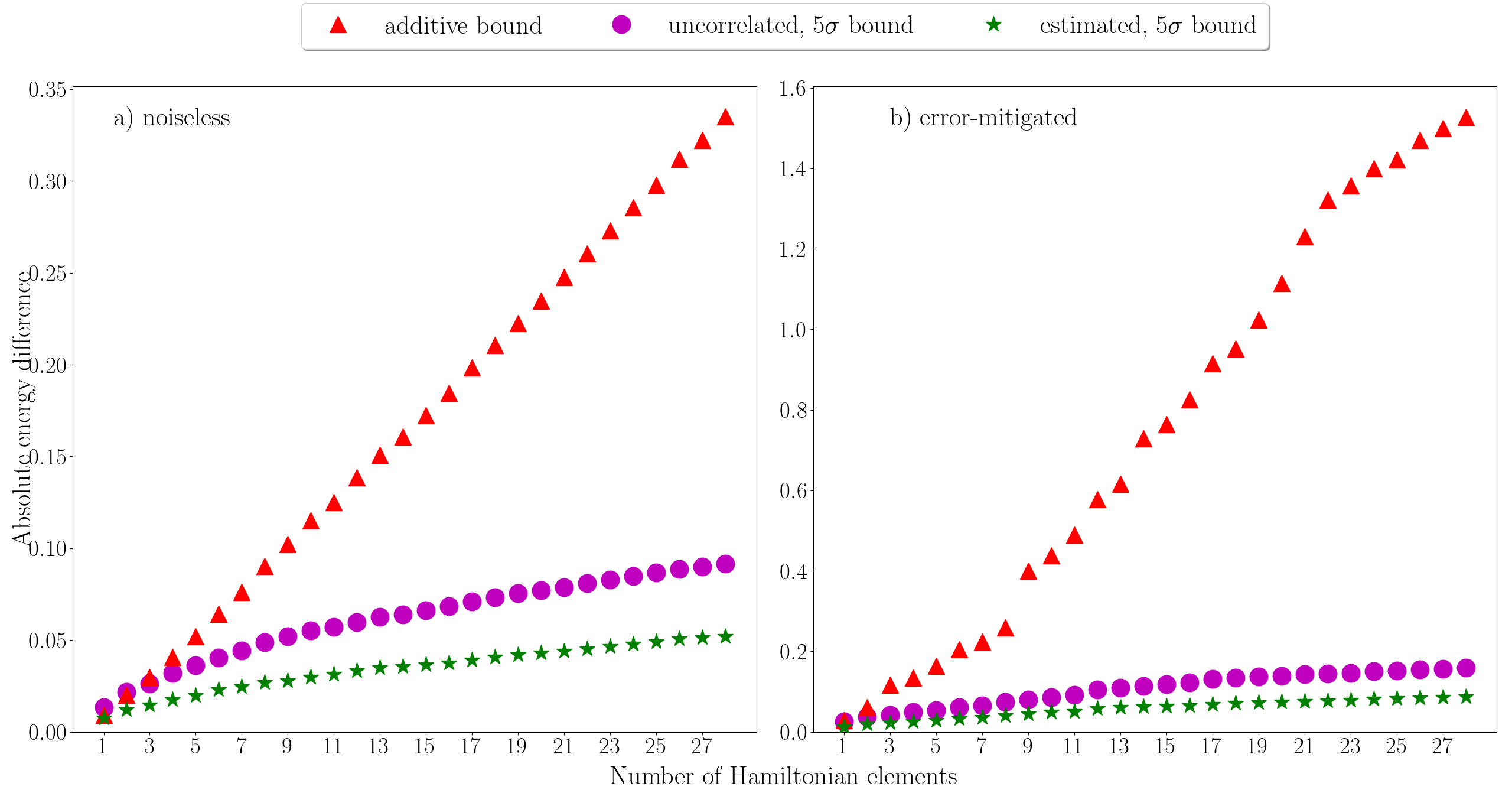}
\caption{\label{fig:statistical_illustration}Numerical results illustrating low statistical errors in the estimation of the energy of local Hamiltonians. 
 We show the results for the energy of 8-qubit random Sherrington-Kirkpatrick Hamiltonians, calculated on quantum states obtained \emph{after the third layer} of QAOA optimization.
 The horizontal axis shows the number of Hamiltonian elements (in a sense that for a given Hamiltonian one estimates only some local terms).
 The Y-axis shows the absolute difference energy error.
 Each data point is an average over $92$ Hamiltonians.
 Red triangles correspond to additive error bound from Eq.~\eqref{eq:additive_bound_total} (with correction matrix norm set to 1 and approximation errors to $0$ for noiseless scenario).
 Magenta circles show the confidence intervals on energy estimation corresponding to $5\sigma$-confidence (via Chebyshev's bound) assuming \emph{uncorrelated} variables for the sample size $s\approx 10^4$ (see Eq.~\eqref{eq:chebyshev_confidence} and description in the text).
 Green stars are analogous confidence intervals \emph{estimated} in numerical simulations (for each Hamiltonian empirical standard deviations were obtained by repeating simulation 100 times).
 Plot a) shows results for the noiseless scenario, while b) shows results for the error-mitigated scenario with noise mitigation performed on marginals and noise model inspired by IBM's device's characterization. 
 }
\end{figure} 
\end{center}

\subsection{Effects of measurement noise}
To conclude, let us provide some analysis of the effects of measurement noise mitigation on the above considerations.
First, let us note that it is straightforward to generalize all of the above arguments to include the \textit{uncorrelated} readout noise. 
Intuitively, if the measurement noise is not correlated, it cannot increase the level of correlations in the energy estimators, and the same holds for noise-mitigation for such model (see Appendix~\ref{sec:app:sampling_covariances} for derivations).

For the correlated noise model, it is hard to provide analytical results, however, one can expect that if correlations in measurement noise are mild, then it should not drastically increase sample complexity.
To test this hypothesis for small system sizes, we performed numerical simulations in the following way.
Consider confidence intervals of the energy estimation.
The most pessimistic bound on the error, as already explained, is given by Eq.~\eqref{eq:additive_bound_total} and is additive in the number of Hamiltonian terms $K$ (if one considers the situation without measurement noise, then it suffices to set $\delta_{\alpha}=0$ and $||C^{\S_{\alpha}}_{av}||_{1\rightarrow 1}=1$ and  Eq.~\eqref{eq:additive_bound_total} still holds).
To obtain confidence intervals expected for uncorrelated variables, we simply set $\mathrm{Var}\rbracket{\H} = \sum_{\alpha}\mathrm{Var}\rbracket{H_{\alpha}}$, and provide corresponding confidence interval by calculating LHS of Eq.~\eqref{eq:chebyshev_confidence}.
To test whether resulting bounds are close to what happens in practice, we numerically \emph{estimate} the variance of $\H$ and calculate the resulting confidence intervals.
Such estimation of variance can be done, for example, by performing multiple numerical experiments, each giving an empirical estimate of $\mathrm{Var}\rbracket{\H}$ (for a fixed number of samples $s$), and taking the mean of those estimates.
Such comparison is plotted on Fig.~\ref{fig:statistical_illustration} for the system of $N=8$ qubits and Hamiltonians estimated on the states coming from $3$-layer QAOA.
Bound $P_f$ on the probability (of energy estimator being outside the calculated confidence interval) in Chebyshev's inequality is set to $P_f=0.04$ which corresponds to $5\sigma$-confidence.
Shown are results for the noiseless scenario, and for the error-mitigated estimators with the noise model inspired by IBM's device characterization.
It is clear that for tested Hamiltonians the confidence intervals in the estimation behave roughly like for uncorrelated variables.

\vspace{-0.175cm}
\section{Effects of measurement noise on QAOA --  numerical study}\label{sec:QAOA}
\vspace{-0.175cm}
In this section we apply our noise characterization and mitigation strategies to numerically study the effects of correlated measurement errors on QAOA and how they can be reduced with our techniques.
As test Hamiltonians, we choose those encoding random MAX-2SAT instances with clause density 4, the Hamiltonians corresponding to fully-connected graphs with random interactions and local fields with magnitude from $\sbracket{-1,1}$, and the Sherrington-Kirkpatrick (SK) model in $2$D (i.e. random Gaussian ZZ-interactions on a square lattice).

\begin{@twocolumnfalse}
\begin{figure*}
\begin{center}
\centering
\includegraphics[width=0.98\textwidth]{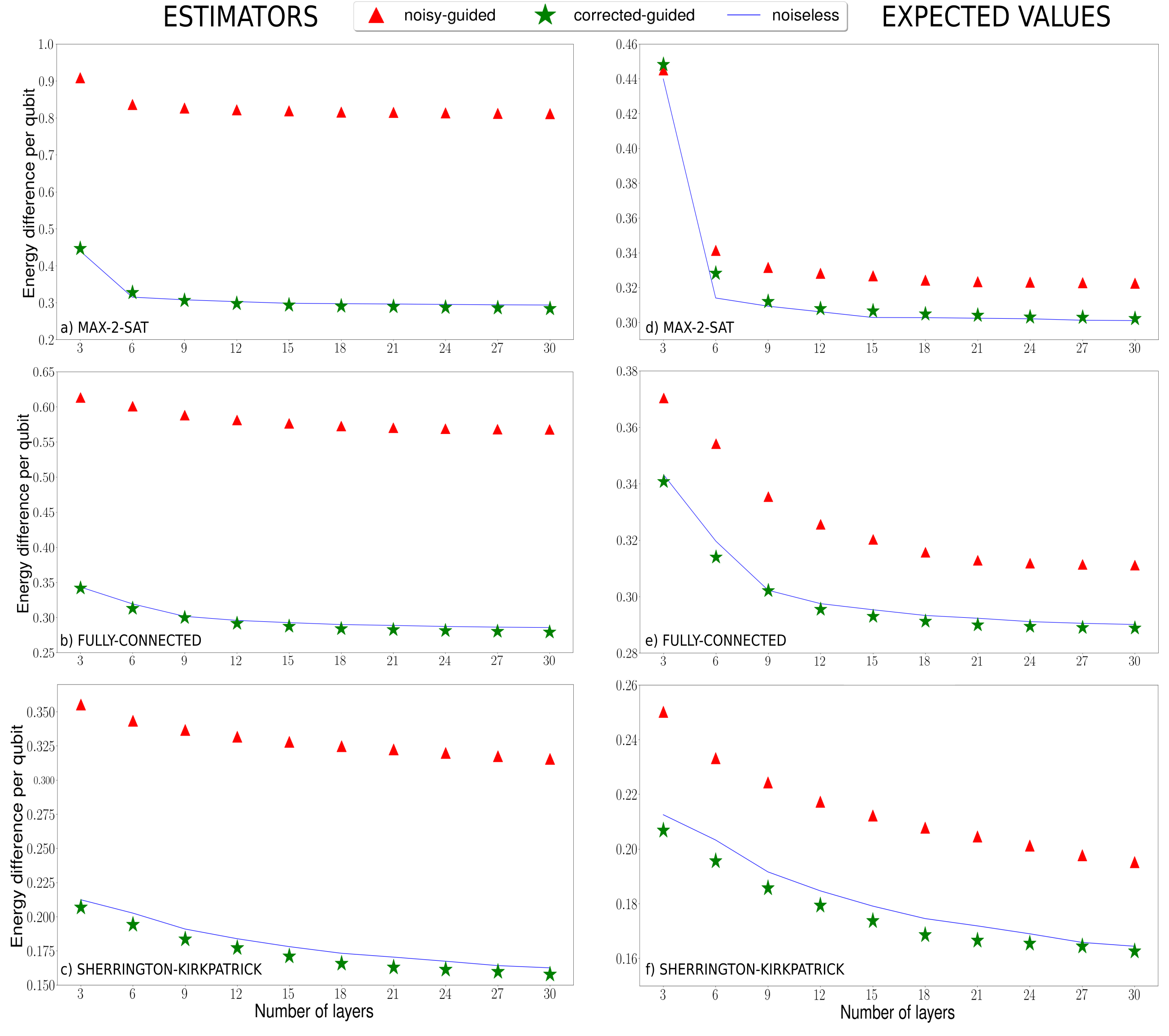}
 \captionof{figure}{
Numerical study of the effects of readout noise on energy estimation for QAOA on an 8-qubit system. 
The algorithm is used to prepare the ground state of a,d) Hamiltonians encoding random MAX 2-SAT problems, b,e) Hamiltonian corresponding to a fully-connected graph with uniformly random (from $\sbracket{-1,1}$) interactions and local fields, and c,f) the $2$D Sherrington-Kirkpatrick model.
In each plot, the horizontal axis shows a number of layers in the QAOA optimization, while the vertical axis shows the absolute difference between obtained and theoretical energies per qubit.
Each data point is the \emph{average} over 96 Hamiltonians.
Note that we made an offset on the y-axis in order to make differences visible.
Red data-points indicate optimization guided by noisy function evaluations, green points indicate noise-mitigated evaluations, while blue lines correspond to noiseless optimization given for reference.
The estimators were obtained from $\approx 10^4$ samples.
Due to differences in spectra of various Hamiltonians over which the mean is calculated the fluctuations around the presented means are very high, therefore for clarity we decided not to include error bars.}
\label{fig:QAOA_estimators}
\end{center}
\end{figure*}
\end{@twocolumnfalse}

\noindent For all models, we classically simulate a QAOA algorithm on an $8$-qubit device with a number of layers ranging from $p = 3$ to $30$.
The parameter optimization is performed using Simultaneous perturbation stochastic approximation (SPSA) \cite{Spall1998ANOO,cade2019strategies,montanaro2020compressed,Kandala2017}  (see Appendix~\ref{sec:app:misc_QAOA} for details of optimization).
As a correlated noise model, we adopt one inspired by our previous characterization of the IBM 15-qubit device. 

We first simulated the result obtained by a QAOA algorithm where both the energy estimation and the gradient estimation used to guide the state evolution were affected by readout noise. 
For different numbers of gate layers, we compared the resulting energy estimators between an optimization guided by noisy estimators (``noisy-guided'' in Fig.~\ref{fig:QAOA_estimators}) and the optimization guided by estimators on which error-mitigation was performed (``corrected-guided'' in Fig.~\ref{fig:QAOA_estimators}).
 The results of our numerical studies are presented in the first column of Fig.~\ref{fig:QAOA_estimators}, together with the noiseless case as a reference.
 Note that to make differences more visible, we set offset on the vertical axis.
It is clear that for the considered Hamiltonians, the noise-mitigated estimators are much better than noisy ones.
This suggests that our noise-mitigation scheme can be used to obtain an overall more reliable QAOA algorithm.

For our second analysis, we wanted to closely analyze the effects of noise (and its mitigation) on the parameter optimization only. 
To do so, we still compare the results of QAOA between the cases where the optimization is guided by a noiseless, noisy (``noisy-guided'' in Fig.~\ref{fig:QAOA_estimators}) and noise-mitigated (``corrected-guided'' in Fig.~\ref{fig:QAOA_estimators}) energy estimators. 
However, in order to isolate the effect of noise on the optimization procedure, instead of sampling the energy from the noisy probability distributions, we calculated its expectation value directly on the quantum state obtained at each layer of the QAOA circuit (the same circuits as in left column).
As a theoretical comparison between the optimal parameters found by QAOA, we took the distance between the resulting average energies on the state, divided by the number of qubits.
The plots in the second column of Fig.~\ref{fig:QAOA_estimators} show our numerical results. 
For both of the considered models, the noise-mitigated optimization leads to better parameters regions.
We note that the difference is not high, yet it is systematic. 
Since this is the case already for $8$ qubits, one can expect that for bigger systems, the relative improvement will be higher.

To conclude, let us stress that the purpose of this section was to illustrate the possible effects of mildly correlated, realistic readout noise on QAOA algorithm. 
The above results clearly indicate that correlated readout noise can potentially influence the optimization, as opposed to an identical and uncorrelated one (see Ref.~\cite{xue2019effects}). 
We find that in a number of instances the error-mitigation strategy helps to land in better parameters regions, however, we also note that the effects of noise on the optimization are not dramatic in the studied cases.
From the point of view of near-term applications, this should be viewed as a positive result -- the noise does not seem to strongly affect optimization, yet even those mild effects can be reduced by performing our error-mitigation.
Furthermore, we note that due to the stochastic nature of the SPSA optimizer, the results might differ from run to run (each data point presented in the plots comes from the results of the optimization run which was the better one amongst two performed independent optimizations, see Appendix~\ref{sec:app:misc_QAOA} for details).
Finally, we note that to obtain accurate estimates of energy, the noise-mitigation, unsurprisingly, remains highly beneficial.

\section{Discussion}\label{sec:conclusions}
\subsection{Conclusions} 

In the first part of this work, we proposed an efficiently describable model of correlated measurement noise in quantum detectors.
The basic idea of the model is to group qubits into strongly-correlated clusters that are mildly affected by their neighborhoods, which, provided that the size of those groups is bounded by a constant, allows to describe a global noise model by much smaller number of parameters compared to the most generic situation.
To characterize our noise model, we have introduced Diagonal Detector Overlapping Tomography, which is a procedure inspired by recently introduced Quantum Overlapping Tomography \cite{Cotler2020quantum}, tailored to the efficient characterization of the proposed measurement noise model.
Similarly to the \cite{Cotler2020quantum}, proposed method can estimate $k$-qubit correlations in measurement noise affecting $N$-qubit device using $O\rbracket{k2^k\log\rbracket{N}}$ quantum circuits.
We have shown that the measurement noise can be efficiently mitigated in problems that require estimation of multiple marginal probability distributions, an example of which is Quantum Approximate Optimization Algorithm \cite{farhi2014quantum}. 
Importantly, from the fact that noise-mitigation is performed on marginal distributions it follows that sampling-complexity of noise mitigation is similar to that of the original problem, provided that cross-talk in readout noise is of bounded locality.
We proposed a benchmark of the noise model and error reduction, which we implemented in experiments on up to 15 qubits on IBM's \textit{Melbourne} device and on 23 qubits on Rigetti's \textit{Aspen-8} device, and concluded significant improvements compared to simple, uncorrelated noise model.
Interestingly, additional experimental data have pointed at previously unreported memory effects in IBM's device that were demonstrated to not-negligibly change the results of experiments.

In the second part of the work, we provided an analysis of the statistical errors one may expect when performing the simultaneous estimation of multiple local terms of Hamiltonian on various classes of states.
We provided simple arguments why low sampling complexity (i.e., scaling of the variance as the square-root of a number of local Hamiltonian terms) should be expected from Haar-random quantum states, and for states generated by shallow random quantum circuits.
Similarly, we gave some arguments based on \cite{farhi2020quantum}, why for states appearing at the beginning and at the end of the QAOA, one may expect that estimated local terms will effectively behave as uncorrelated variables, reducing the sampling complexity.
Furthermore, we have provided analytical results for Hamiltonians encoding random MAX-2-SAT instances.

In the last part of the work, we have presented numerical results and extended discussion of the effects that correlated measurement noise can have on the performance of QAOA.
We have demonstrated that already for $8$ qubits the correlated measurement noise can alter the energy landscape in such a way, that the quantum-classical optimization leads to sub-optimal energy regions (compared to reference runs without noise). 
At the same time, we demonstrated that our noise-mitigation procedure can reduce those effects, improving optimization.

\subsection{Future research} We believe that our findings will prove useful in both near- and long-term applications as efficient methods of characterizing and reducing measurement noise.
At the same time, we find that a number of future research directions opens.
It is natural to ask how well the error mitigation will perform in actual multi-qubit experiments of QAOA, which can be tested only with more access to the quantum devices than is available for the public via cloud services.
Similarly, testing noise mitigation in more general, Variational Quantum Eigensolver (VQE) scenarios is of great interest \cite{huggins2020efficient,barron2020measurement}.
The analysis of statistical errors in the VQE setup is a natural extension -- we note that since the local terms in VQE Hamiltonians do not commute, it is less straightforward than in the QAOA scenario.
Another interesting problem is to design more benchmarks for the noise model and mitigation than the one method described in this work.
A natural extension of techniques presented in this work would be to develop more methods for inference of correlations structure and construction of noise model from DDOT data. 
Having a noise model for marginal probability distributions, it is desirable to test and compare techniques of noise mitigiation that go beyond a standard noise matrix inversion analyzed in this work.
In particular, methods based on Bayesian inference seem promising \cite{Nachman2019unfolding}.
Moreover, it is tempting to ask whether error-mitigation performance and demonstrated memory effects are stable over time, a question which is often omitted (see however recent work \cite{Proctor2020detecting} where systematic methods for studying time instabilities were developed and \cite{dasgupta2020characterizing} which made an important contribution by studying the stability of various types of noise over time for IBM's devices).
Another open problem is to find out whether further generalizations of Quantum Overlapping Tomography can be used to perform reliable characterization of more general types of noise -- coherent measurement noise and the noise affecting quantum gates.
For example, in the work \cite{Govia2020bootstraping} the authors consider an ansatz for the generic noisy quantum process which uses only 2-local noise processes -- a setup which seems natural to benefit from similar techniques.
In Refs.~\cite{Flammia2020efficient,Harper2020efficient} the authors develop methods of estimating generic Pauli channels, and it would be of great interest to assess whether those methods can benefit from using measurement error-mitigation techniques such as ours.
A very important research problem is that of mitigating measurement noise in scenarios which include estimators obtained from a very few samples, such as those in Refs.~\cite{Huang2020predicting,chen2020robust}.
We note that in its current form our methods cannot be directly implemented in such scenarios because they operate on marginal probability distributions.
Finally, it would be extremely interesting to investigate whether our model of measurement noise is accurate for quantum devices based on architectures different than transmon qubits, such as flux qubits \cite{boothby2020nextgeneration}, trapped ions \cite{bruzewicz2019trapped} or photonic quantum devices \cite{Wang2019integrated}. 
We intend to investigate some of the listed problems in future works.

\vspace{0.5cm}
\noindent\textbf{Data availability}\\
The Python code for creating DDOT perfect collections and for characterization of readout noise (in qiskit \cite{Qiskit_ref_proper} and pyquill) will be available in our open-source Quantum Readout-Error Mitigation (QREM) GitHub repository \cite{qrem}.
We intend to further develop the package to contain practically relevant techniques of noise characterization and mitigation, including modifications and refinements of algorithms presented in this work.
The results of experiments used in this work will be available as exemplary data-sets in the QREM respository.
Any requests for more experimental data or codes used in the creation of this work are welcome.

\vspace{0.5cm}
\noindent\textbf{Acknowledgements}\\
 We are grateful to Tomek Rybotycki for providing help with running experiments on Rigetti's device. 
 We acknowledge the use of the IBM Quantum Experience and Amazon AWS  Braket for this work. The views expressed are those of the authors and do not reflect the official policy or position of IBM, or the Rigetti, or the Amazon.
 
 \vspace{0.5cm}
\noindent\textbf{Funding}\\
 FBM and MO acknowledge support by the Foundation for Polish Science through IRAP project co-financed by EU within Smart Growth Operational Programme (contract no. 2018/MAB/5) and through TEAM-NET project (contract no. POIR.04.04.00-00-17C1/18-00). 
 FB acknowledges support by the Deutsche Forschungsgemeinschaft (DFG, German Research Foundation) - Project number 414325145 in the framework of the Austrian Science Fund (FWF): SFB F71.
 ZZ was supported by the NKFIH through the Quantum Technology National Excellence Program (Project No. 2017-1.2.1-NKP-2017-00001) and the Quantum Information National Laboratory of Hungary, and by the grants K124152, K124176, KH129601, K120569, FK 135220.

\bibliography{mitigationrefs.bib}

\begin{thebibliography}{79}
\providecommand{\natexlab}[1]{#1}
\providecommand{\url}[1]{\texttt{#1}}
\expandafter\ifx\csname urlstyle\endcsname\relax
  \providecommand{\doi}[1]{doi: #1}\else
  \providecommand{\doi}{doi: \begingroup \urlstyle{rm}\Url}\fi

\bibitem[Arute et~al.(2019)]{Arute2019quantum}
Frank Arute et~al.
\newblock Quantum supremacy using a programmable superconducting processor.
\newblock \emph{Nature}, 574\penalty0 (7779):\penalty0 505--510, Oct 2019.
\newblock ISSN 1476-4687.
\newblock \doi{10.1038/s41586-019-1666-5}.
\newblock URL \url{https://doi.org/10.1038/s41586-019-1666-5}.

\bibitem[Villalonga et~al.(2020)Villalonga, Lyakh, Boixo, Neven, Humble,
  Biswas, Rieffel, Ho, and Mandrà]{Villalonga2020establishing}
Benjamin Villalonga, Dmitry Lyakh, Sergio Boixo, Hartmut Neven, Travis~S
  Humble, Rupak Biswas, Eleanor~G Rieffel, Alan Ho, and Salvatore Mandrà.
\newblock Establishing the quantum supremacy frontier with a 281 pflop/s
  simulation.
\newblock \emph{Quantum Science and Technology}, 5\penalty0 (3):\penalty0
  034003, Apr 2020.
\newblock ISSN 2058-9565.
\newblock \doi{10.1088/2058-9565/ab7eeb}.
\newblock URL \url{http://dx.doi.org/10.1088/2058-9565/ab7eeb}.

\bibitem[{Farhi} and {Harrow}(2016)]{farhi2016quantum}
Edward {Farhi} and Aram~W {Harrow}.
\newblock {Quantum Supremacy through the Quantum Approximate Optimization
  Algorithm}.
\newblock \emph{arXiv e-prints}, art. arXiv:1602.07674, Feb 2016.
\newblock URL \url{https://arxiv.org/abs/1602.07674}.

\bibitem[{Moll} et~al.(2018){Moll}, {Barkoutsos}, {Bishop}, {Chow}, {Cross},
  {Egger}, {Filipp}, {Fuhrer}, {Gambetta}, {Ganzhorn}, {Kandala}, {Mezzacapo},
  {M{\"u}ller}, {Riess}, {Salis}, {Smolin}, {Tavernelli}, and
  {Temme}]{Moll2018Quantum}
Nikolaj {Moll}, Panagiotis {Barkoutsos}, Lev~S. {Bishop}, Jerry~M. {Chow},
  Andrew {Cross}, Daniel~J. {Egger}, Stefan {Filipp}, Andreas {Fuhrer}, Jay~M.
  {Gambetta}, Marc {Ganzhorn}, Abhinav {Kandala}, Antonio {Mezzacapo}, Peter
  {M{\"u}ller}, Walter {Riess}, Gian {Salis}, John {Smolin}, Ivano
  {Tavernelli}, and Kristan {Temme}.
\newblock {Quantum optimization using variational algorithms on near-term
  quantum devices}.
\newblock \emph{Quantum Science and Technology}, 3\penalty0 (3):\penalty0
  030503, Jul 2018.
\newblock \doi{10.1088/2058-9565/aab822}.
\newblock URL \url{https://arxiv.org/abs/1710.01022v2}.

\bibitem[Preskill(2018)]{Preskill2018}
John Preskill.
\newblock Quantum {C}omputing in the {NISQ} era and beyond.
\newblock \emph{{Quantum}}, 2:\penalty0 79, August 2018.
\newblock ISSN 2521-327X.
\newblock \doi{10.22331/q-2018-08-06-79}.
\newblock URL \url{https://doi.org/10.22331/q-2018-08-06-79}.

\bibitem[Wallman and Emerson(2016)]{Wallman2016noise}
Joel~J. Wallman and Joseph Emerson.
\newblock Noise tailoring for scalable quantum computation via randomized
  compiling.
\newblock \emph{Physical Review A}, 94\penalty0 (5), Nov 2016.
\newblock ISSN 2469-9934.
\newblock \doi{10.1103/physreva.94.052325}.
\newblock URL \url{http://dx.doi.org/10.1103/PhysRevA.94.052325}.

\bibitem[Li and Benjamin(2017)]{Li2017efficient}
Ying Li and Simon~C. Benjamin.
\newblock Efficient variational quantum simulator incorporating active error
  minimization.
\newblock \emph{Physical Review X}, 7\penalty0 (2), Jun 2017.
\newblock ISSN 2160-3308.
\newblock \doi{10.1103/physrevx.7.021050}.
\newblock URL \url{http://dx.doi.org/10.1103/PhysRevX.7.021050}.

\bibitem[Temme et~al.(2017)Temme, Bravyi, and Gambetta]{Temme2017error}
Kristan Temme, Sergey Bravyi, and Jay~M. Gambetta.
\newblock Error mitigation for short-depth quantum circuits.
\newblock \emph{Physical Review Letters}, 119\penalty0 (18), Nov 2017.
\newblock ISSN 1079-7114.
\newblock \doi{10.1103/physrevlett.119.180509}.
\newblock URL \url{http://dx.doi.org/10.1103/PhysRevLett.119.180509}.

\bibitem[{Endo} et~al.(2018){Endo}, {Benjamin}, and {Li}]{Endo2018}
Suguru {Endo}, Simon~C. {Benjamin}, and Ying {Li}.
\newblock {Practical Quantum Error Mitigation for Near-Future Applications}.
\newblock \emph{Physical Review X}, 8:\penalty0 031027, Jul 2018.
\newblock \doi{10.1103/PhysRevX.8.031027}.

\bibitem[Kandala et~al.(2019)Kandala, Temme, Córcoles, Mezzacapo, Chow, and
  Gambetta]{Kandala2019error}
Abhinav Kandala, Kristan Temme, Antonio~D. Córcoles, Antonio Mezzacapo,
  Jerry~M. Chow, and Jay~M. Gambetta.
\newblock Error mitigation extends the computational reach of a noisy quantum
  processor.
\newblock \emph{Nature}, 567\penalty0 (7749):\penalty0 491–495, Mar 2019.
\newblock ISSN 1476-4687.
\newblock \doi{10.1038/s41586-019-1040-7}.
\newblock URL \url{http://dx.doi.org/10.1038/s41586-019-1040-7}.

\bibitem[Sun et~al.(2021)Sun, Yuan, Tsunoda, Vedral, Benjamin, and
  Endo]{sun2020mitigating}
Jinzhao Sun, Xiao Yuan, Takahiro Tsunoda, Vlatko Vedral, Simon~C. Benjamin, and
  Suguru Endo.
\newblock Mitigating realistic noise in practical noisy intermediate-scale
  quantum devices.
\newblock \emph{Physical Review Applied}, 15\penalty0 (3), Mar 2021.
\newblock ISSN 2331-7019.
\newblock \doi{10.1103/physrevapplied.15.034026}.
\newblock URL \url{http://dx.doi.org/10.1103/PhysRevApplied.15.034026}.

\bibitem[Huggins et~al.(2020)Huggins, McArdle, O'Brien, Lee, Rubin, Boixo,
  Whaley, Babbush, and McClean]{huggins2020virtual}
William~J. Huggins, Sam McArdle, Thomas~E. O'Brien, Joonho Lee, Nicholas~C.
  Rubin, Sergio Boixo, K.~Birgitta Whaley, Ryan Babbush, and Jarrod~R. McClean.
\newblock Virtual distillation for quantum error mitigation.
\newblock 2020.
\newblock URL \url{https://arxiv.org/abs/2011.07064}.

\bibitem[Maciejewski et~al.(2020{\natexlab{a}})Maciejewski, Zimbor{\'{a}}s, and
  Oszmaniec]{Maciejewski2020mitigation}
Filip~B. Maciejewski, Zolt{\'{a}}n Zimbor{\'{a}}s, and Micha{\l{}} Oszmaniec.
\newblock Mitigation of readout noise in near-term quantum devices by classical
  post-processing based on detector tomography.
\newblock \emph{{Quantum}}, 4:\penalty0 257, April 2020{\natexlab{a}}.
\newblock ISSN 2521-327X.
\newblock \doi{10.22331/q-2020-04-24-257}.
\newblock URL \url{https://doi.org/10.22331/q-2020-04-24-257}.

\bibitem[Chen et~al.(2019)Chen, Farahzad, Yoo, and Wei]{Chen2019detector}
Yanzhu Chen, Maziar Farahzad, Shinjae Yoo, and Tzu-Chieh Wei.
\newblock Detector tomography on ibm quantum computers and mitigation of an
  imperfect measurement.
\newblock \emph{Physical Review A}, 100\penalty0 (5), Nov 2019.
\newblock ISSN 2469-9934.
\newblock \doi{10.1103/physreva.100.052315}.
\newblock URL \url{http://dx.doi.org/10.1103/PhysRevA.100.052315}.

\bibitem[Bravyi et~al.(2021)Bravyi, Sheldon, Kandala, Mckay, and
  Gambetta]{Bravyi2020mitigating}
Sergey Bravyi, Sarah Sheldon, Abhinav Kandala, David~C. Mckay, and Jay~M.
  Gambetta.
\newblock Mitigating measurement errors in multiqubit experiments.
\newblock \emph{Physical Review A}, 103\penalty0 (4), Apr 2021.
\newblock ISSN 2469-9934.
\newblock \doi{10.1103/physreva.103.042605}.
\newblock URL \url{http://dx.doi.org/10.1103/PhysRevA.103.042605}.

\bibitem[Geller and Sun(2021)]{Geller2020efficient}
Michael~R Geller and Mingyu Sun.
\newblock Toward efficient correction of multiqubit measurement errors: pair
  correlation method.
\newblock \emph{Quantum Science and Technology}, 6\penalty0 (2):\penalty0
  025009, feb 2021.
\newblock \doi{10.1088/2058-9565/abd5c9}.
\newblock URL \url{https://doi.org/10.1088/2058-9565/abd5c9}.

\bibitem[Geller(2020)]{Geller2020rigorous}
Michael~R Geller.
\newblock Rigorous measurement error correction.
\newblock \emph{Quantum Science and Technology}, 5\penalty0 (3):\penalty0
  03LT01, Jun 2020.
\newblock ISSN 2058-9565.
\newblock \doi{10.1088/2058-9565/ab9591}.
\newblock URL \url{http://dx.doi.org/10.1088/2058-9565/ab9591}.

\bibitem[Nachman et~al.(2020)Nachman, Urbanek, de~Jong, and
  Bauer]{Nachman2019unfolding}
Benjamin Nachman, Miroslav Urbanek, Wibe~A. de~Jong, and Christian~W. Bauer.
\newblock Unfolding quantum computer readout noise.
\newblock \emph{npj Quantum Information}, 6\penalty0 (1):\penalty0 84, Sep
  2020.
\newblock ISSN 2056-6387.
\newblock \doi{10.1038/s41534-020-00309-7}.
\newblock URL \url{https://doi.org/10.1038/s41534-020-00309-7}.

\bibitem[Kwon and Bae(2020)]{Kwon2020hybrid}
Hyeokjea Kwon and Joonwoo Bae.
\newblock A hybrid quantum-classical approach to mitigating measurement errors
  in quantum algorithms.
\newblock \emph{IEEE Transactions on Computers}, page 1–1, 2020.
\newblock ISSN 2326-3814.
\newblock \doi{10.1109/tc.2020.3009664}.
\newblock URL \url{http://dx.doi.org/10.1109/TC.2020.3009664}.

\bibitem[Hamilton et~al.(2020)Hamilton, Kharazi, Morris, McCaskey, Bennink, and
  Pooser]{Hamilton2020scalable}
Kathleen~E. Hamilton, Tyler Kharazi, Titus Morris, Alexander~J. McCaskey,
  Ryan~S. Bennink, and Raphael~C. Pooser.
\newblock Scalable quantum processor noise characterization.
\newblock In \emph{2020 IEEE International Conference on Quantum Computing and
  Engineering (QCE)}, pages 430--440, 2020.
\newblock \doi{10.1109/QCE49297.2020.00060}.
\newblock URL \url{https://arxiv.org/abs/2006.01805}.

\bibitem[Dahlhauser and Humble(2021)]{lilly2020modeling}
Megan~L. Dahlhauser and Travis~S. Humble.
\newblock Modeling noisy quantum circuits using experimental characterization.
\newblock \emph{Phys. Rev. A}, 103:\penalty0 042603, Apr 2021.
\newblock \doi{10.1103/PhysRevA.103.042603}.
\newblock URL \url{https://link.aps.org/doi/10.1103/PhysRevA.103.042603}.

\bibitem[Funcke et~al.(2020)Funcke, Hartung, Jansen, Kühn, Stornati, and
  Wang]{funcke2020measurement}
Lena Funcke, Tobias Hartung, Karl Jansen, Stefan Kühn, Paolo Stornati, and
  Xiaoyang Wang.
\newblock Measurement error mitigation in quantum computers through classical
  bit-flip correction.
\newblock 2020.
\newblock URL \url{https://arxiv.org/abs/2007.03663}.

\bibitem[Zheng et~al.(2020)Zheng, Li, Terlaky, and Yang]{zheng2020bayesian}
Muqing Zheng, Ang Li, Tamás Terlaky, and Xiu Yang.
\newblock A bayesian approach for characterizing and mitigating gate and
  measurement errors.
\newblock 2020.
\newblock URL \url{https://arxiv.org/abs/2010.09188}.

\bibitem[Farhi et~al.(2014)Farhi, Goldstone, and Gutmann]{farhi2014quantum}
Edward Farhi, Jeffrey Goldstone, and Sam Gutmann.
\newblock A quantum approximate optimization algorithm.
\newblock 2014.
\newblock URL \url{https://arxiv.org/abs/1411.4028}.

\bibitem[Cerezo et~al.(2020)Cerezo, Arrasmith, Babbush, Benjamin, Endo, Fujii,
  McClean, Mitarai, Yuan, Cincio, and Coles]{cerezo2020variational}
M~Cerezo, Andrew Arrasmith, Ryan Babbush, Simon~C Benjamin, Suguru Endo,
  Keisuke Fujii, Jarrod~R McClean, Kosuke Mitarai, Xiao Yuan, Lukasz Cincio,
  and P.~J. Coles.
\newblock Variational quantum algorithms.
\newblock 2020.
\newblock URL \url{https://arxiv.org/abs/2012.09265}.

\bibitem[Hadfield et~al.(2019)Hadfield, Wang, O’Gorman, Rieffel, Venturelli,
  and Biswas]{hadfield2019quantum}
Stuart Hadfield, Zhihui Wang, Bryan O’Gorman, Eleanor Rieffel, Davide
  Venturelli, and Rupak Biswas.
\newblock From the quantum approximate optimization algorithm to a quantum
  alternating operator ansatz.
\newblock \emph{Algorithms}, 12\penalty0 (2):\penalty0 34, Feb 2019.
\newblock ISSN 1999-4893.
\newblock \doi{10.3390/a12020034}.
\newblock URL \url{http://dx.doi.org/10.3390/a12020034}.

\bibitem[Lloyd(2018)]{lloyd2018quantum}
Seth Lloyd.
\newblock Quantum approximate optimization is computationally universal.
\newblock 2018.
\newblock URL \url{https://arxiv.org/abs/1812.11075}.

\bibitem[Morales et~al.(2020)Morales, Biamonte, and
  Zimbor{\'a}s]{Morales2020universality}
M.~E.~S. Morales, J.~D. Biamonte, and Z.~Zimbor{\'a}s.
\newblock On the universality of the quantum approximate optimization
  algorithm.
\newblock \emph{Quantum Information Processing}, 19\penalty0 (9):\penalty0 291,
  Aug 2020.
\newblock ISSN 1573-1332.
\newblock \doi{10.1007/s11128-020-02748-9}.
\newblock URL \url{https://doi.org/10.1007/s11128-020-02748-9}.

\bibitem[Hansen and Jaumard(1990)]{Hansen1990algorithms}
Pierre Hansen and Brigitte Jaumard.
\newblock Algorithms for the maximum satisfiability problem.
\newblock \emph{Computing}, 44\penalty0 (4):\penalty0 279--303, Dec 1990.
\newblock ISSN 1436-5057.
\newblock \doi{10.1007/BF02241270}.
\newblock URL \url{https://doi.org/10.1007/BF02241270}.

\bibitem[Guerreschi and Matsuura(2019)]{Guerreschi2019QAOA}
G.~G. Guerreschi and A.~Y. Matsuura.
\newblock Qaoa for max-cut requires hundreds of qubits for quantum speed-up.
\newblock \emph{Scientific Reports}, 9\penalty0 (1):\penalty0 6903, May 2019.
\newblock ISSN 2045-2322.
\newblock \doi{10.1038/s41598-019-43176-9}.
\newblock URL \url{https://doi.org/10.1038/s41598-019-43176-9}.

\bibitem[Panchenko(2012)]{Panchenko2012Sherrington}
Dmitry Panchenko.
\newblock The sherrington-kirkpatrick model: An overview.
\newblock \emph{Journal of Statistical Physics}, 149\penalty0 (2):\penalty0
  362–383, Sep 2012.
\newblock ISSN 1572-9613.
\newblock \doi{10.1007/s10955-012-0586-7}.
\newblock URL \url{http://dx.doi.org/10.1007/s10955-012-0586-7}.

\bibitem[Farhi et~al.(2019)Farhi, Goldstone, Gutmann, and
  Zhou]{farhi2019quantum}
Edward Farhi, Jeffrey Goldstone, Sam Gutmann, and Leo Zhou.
\newblock The quantum approximate optimization algorithm and the
  sherrington-kirkpatrick model at infinite size.
\newblock 2019.
\newblock URL \url{https://arxiv.org/abs/1910.08187}.

\bibitem[Cotler and Wilczek(2020)]{Cotler2020quantum}
Jordan Cotler and Frank Wilczek.
\newblock Quantum overlapping tomography.
\newblock \emph{Physical Review Letters}, 124\penalty0 (10), Mar 2020.
\newblock ISSN 1079-7114.
\newblock \doi{10.1103/physrevlett.124.100401}.
\newblock URL \url{http://dx.doi.org/10.1103/PhysRevLett.124.100401}.

\bibitem[Koch et~al.(2007)Koch, Yu, Gambetta, Houck, Schuster, Majer, Blais,
  Devoret, Girvin, and Schoelkopf]{Koch2007}
Jens Koch, Terri~M. Yu, Jay Gambetta, A.~A. Houck, D.~I. Schuster, J.~Majer,
  Alexandre Blais, M.~H. Devoret, S.~M. Girvin, and R.~J. Schoelkopf.
\newblock Charge-insensitive qubit design derived from the cooper pair box.
\newblock \emph{Phys. Rev. A}, 76:\penalty0 042319, Oct 2007.
\newblock \doi{10.1103/PhysRevA.76.042319}.
\newblock URL \url{https://link.aps.org/doi/10.1103/PhysRevA.76.042319}.

\bibitem[Xue et~al.(2021)Xue, Chen, Wu, and Guo]{xue2019effects}
Cheng Xue, Zhao-Yun Chen, Yu-Chun Wu, and Guo-Ping Guo.
\newblock Effects of quantum noise on quantum approximate optimization
  algorithm.
\newblock \emph{Chinese Physics Letters}, 38\penalty0 (3):\penalty0 030302, mar
  2021.
\newblock \doi{10.1088/0256-307x/38/3/030302}.
\newblock URL \url{https://doi.org/10.1088/0256-307x/38/3/030302}.

\bibitem[Marshall et~al.(2020)Marshall, Wudarski, Hadfield, and
  Hogg]{marshall2020characterizing}
Jeffrey Marshall, Filip Wudarski, Stuart Hadfield, and Tad Hogg.
\newblock Characterizing local noise in qaoa circuits.
\newblock \emph{IOP SciNotes}, 1\penalty0 (2):\penalty0 025208, Aug 2020.
\newblock ISSN 2633-1357.
\newblock \doi{10.1088/2633-1357/abb0d7}.
\newblock URL \url{http://dx.doi.org/10.1088/2633-1357/abb0d7}.

\bibitem[Alam et~al.(2019)Alam, Ash-Saki, and Ghosh]{alam2019analysis}
Mahabubul Alam, Abdullah Ash-Saki, and Swaroop Ghosh.
\newblock Analysis of quantum approximate optimization algorithm under
  realistic noise in superconducting qubits.
\newblock 2019.
\newblock URL \url{https://arxiv.org/abs/1907.09631}.

\bibitem[Harrigan et~al.(2021)]{arute2020quantum}
Matthew~P. Harrigan et~al.
\newblock Quantum approximate optimization of non-planar graph problems on a
  planar superconducting processor.
\newblock \emph{Nature Physics}, 17\penalty0 (3):\penalty0 332--336, Mar 2021.
\newblock ISSN 1745-2481.
\newblock \doi{10.1038/s41567-020-01105-y}.
\newblock URL \url{https://doi.org/10.1038/s41567-020-01105-y}.

\bibitem[Montanaro and Stanisic(2020)]{montanaro2020compressed}
Ashley Montanaro and Stasja Stanisic.
\newblock Compressed variational quantum eigensolver for the fermi-hubbard
  model.
\newblock 2020.
\newblock URL \url{https://arxiv.org/abs/2006.01179}.

\bibitem[Gokhale et~al.(2020)Gokhale, Javadi-Abhari, Earnest, Shi, and
  Chong]{gokhale2020optimized}
Pranav Gokhale, Ali Javadi-Abhari, Nathan Earnest, Yunong Shi, and Frederic~T.
  Chong.
\newblock Optimized quantum compilation for near-term algorithms with
  openpulse.
\newblock 2020.
\newblock URL \url{https://arxiv.org/abs/2004.11205}.

\bibitem[Peres(2006)]{Peres2006}
Asher Peres.
\newblock \emph{Quantum theory: Concepts and methods}, volume~57.
\newblock Springer Science \& Business Media, 2006.
\newblock \doi{https://doi.org/10.1007/0-306-47120-5}.

\bibitem[Lundeen et~al.(2008)Lundeen, Feito, Coldenstrodt-Ronge, Pregnell,
  Silberhorn, Ralph, Eisert, Plenio, and Walmsley]{Lundeen2008}
J.~S. Lundeen, A.~Feito, H.~Coldenstrodt-Ronge, K.~L. Pregnell, Ch. Silberhorn,
  T.~C. Ralph, J.~Eisert, M.~B. Plenio, and I.~A. Walmsley.
\newblock Tomography of quantum detectors.
\newblock \emph{Nature Physics}, 5:\penalty0 27, November 2008.
\newblock URL \url{http://dx.doi.org/10.1038/nphys1133}.

\bibitem[Hradil et~al.(2004)Hradil, {\v{R}}eh{\'a}{\v{c}}ek,
  Fiur{\'a}{\v{s}}ek, and Je{\v{z}}ek]{Hradil2004}
Zden{\v{e}}k Hradil, Jaroslav {\v{R}}eh{\'a}{\v{c}}ek, Jarom{\'i}r
  Fiur{\'a}{\v{s}}ek, and Miroslav Je{\v{z}}ek.
\newblock \emph{3 Maximum-Likelihood Methodsin Quantum Mechanics}, pages
  59--112.
\newblock Springer Berlin Heidelberg, Berlin, Heidelberg, 2004.
\newblock ISBN 978-3-540-44481-7.
\newblock \doi{10.1007/978-3-540-44481-7_3}.
\newblock URL \url{https://doi.org/10.1007/978-3-540-44481-7_3}.

\bibitem[{Fiur{\'a}{\v{s}}ek}(2001)]{Fiurasek2001}
Jarom{\'\i}r {Fiur{\'a}{\v{s}}ek}.
\newblock {Maximum-likelihood estimation of quantum measurement}.
\newblock \emph{Physical Review A}, 64:\penalty0 024102, August 2001.
\newblock \doi{10.1103/PhysRevA.64.024102}.

\bibitem[Gianani et~al.(2020)Gianani, Teo, Cimini, Jeong, Leuchs, Barbieri, and
  Sánchez-Soto]{Gianani2020compressively}
I.~Gianani, Y.S. Teo, V.~Cimini, H.~Jeong, G.~Leuchs, M.~Barbieri, and L.L.
  Sánchez-Soto.
\newblock Compressively certifying quantum measurements.
\newblock \emph{PRX Quantum}, 1\penalty0 (2), Oct 2020.
\newblock ISSN 2691-3399.
\newblock \doi{10.1103/prxquantum.1.020307}.
\newblock URL \url{http://dx.doi.org/10.1103/PRXQuantum.1.020307}.

\bibitem[Evans et~al.(2019)Evans, Harper, and Flammia]{evans2019scalable}
Tim~J. Evans, Robin Harper, and Steven~T. Flammia.
\newblock Scalable bayesian hamiltonian learning.
\newblock 2019.
\newblock URL \url{https://arxiv.org/abs/1912.07636}.

\bibitem[Yu(2020)]{yu2020sample}
Nengkun Yu.
\newblock Sample efficient tomography via {Pauli} measurements.
\newblock 2020.
\newblock URL \url{https://arxiv.org/abs/2009.04610}.

\bibitem[Majewski et~al.(1996)Majewski, Wormald, Havas, and
  Czech]{Majewski1996family}
B.~S. Majewski, N.~C. Wormald, G.~Havas, and Z.~J. Czech.
\newblock {A Family of Perfect Hashing Methods}.
\newblock \emph{The Computer Journal}, 39\penalty0 (6):\penalty0 547--554, 01
  1996.
\newblock ISSN 0010-4620.
\newblock \doi{10.1093/comjnl/39.6.547}.
\newblock URL \url{https://doi.org/10.1093/comjnl/39.6.547}.

\bibitem[Stinson et~al.(2000)Stinson, Wei, and Zhu]{Stinson2000perfect}
D.~R. Stinson, R.~Wei, and L.~Zhu.
\newblock New constructions for perfect hash families and related structures
  using combinatorial designs and codes.
\newblock \emph{Journal of Combinatorial Designs}, 8\penalty0 (3):\penalty0
  189--200, 2000.
\newblock
  \doi{https://doi.org/10.1002/(SICI)1520-6610(2000)8:3<189::AID-JCD4>3.0.CO;2-A}.

\bibitem[Blackburn(2000)]{blackburn2000perfect}
Simon~R. Blackburn.
\newblock Perfect hash families: Probabilistic methods and explicit
  constructions.
\newblock \emph{Journal of Combinatorial Theory, Series A}, 92\penalty0
  (1):\penalty0 54 -- 60, 2000.
\newblock ISSN 0097-3165.
\newblock \doi{https://doi.org/10.1006/jcta.1999.3050}.
\newblock URL
  \url{https://www.sciencedirect.com/science/article/pii/S0097316599930509}.

\bibitem[{Alon} and {Gutner}(2008)]{alon2008balanced}
Noga {Alon} and Shai {Gutner}.
\newblock {Balanced Families of Perfect Hash Functions and Their Applications}.
\newblock May 2008.
\newblock URL \url{https://arxiv.org/abs/0805.4300}.

\bibitem[Santra et~al.(2014)Santra, Quiroz, Steeg, and Lidar]{Santra2014max}
Siddhartha Santra, Gregory Quiroz, Greg~Ver Steeg, and Daniel~A Lidar.
\newblock Max 2-{SAT} with up to 108 qubits.
\newblock \emph{New Journal of Physics}, 16\penalty0 (4):\penalty0 045006, apr
  2014.
\newblock \doi{10.1088/1367-2630/16/4/045006}.
\newblock URL \url{https://doi.org/10.1088/1367-2630/16/4/045006}.

\bibitem[Rudinger et~al.(2019)Rudinger, Proctor, Langharst, Sarovar, Young, and
  Blume-Kohout]{Rudinger2019probing}
Kenneth Rudinger, Timothy Proctor, Dylan Langharst, Mohan Sarovar, Kevin Young,
  and Robin Blume-Kohout.
\newblock Probing context-dependent errors in quantum processors.
\newblock \emph{Physical Review X}, 9\penalty0 (2), Jun 2019.
\newblock ISSN 2160-3308.
\newblock \doi{10.1103/physrevx.9.021045}.
\newblock URL \url{http://dx.doi.org/10.1103/PhysRevX.9.021045}.

\bibitem[Farhi et~al.(2020)Farhi, Gamarnik, and Gutmann]{farhi2020quantum}
Edward Farhi, David Gamarnik, and Sam Gutmann.
\newblock The quantum approximate optimization algorithm needs to see the whole
  graph: A typical case.
\newblock 2020.
\newblock URL \url{https://arxiv.org/abs/2004.09002}.

\bibitem[{Popescu} et~al.(2006){Popescu}, {Short}, and
  {Winter}]{entFOUNDstatMech}
Sandu {Popescu}, Anthony~J. {Short}, and Andreas {Winter}.
\newblock {Entanglement and the foundations of statistical mechanics}.
\newblock \emph{Nature Physics}, 2\penalty0 (11):\penalty0 754--758, November
  2006.
\newblock \doi{10.1038/nphys444}.

\bibitem[Oszmaniec et~al.(2016)Oszmaniec, Augusiak, Gogolin,
  Ko\l{}ody\ifmmode~\acute{n}\else \'{n}\fi{}ski, Ac\'{\i}n, and
  Lewenstein]{OszmaniecMetro}
M.~Oszmaniec, R.~Augusiak, C.~Gogolin, J.~Ko\l{}ody\ifmmode~\acute{n}\else
  \'{n}\fi{}ski, A.~Ac\'{\i}n, and M.~Lewenstein.
\newblock Random bosonic states for robust quantum metrology.
\newblock \emph{Phys. Rev. X}, 6:\penalty0 041044, Dec 2016.
\newblock \doi{10.1103/PhysRevX.6.041044}.
\newblock URL \url{https://link.aps.org/doi/10.1103/PhysRevX.6.041044}.

\bibitem[{Brand{\~a}o} et~al.(2016){Brand{\~a}o}, {Harrow}, and
  {Horodecki}]{BHH2016}
Fernando G.~S.~L. {Brand{\~a}o}, Aram~W. {Harrow}, and Micha{\l} {Horodecki}.
\newblock {Local Random Quantum Circuits are Approximate Polynomial-Designs}.
\newblock \emph{Communications in Mathematical Physics}, 346\penalty0
  (2):\penalty0 397--434, September 2016.
\newblock \doi{10.1007/s00220-016-2706-8}.

\bibitem[{Cotler} et~al.(2020){Cotler}, {Hunter-Jones}, and {Ranard}]{Nick2020}
Jordan {Cotler}, Nicholas {Hunter-Jones}, and Daniel {Ranard}.
\newblock {Fluctuations of subsystem entropies at late times}.
\newblock October 2020.
\newblock URL \url{https://arxiv.org/abs/2010.11922}.

\bibitem[Spall(1998)]{Spall1998ANOO}
J.~Spall.
\newblock An overview of the simultaneous perturbation method for efficient
  optimization.
\newblock \emph{Johns Hopkins Apl Technical Digest}, 19:\penalty0 482--492,
  1998.
\newblock URL
  \url{https://www.jhuapl.edu/Content/techdigest/pdf/V19-N04/19-04-Spall.pdf}.

\bibitem[Cade et~al.(2020)Cade, Mineh, Montanaro, and
  Stanisic]{cade2019strategies}
Chris Cade, Lana Mineh, Ashley Montanaro, and Stasja Stanisic.
\newblock Strategies for solving the fermi-hubbard model on near-term quantum
  computers.
\newblock \emph{Phys. Rev. B}, 102:\penalty0 235122, Dec 2020.
\newblock \doi{10.1103/PhysRevB.102.235122}.
\newblock URL \url{https://link.aps.org/doi/10.1103/PhysRevB.102.235122}.

\bibitem[{Kandala} et~al.(2017){Kandala}, {Mezzacapo}, {Temme}, {Takita},
  {Brink}, {Chow}, and {Gambetta}]{Kandala2017}
Abhinav {Kandala}, Antonio {Mezzacapo}, Kristan {Temme}, Maika {Takita}, Markus
  {Brink}, Jerry~M. {Chow}, and Jay~M. {Gambetta}.
\newblock {Hardware-efficient variational quantum eigensolver for small
  molecules and quantum magnets}.
\newblock \emph{\nat}, 549:\penalty0 242--246, Sep 2017.
\newblock \doi{10.1038/nature23879}.

\bibitem[Huggins et~al.(2021)Huggins, McClean, Rubin, Jiang, Wiebe, Whaley, and
  Babbush]{huggins2020efficient}
William~J. Huggins, Jarrod~R. McClean, Nicholas~C. Rubin, Zhang Jiang, Nathan
  Wiebe, K.~Birgitta Whaley, and Ryan Babbush.
\newblock Efficient and noise resilient measurements for quantum chemistry on
  near-term quantum computers.
\newblock \emph{npj Quantum Information}, 7\penalty0 (1):\penalty0 23, Feb
  2021.
\newblock ISSN 2056-6387.
\newblock \doi{10.1038/s41534-020-00341-7}.
\newblock URL \url{https://doi.org/10.1038/s41534-020-00341-7}.

\bibitem[Barron and Wood(2020)]{barron2020measurement}
George~S. Barron and Christopher~J. Wood.
\newblock Measurement error mitigation for variational quantum algorithms.
\newblock 2020.
\newblock URL \url{https://arxiv.org/abs/2010.08520}.

\bibitem[Proctor et~al.(2020)Proctor, Revelle, Nielsen, Rudinger, Lobser,
  Maunz, Blume-Kohout, and Young]{Proctor2020detecting}
Timothy Proctor, Melissa Revelle, Erik Nielsen, Kenneth Rudinger, Daniel
  Lobser, Peter Maunz, Robin Blume-Kohout, and Kevin Young.
\newblock Detecting and tracking drift in quantum information processors.
\newblock \emph{Nature Communications}, 11\penalty0 (1):\penalty0 5396, Oct
  2020.
\newblock ISSN 2041-1723.
\newblock \doi{10.1038/s41467-020-19074-4}.
\newblock URL \url{https://doi.org/10.1038/s41467-020-19074-4}.

\bibitem[Dasgupta and Humble(2020)]{dasgupta2020characterizing}
Samudra Dasgupta and Travis~S. Humble.
\newblock Characterizing the stability of nisq devices.
\newblock In \emph{2020 IEEE International Conference on Quantum Computing and
  Engineering (QCE)}, pages 419--429, 2020.
\newblock \doi{10.1109/QCE49297.2020.00059}.

\bibitem[Govia et~al.(2020)Govia, Ribeill, Rist{\`e}, Ware, and
  Krovi]{Govia2020bootstraping}
L.~C.~G. Govia, G.~J. Ribeill, D.~Rist{\`e}, M.~Ware, and H.~Krovi.
\newblock Bootstrapping quantum process tomography via a perturbative ansatz.
\newblock \emph{Nature Communications}, 11\penalty0 (1):\penalty0 1084, Feb
  2020.
\newblock ISSN 2041-1723.
\newblock \doi{10.1038/s41467-020-14873-1}.
\newblock URL \url{https://doi.org/10.1038/s41467-020-14873-1}.

\bibitem[Flammia and Wallman(2020)]{Flammia2020efficient}
Steven~T. Flammia and Joel~J. Wallman.
\newblock Efficient estimation of {Pauli} channels.
\newblock \emph{ACM Transactions on Quantum Computing}, 1\penalty0
  (1):\penalty0 1–32, Dec 2020.
\newblock ISSN 2643-6817.
\newblock \doi{10.1145/3408039}.
\newblock URL \url{http://dx.doi.org/10.1145/3408039}.

\bibitem[Harper et~al.(2020)Harper, Flammia, and Wallman]{Harper2020efficient}
Robin Harper, Steven~T. Flammia, and Joel~J. Wallman.
\newblock Efficient learning of quantum noise.
\newblock \emph{Nature Physics}, 16\penalty0 (12):\penalty0 1184–1188, Aug
  2020.
\newblock ISSN 1745-2481.
\newblock \doi{10.1038/s41567-020-0992-8}.
\newblock URL \url{http://dx.doi.org/10.1038/s41567-020-0992-8}.

\bibitem[Huang et~al.(2020)Huang, Kueng, and Preskill]{Huang2020predicting}
Hsin-Yuan Huang, Richard Kueng, and John Preskill.
\newblock Predicting many properties of a quantum system from very few
  measurements.
\newblock \emph{Nature Physics}, 16\penalty0 (10):\penalty0 1050–1057, Jun
  2020.
\newblock ISSN 1745-2481.
\newblock \doi{10.1038/s41567-020-0932-7}.
\newblock URL \url{http://dx.doi.org/10.1038/s41567-020-0932-7}.

\bibitem[Chen et~al.(2020)Chen, Yu, Zeng, and Flammia]{chen2020robust}
Senrui Chen, Wenjun Yu, Pei Zeng, and Steven~T. Flammia.
\newblock Robust shadow estimation.
\newblock 2020.
\newblock URL \url{https://arxiv.org/abs/2011.09636}.

\bibitem[Boothby et~al.(2020)Boothby, Bunyk, Raymond, and
  Roy]{boothby2020nextgeneration}
Kelly Boothby, Paul Bunyk, Jack Raymond, and Aidan Roy.
\newblock Next-generation topology of d-wave quantum processors.
\newblock 2020.
\newblock URL \url{https://arxiv.org/abs/2003.00133}.

\bibitem[Bruzewicz et~al.(2019)Bruzewicz, Chiaverini, McConnell, and
  Sage]{bruzewicz2019trapped}
Colin~D. Bruzewicz, John Chiaverini, Robert McConnell, and Jeremy~M. Sage.
\newblock Trapped-ion quantum computing: Progress and challenges.
\newblock \emph{Applied Physics Reviews}, 6\penalty0 (2):\penalty0 021314, Jun
  2019.
\newblock ISSN 1931-9401.
\newblock \doi{10.1063/1.5088164}.
\newblock URL \url{http://dx.doi.org/10.1063/1.5088164}.

\bibitem[Wang et~al.(2019)Wang, Sciarrino, Laing, and
  Thompson]{Wang2019integrated}
Jianwei Wang, Fabio Sciarrino, Anthony Laing, and Mark~G. Thompson.
\newblock Integrated photonic quantum technologies.
\newblock \emph{Nature Photonics}, 14\penalty0 (5):\penalty0 273–284, Oct
  2019.
\newblock ISSN 1749-4893.
\newblock \doi{10.1038/s41566-019-0532-1}.
\newblock URL \url{http://dx.doi.org/10.1038/s41566-019-0532-1}.

\bibitem[Abraham et~al.(2019)]{Qiskit_ref_proper}
H{\'e}ctor Abraham et~al.
\newblock Qiskit: An open-source framework for quantum computing, 2019.
\newblock URL \url{https://qiskit.org/documentation/}.

\bibitem[Maciejewski et~al.(2020{\natexlab{b}})Maciejewski, Rybotycki, and
  Oszmaniec]{qrem}
F.~B. Maciejewski, T.~Rybotycki, and M.~Oszmaniec.
\newblock Quantum readout errors mitigation (qrem) -- open source github
  repository, 2020{\natexlab{b}}.
\newblock URL \url{https://github.com/fbm2718/QREM}.

\bibitem[Weissman et~al.(2003)Weissman, Ordentlich, Seroussi, Verdu1, and
  Weinberger]{Weissman2003}
Tsachy Weissman, Erik Ordentlich, Gadiel Seroussi, Sergio Verdu1, and
  Marcelo~J. Weinberger.
\newblock Inequalities for the l1 deviation of the empirical distribution.
\newblock \emph{Technical Report HPL-2003-97R1, Hewlett-Packard Labs}, 08 2003.
\newblock URL
  \url{https://www.hpl.hp.com/techreports/2003/HPL-2003-97R1.pdf?origin=publicationDetail.}

\bibitem[Pucha\l{}a et~al.(2018)Pucha\l{}a, Pawela, Krawiec, and
  Kukulski]{Puchala2018}
Zbigniew Pucha\l{}a, \L{}ukasz Pawela, Aleksandra Krawiec, and Ryszard
  Kukulski.
\newblock Strategies for optimal single-shot discrimination of quantum
  measurements.
\newblock \emph{Physical Review A}, 98\penalty0 (4), Oct 2018.
\newblock ISSN 2469-9934.
\newblock \doi{10.1103/physreva.98.042103}.
\newblock URL \url{http://dx.doi.org/10.1103/PhysRevA.98.042103}.

\bibitem[Nielsen and Chuang(2010)]{Nielsen2010}
Michael~A. Nielsen and Isaac~L. Chuang.
\newblock \emph{Quantum Computation and Quantum Information: 10th Anniversary
  Edition}.
\newblock Cambridge University Press, 2010.
\newblock \doi{10.1017/CBO9780511976667}.

\bibitem[Akshay et~al.(2020)Akshay, Philathong, Morales, and
  Biamonte]{Akshay2020reachability}
V.~Akshay, H.~Philathong, M.~E.~S. Morales, and J.~D. Biamonte.
\newblock Reachability deficits in quantum approximate optimization.
\newblock \emph{Phys. Rev. Lett.}, 124:\penalty0 090504, Mar 2020.
\newblock \doi{10.1103/PhysRevLett.124.090504}.
\newblock URL \url{https://link.aps.org/doi/10.1103/PhysRevLett.124.090504}.

\end{thebibliography}

\onecolumngrid
\appendix
\section*{Appendices}

We collect here technical results that are used in the main part of the paper.  Some of the results stated here can be of independent interest for further works on quantum error mitigation. In Sec.~\ref{sec:app:marginals} we discuss the details of error mitigation on marginals for correlated readout noise models, while Sec.~\ref{sec:app:ddot} provides details of our noise characterization procedure based on the Diagonal Detector Overlapping Tomography technique. In Sec.~\ref{sec:app:characterization} we give a short overview of the whole noise characterization scheme in a step-by-step manner.  
 Results concerning the sample complexity of energy estimation are discussed in Sec.~\ref{sec:app:sampling}, and finally in Sec.~\ref{sec:app:misc} we provide some additional experimental data and details of the numerical simulations.

\section{Correlated readout noise model and its usage for error-mitigation on marginals\label{sec:app:marginals}}

In this section, we provide some details on how correlated measurement noise affects marginal probability distributions. 
We start by providing in Appendix~\ref{sec:app:marginals_noise_model} explicit relation between ideal and noisy marginals, when the global probability distribution is affected by readout noise given by our model.
Then in Appendix~\ref{sec:app:marginals_proof1}, we discuss the error-mitigation on the level of marginals, including proof of Proposition~I from the main text.
We finish by analyzing errors in simultaneous estimation of multiple marginal distributions (Appendix~\ref{sec:app:marginals_statistical}) and estimation of multiple expected values of local terms (Appendix~\ref{sec:app:marginals_energy}).

\subsection{Noise model for marginal distributions\label{sec:app:marginals_noise_model}}

In this subsection, we show how to translate a noise model for the full probability distribution given in Eq.~\eqref{eq:noise_model_correlated} into a simple noise model for marginal distributions. 
Our results can be summarised in the following Proposition~\ref{lem:noise_on_marginal}.

\begin{proposition}\label{lem:noise_on_marginal}

Let $p (Y_1 Y_2\ldots,Y_N)$ and $\tilde{p} (X_1 X_2,\ldots,X_N)$ be the ideal distribution and the one obtained on noisy detector respectively.
Define by $p(\Y_{S})$ and $\tilde{p} (\X_{\S})$ their corresponding marginal distributions for qubits in some subset $\S$. 
If the full distributions are connected by a stochastic matrix of the form Eq.~\eqref{eq:noise_model_correlated} by means of
\begin{equation}\label{eqapp:noisydis}
\tilde{p}\rbracket{X_{1}\dots X_N} =  \sum_{Y_1Y_2\dots Y_N} p\rbracket{Y_1Y_2\dots Y_N} \, \Lambda_{X_1X_2\dots X_N|Y_1Y_2\dots Y_N}    \ ,    
\end{equation}
then there exists a left-stochastic matrix $\Lambda^{(\S)}$ that connects the noisy marginal to the ideal one, namely
\begin{equation}\label{eqapp:noisemarginal}
\tilde{p} (\X_{\S})  = \sum_{\Y_{\S}}   \, \Lambda^{(\S)}_{\X_{\S} | \Y_{\S}} \, p (\Y_{\S}) \, \, .
\end{equation}
Moreover, the form of the marginal noise matrix can be explicitly derived from the decomposition Eq.~\eqref{eq:noise_model_correlated}.
\end{proposition}

We will now proceed to prove the above Proposition by deriving the concrete expression of $\Lambda^{(\S)}_{\X_{\S} | \Y_{\S} }$ as a function of the matrices in the given noise model. 
Let us start by computing the marginal of the noisy distribution on general $\S$.
Let us denote by $C\rbracket{\S} = \cbracket{\mathcal{C}_{\chi}}_{\chi}$ set of clusters which contain qubits from $\S$ (note that we might be interested in taking marginal over some qubits from clusters to which qubits from $\S$ belong), and by $L\rbracket{\S}=\cbracket{\chi}_{\chi}$ the corresponding set of labels of those clusters.
Following from Eq.~ \eqref{eqapp:noisydis} with global noise noise map given by Eq.~\eqref{eq:noise_model_correlated} we have
\begin{align*}\label{eq:appendix_noisymargder}
	\tilde{p} \rbracket{\X_{\S}} &= \sum_{X\notin \X_{\S}} \tilde{p} \rbracket{X_1X_2\dots X_N}    \\
	&= 
	\sum_{\substack{X \in C\rbracket{\S}\\X \notin \X_{\S}}}\sum_{X\notin C\rbracket{\S}}\sum_{Y_1Y_2\dots Y_N} p\rbracket{Y_1Y_2\dots Y_N} \prod_{\chi} \Lambda_{\X_{\C_{\chi}}|\Y_{\C_{\chi}}|}^{\Y_{\N_{\chi}}}  \ = \\
	&= 	\sum_{\substack{X \in C\rbracket{\S}\\X \notin \X_{\S}}}\sum_{Y_1Y_2\dots Y_N} p\rbracket{Y_1Y_2\dots Y_N} \prod_{\chi \in L\rbracket{\S}} \Lambda_{\X_{\C_{\chi}}|\Y_{\C_{\chi}}|}^{\Y_{\N_{\chi}}}  \prod_{\chi \notin L\rbracket{\S}}\sum_{X\notin C\rbracket{\S}}\Lambda_{\X_{\C_{\chi}}|\Y_{\C_{\chi}}}^{\Y_{\N_{\chi}}}  \ = \\
	&= 	\sum_{\substack{X \in C\rbracket{\S}\\X \notin \X_{\S}}}\sum_{Y_1Y_2\dots Y_N} p\rbracket{Y_1Y_2\dots Y_N} \prod_{\chi \in L\rbracket{\S}} \Lambda_{\X_{\C_{\chi}}|\Y_{\C_{\chi}}|}^{\Y_{\N_{\chi}}} , = 
 \\
	&= 	\sum_{\substack{X \in C\rbracket{\S}\\X \notin \X_{\S}}}\ \ \sum_{Y \in \Y_{\N\rbracket{\S}\cup C\rbracket{\S}} } p\rbracket{\Y_{\N\rbracket{\S}\cup C\rbracket{\S}}} \prod_{\chi \in L\rbracket{\S}} \Lambda_{\X_{\C_{\chi}}|\Y_{\C_{\chi}}|}^{\Y_{\N_{\chi}}}\ , \numberthis
\end{align*}
where to obtain the above simplifications, we have exploited the fact that the noise matrices $\Lambda^{\Y_{\N_\chi}}$ are all left-stochastic (for any fixed $Y_{\N_\chi}$) and we have defined $ p\rbracket{\Y_{ \N\rbracket{\S}\cup C\rbracket{\S}}}$ as the marginal of the ideal distribution on the qubits belonging both the clusters and the neighborhoods of qubits from $\S$.
By $\N\rbracket{\S}$ we denote set of qubits from neighbourhoods of clusters $C\rbracket{\S}$ but without including the qubits from $\S$, i.e., $\N\rbracket{\S} = \cup_{\chi \in L\rbracket{\S}} \rbracket{\N_{\chi} / \S}$.
Note that this additional requirement is necessary since qubits which are neighbors of some clusters in $C\rbracket{\S}$ might belong to some other clusters in $C\rbracket{\S}$.

Now by means of the chain rule for probability we decompose
\begin{align}
p\rbracket{\Y_{ \N\rbracket{\S}\cup C\rbracket{\S}}} = p\rbracket{\Y_{\S}} p\rbracket{\Y_{\lbrace \N\rbracket{\S} \cup C\rbracket{\S} \rbrace / \lbrace \S \rbrace } |\Y_{\S}}
\end{align}
which after substituting to the Eq.~\eqref{eq:appendix_noisymargder} gives
\begin{align}\label{eq:appendix_marginal_derivation_step3}
  \tilde{p} \rbracket{\X_\S} =  \sum_{Y \in \Y_{\S}}p\rbracket{\Y_{\S}}	\sum_{\substack{X \in C\rbracket{\S}\\X \notin \X_{\S}}}\ \sum_{\substack{Y \in \Y_{\N\rbracket{\S}\cup C\rbracket{\S}}\\ Y\notin \Y_\S }} p\rbracket{\Y_{\lbrace \N\rbracket{\S} \cup C\rbracket{\S} \rbrace / \lbrace \S \rbrace } |\Y_{\S}} \prod_{\chi \in L\rbracket{\S}} \Lambda_{\X_{\C_{\chi}}|\Y_{\C_{\chi}}|}^{\Y_{\N_{\chi}}} \ .
\end{align}

Now notice that Eq.~\eqref{eq:appendix_marginal_derivation_step3} coincides exactly with Eq.~\eqref{eqapp:noisemarginal} if we identify the marginal noise matrix elements as
\begin{align}\label{eq:appendix_marginal_lambda2q}
\Lambda^{(\S)}_{\X_\S|Y_\S} \coloneqq 
\sum_{\substack{X \in C\rbracket{\S}\\X \notin \X_{\S}}}\ \sum_{\substack{Y \in \Y_{\N\rbracket{\S}\cup C\rbracket{\S}}\\ Y\notin \Y_\S }} p\rbracket{\Y_{\lbrace \N\rbracket{\S} \cup \cluster\rbracket{\S} \rbrace / \lbrace \S \rbrace } |\Y_{\S}} \prod_{\chi \in L\rbracket{\S}} \Lambda_{\X_{\C_{\chi}}|\Y_{\C_{\chi}}|}^{\Y_{\N_{\chi}}} \ ,
\end{align}
which shows, unsurprisingly, that the noise matrix acting on the marginal is state-dependent (via marginal distribution $p\rbracket{\Y_{\lbrace \N\rbracket{\S} \cup \cluster\rbracket{\S} \rbrace / \lbrace \S \rbrace } |\Y_{\S}}$). 
Note that the above equation \emph{does not} include simple matrix multiplication of noise matrices for neighborhoods $\N\rbracket{\S}$.
In other words, the generic marginal noise matrix \emph{is not} a convex combination of the matrices from the set
\begin{align}
\cup_{\cindex \in L\rbracket{\S}} \cbracket{\Lambda^{\Y_{\N_\cindex}}}_{\Y_{\N_\cindex}}
\end{align}
even if the neighborhoods do not overlap.

It is now left to show that the matrix defined in Eq.~\eqref{eq:appendix_marginal_lambda2q} is left-stochastic. 
To this aim, let us sum over string $\X_{\S}$ and obtain
\begin{align*}
\sum_{\X_\S}\Lambda^{(\S)}_{\X_\S|Y_\S} &\coloneqq 
\sum_{\X_\S}\sum_{\substack{X \in C\rbracket{\S}\\X \notin \X_{\S}}}\ \sum_{\substack{Y \in \Y_{\N\rbracket{\S}\cup C\rbracket{\S}}\\ Y\notin \Y_\S }} p\rbracket{\Y_{\lbrace \N\rbracket{\S} \cup C\rbracket{\S} \rbrace / \lbrace \S \rbrace } |\Y_{\S}} \prod_{\chi \in L\rbracket{\S}} \Lambda_{\X_{\C_{\chi}}|\Y_{\C_{\chi}}}^{\Y_{\N_{\chi}}} \ = \\
&=  \sum_{\substack{Y \in \Y_{\N\rbracket{\S}\cup C\rbracket{\S}}\\ Y\notin \Y_\S }} p\rbracket{\Y_{\lbrace \N\rbracket{\S} \cup C\rbracket{\S} \rbrace / \lbrace \S \rbrace } |\Y_{\S}} \sum_{\X_\S} \prod_{\chi \in L\rbracket{\S}} \Lambda_{\X_{\C_{\chi}}|\Y_{\C_{\chi}}}^{\Y_{\N_{\chi}}} \ = \\
&= \sum_{\substack{Y \in \Y_{\N\rbracket{\S}\cup C\rbracket{\S}}\\ Y\notin \Y_\S }} p\rbracket{\Y_{\lbrace \N\rbracket{\S} \cup C\rbracket{\S} \rbrace / \lbrace \S \rbrace } |\Y_{\S}} = 1 \ ,
\end{align*}
where we simply used the fact that the cluster noise matrices are left-stochastic and that the conditional distribution $p\rbracket{\Y_{\lbrace \N\rbracket{\S} \cup C\rbracket{\S} \rbrace / \lbrace \S \rbrace } |\Y_{\S}}$ is normalised.

\subsection{Noise mitigation for marginal distributions -- proof of Proposition~\ref{lem:error_mitigation}\label{sec:app:marginals_proof1}}

Here we analyze in detail the noise mitigation strategy outlined in the main text, focusing on its scalability and effectiveness to recover the ideal marginal distribution. 
Again let us assume that we are interested in marginal on qubits from some subset $\S$.

Let us start by outlining the reasoning behind our choice of the mitigation strategy. 
Recall that, as proved in the previous section, the noisy marginals can always be related to ideal ones by means of a marginal stochastic matrix as for Eq.~\eqref{eqapp:noisemarginal}. 
Hence, taking the inverse of the marginal noise matrix would the natural strategy for correct noise mitigation.
However, recall that, as shown in Eq.~\eqref{eq:appendix_marginal_lambda2q}, the form of the marginal noise matrix depends on the ideal distribution itself, which is generally unknown. 
Therefore one needs an ansatz for such a distribution, which hopefully will work in the mitigation generic case of arbitrary conditional distribution. As indicated in the main text, the natural choice for such ansatz is a uniform distribution (inserted in place of $p\rbracket{\Y_{\lbrace \N\rbracket{\S} \cup \cluster\rbracket{\S} \rbrace \setminus \lbrace \S \rbrace } |\Y_{\S}}$ in Eq.~\eqref{eq:appendix_marginal_lambda2q}). 

Now, let us make the following important observation. 
We are interested in mitigation of the noise of a $|\S|$-qubit marginal distribution of qubits generally belonging to some set of distinct clusters $C\rbracket{\S}$. 
Note that, by definition, a cluster is a set of qubits with correlations in errors so big, that one can not consider the measurement outcomes on them separately.
We argue that, if one is interested in correcting the marginal distribution of only \emph{parts} of the clusters, say $\tilde{p} (\X_{\S})$, it is still a better idea \emph{first} correct the marginal distribution on the whole variables in clusters, namely 
\begin{align}\label{eq:appendix_cluster_noise_dist}
\tilde{p} (\X_{C\rbracket{\S}}) = \sum_{\Y_{C\rbracket{\S}}} p (\Y_{C\rbracket{\S}}) \ \Lambda_{\X_{C\rbracket{\S}} | \Y_{C\rbracket{\S}}}^{\rbracket{C\rbracket{\S}}}  \ .  
\end{align}
Notice that the corrected cluster marginal distribution can always be post-processed to obtain the two-body marginal distribution of interest. 
By proceeding in a similar manner as in the previous section, one can obtain the following expression for the cluster noise matrix
\begin{align}\label{eq:appendix_marginal_lambda_mitigation}
\Lambda_{\X_{C\rbracket{\S}} | \Y_{C\rbracket{\S}}}^{(C\rbracket{\S})}  = 
\sum_{\Y_{\N\rbracket{\S}}} p\rbracket{\Y_{\N\rbracket{\S}}|\Y_{\cluster\rbracket{\S}}}
\prod_{\chi \in L\rbracket{\S}}  \Lambda_{\X_{\cluster_{\chi}}|\Y_{\cluster_{\chi}}}^{\Y_{\N_{\chi}}} \ ,
\end{align}
where $L\rbracket{\S}$ is a set of labels of the clusters $\cluster\rbracket{\S}$ (note that those are also labels for neighbours of those clusters).
Recall that $\N\rbracket{\S}$ was defined as set of qubits from neighbourhoods of clusters $C\rbracket{\S}$ but without including the qubits from $\S$, i.e., $\N\rbracket{\S} = \cup_{\chi \in L\rbracket{\S}} \rbracket{\N_{\chi} \setminus \S}$, so in the above definition we do not average over some qubits in clusters.
As mentioned before, since the exact form of such a matrix requires access to the unknown information of $p\rbracket{\Y_{\N\rbracket{\S}}|\Y_{\cluster\rbracket{\S}}}$, we propose instead to invert the average cluster matrix
\begin{align}\label{eq:appendix_average_cluster_matrix}
 \Lambda^{C\rbracket{\S}}_{av}\coloneqq \frac{1}{2^{|\N\rbracket{\S}|}} \sum_{\Y_{\N\rbracket{\S}}} \Lambda^{\Y_{{\N(\S)}}} \ ,
\end{align}
and use it to perform the error mitigation. 
From our numerical analysis, the replacement of the cluster noise matrix with its average version yields a much better correction than making a similar replacement at the level of the $|\S|$-body marginals (however we tested in only for the case of 2-qubit marginal distributions). 
This provides a clear indication in favor of performing the error mitigation at the level of the clusters instead of the single variables. Moreover, notice that the computational cost of computing the two noise matrices in Eq. \eqref{eq:appendix_marginal_lambda_mitigation} and \eqref{eq:appendix_marginal_lambda2q} is comparable. 
Indeed, both require access to the collections of matrices
$\cbracket{\Lambda^{\Y_{{\N(\S)}}}}_{\Y_{{\N(\S)}}}$. 
Hence, as long as the size of the involved clusters and neighborhoods is reasonably small (particularly, it does not scale with the system size), computing the matrix in Eq. \eqref{eq:appendix_average_cluster_matrix} is efficient.

Due to the above, the following discussion will concern the mitigation on the level of marginals on the level of clusters $C\rbracket{\S}$. 
For clarity, let us from now on use slightly changed notation 
\begin{align}
    \S \rightarrow \S = \cup_{\cindex}\cluster_{\cindex} \ ,
\end{align}
to indicate that we are interested in error-mitigation on the subset of qubits $\S$ which consists of full clusters, i.e., prior to performing error-mitigation we do not wish to take marginal over qubits belonging to the same cluster (while later one can of course marginalize the corrected distributions).

Having obtained $ \Lambda^{\S}_{av}$, one can use its inverse as a correction matrix which can reduce the noise on the marginal distribution, by defining the corrected distribution as
\begin{equation}\label{eq:appendix_corrected_marginal}
\p^{\text{corr}}_{\S} = \left(  \Lambda^{\S}_{av} \right)^{-1} \tilde{\p}_{\S} \, ,  
\end{equation}
where we used vector notation for clarity. 
$\tilde{\p}_{\S}$ is noisy distribution on qubits from $\S$. 
Similarly, the vector without tilde symbol ${\p}_{\S}$ will denote corresponding ideal distribution, i.e., we have $\tilde{\p}_{\S} = \Lambda^{\rbracket{\S}}\p_{\S}$.
Error reduction via average matrix is perfect only in the case \footnote{Clearly, the implicit assumption is that our model of noise is exact.} of infinite-statistics and under the assumption that the actual conditional distribution in Eq.~\eqref{eq:appendix_marginal_lambda_mitigation} is uniform (because then $\Lambda^{\S} = \Lambda^{\S}_{av}$). 
Since this scenario is not realistic, in practice one can hope only for partial noise mitigation, and not its complete reverse.
The errors related to statistical noise when correcting probability distribution were thoroughly analyzed in work \cite{Maciejewski2020mitigation} (we also repeat a similar analysis in the context of expected values of Hamiltonians in next sections).
Here we will analyze the errors which arise due to the fact that we use the inverse of the average cluster matrix $\Lambda^{\rbracket{\S}}_{av}$ instead of inverse of the exact matrix $ \Lambda^{\rbracket{\S}}$, while keeping the assumption of infinite statistics.
We will denote effective correction matrix as $\corr^{\rbracket{\S}}_{av} \coloneqq \rbracket{\Lambda^{\rbracket{\S}}_{av}}^{-1}$ (as in Eq.~\eqref{eq:marginal_correction_average}).

To quantify the errors, we make use the Total-Variation-Distance (see Eq.~\eqref{eq:appendix_TVD}) to measure the distance between probability distributions. For later purposes, let us also recall that the operator norm induced by the vector L1 norm is 
\begin{align}\label{eq:appendix_induced_norm_definition}
||A||_{1\rightarrow1} = \sup_{||\ket{v}||_1=1} ||A\ket{v}||_{1} ,  
\end{align}
where $A$ is any linear operator acting on given vector space.

Now we will proceed to proving Proposition~\ref{lem:error_mitigation} which states that the maximum error due to using average correction matrix is given by
\begin{align}\label{eq:appendix_mitigation_error_bound}
    \text{TVD}\rbracket{\p^{\text{corr}}_{\S},\p^{\S}} \leq \frac{1}{2} ||\corr^{\S}_{av}||_{1\rightarrow1} \max_{\Y_{\N\rbracket{\S}}}||\Lambda^{\S}_{av}-\Lambda^{\Y_{\N\rbracket{\S}}}||_{1\rightarrow 1} \ ,
\end{align}
As a starting point for the proof, let us make following decomposition $\Lambda^{\rbracket{\S}} = \Lambda^{\rbracket{\S}}_{av}  + \rbracket{\Lambda^{\rbracket{\S}} - \Lambda^{\rbracket{\S}}_{av}  }$ so to replace the left-hand-side of Eq. \eqref{eq:appendix_mitigation_error_bound} with
\begin{align}\label{eq:appendix_marginal_errors_step-1}
||\corr^{\rbracket{\S}}_{av}\rbracket{\Lambda^{\rbracket{\S}} \tilde{\p}_{\S}} - \p_{\S}|_{1}  & =
||\corr^{\rbracket{\S}}_{av}\rbracket{\Lambda^{\rbracket{\S}}-\Lambda^{\rbracket{\S}}_{av}} \p_{\S}||_{1} \nonumber \\
& \leq \sup_{\p_{\S}}||\corr^{\rbracket{\S}}_{av}\rbracket{\Lambda^{\rbracket{\S}}-\Lambda^{\rbracket{\S}}_{av}} \p_{\S} ||_{1}  \nonumber \\
& \leq ||\corr^{\rbracket{\S}}_{av}||_{1\rightarrow1}\ ||\Lambda^{\rbracket{\S}}-\Lambda^{\rbracket{\S}}_{av}||_{1\rightarrow1}   \ ,
\end{align}
where the first inequality indicates the maximization over all possible probability distributions over the set $\S$ variables and the second inequality follows from Eq.~\eqref{eq:appendix_induced_norm_definition} and the sub-multiplicativity of the L1 norm.
Now, let us bound the second term of the above equation as follows
\begin{align}\label{eq:appendix_marginal_errors_step0}
||\Lambda^{\rbracket{\S}}-\Lambda^{\rbracket{\S}}_{av}||_{1\rightarrow1} &=
||
\rbracket{\sum_{\Y_{\N\rbracket{\S}}} p\rbracket{\Y_{\N\rbracket{\S}}|\Y_{{\S}}}
\prod_{\chi \in L\rbracket{\S}}  \Lambda_{\X_{\cluster_{\chi}}|\Y_{\cluster_{\chi}}}^{\Y_{\N_{\chi}}} }-\Lambda^{\rbracket{\S}}_{av}||_{1\rightarrow 1} \nonumber \\
& \leq 
\sup_{\cbracket{p\rbracket{\Y_{\N\rbracket{\S}}|\Y_{{\S}}}}}
||
\rbracket{\sum_{\Y_{\N\rbracket{\S}}} p\rbracket{\Y_{\N\rbracket{\S}}|\Y_{{\S}}}
\prod_{\chi \in L\rbracket{\S}}  \Lambda_{\X_{\cluster_{\chi}}|\Y_{\cluster_{\chi}}}^{\Y_{\N_{\chi}}} }-\Lambda^{\rbracket{\S}}_{av}||_{1\rightarrow 1} \ ,
\end{align}
where the above sequence of inequalities removes the dependence on the unknown conditional distribution $p\rbracket{\Y_{\N\rbracket{\S}}|\Y_{{\S}}}$.

The maximization in Eq.~\eqref{eq:appendix_marginal_errors_step0} corresponds to the maximal distance of the average matrix $\Lambda^{\rbracket{\S}}_{av}$ from the convex hull of the set of matrices $\cbracket{\Lambda^{\Y_{\N\rbracket{\S}}}}_{\Y_{\N\rbracket{\S}}}$ (hence over a polytope with extremal points being each of the matrices $\Lambda^{\Y_{\N\rbracket{\S}}}$).
It is straightforward to see that this function is convex, due to the triangle inequality and absolute homogeneity of the L1 norm. 
From this fact, it follows that the maximum in Eq.~\eqref{eq:appendix_marginal_errors_step0} is attained for one of the extremal points of the convex hull, i.e., a particular $\Lambda^{\Y_{\N_i}} \Lambda^{\Y_{\N\rbracket{\S}}}$. 
Hence, combining this last bound with Eq. \eqref{eq:appendix_marginal_errors_step-1} one obtains exactly Eq. \eqref{eq:appendix_mitigation_error_bound}, which proves Proposition~\ref{lem:error_mitigation}.

\subsection{Statistical error bounds \label{sec:app:marginals_statistical}}
	
The idea behind the energy estimation routine in variational algorithms such as QAOA is to estimate Hamiltonian with $K$-terms by estimating those terms separately and then adding them up.
However, if one provides statistical error bounds for each of those terms, one needs to take into account the probabilistic nature of such bounds.
Let us say that we want to estimate some distribution $\p$ with $N$ outcomes by sampling from it $s$ number of times.
It is well known \cite{Weissman2003} that the probability  of estimated distribution being $\epsilon$-close to the true distribution $\p$ in terms of TVD (Eq.~\eqref{eq:TVD} is vanishing exponentially in number of sample
\begin{align}\label{app:eq:probability_statistical}
    \text{Pr}\rbracket{\text{TVD}\rbracket{\p,\p^{est}}\geq \epsilon} \leq \rbracket{2^n-2}\exp\rbracket{-2s\epsilon^2} \ .
\end{align}
Now it is often convenient to set the probability fixed and consider the confidence intervals, i.e., the bound for $\text{TVD}\rbracket{\p,\p^{est}}$. 
Then we can rewrite the above equation as a function $\epsilon^*\rbracket{n,s,P_{1, \text{err}},1}$ of three fixed parameters -- number of outcomes $n$, number of samples $s$ and probability of the upper bound being incorrect $P_{1, \text{err}}$.
Then the basic manipulations of Eq.\eqref{app:eq:probability_statistical} give
\begin{align}\label{eq:statistical_bound_1}
   \text{TVD}\rbracket{\p,\p^{est}} \leq \epsilon^*\rbracket{n,s,P_{1, \text{err}},1} = \sqrt{\frac{\log\rbracket{2^n-2}-\log{\rbracket{P_{1, \text{err}}}}}{2s}}\ .
\end{align}

However, since in estimation of the K-term Hamiltonian one combines the upper bounds of the form \eqref{eq:statistical_bound_1} for particular terms into upper bound $\epsilon^*\rbracket{n,s,P_{1, \text{err}},K}$ for the whole Hamiltonian, one needs to make sure that all upper bounds for particular terms are true at the same time with the desired probability. 
The union bound states that the probability of at least one event from some set occurring is no greater than the sum of probabilities of particular events.
In our case, the interesting set consists of events of the type ``one of the bounds of type \eqref{eq:statistical_bound_1} is not satisfied''.
Hence the probability of at least one bound being wrong is upper bounded by
\begin{align}
    P_{K,\text{err}} \leq  K P_{1,\text{err}}\ .
\end{align}
Therefore, if we wish to ensure that probability of all the bounds being right at the same time is fixed and equal to $P_{1,\text{err}}$ for fixed number of samples $s$, we need to effectively increase the upper bound to
\begin{align*}\label{app:eq:statistical_bound_K}
    &\epsilon^*\rbracket{n,s,P_{1, \text{err}},K} = \\
    \\
    &=\sqrt{\frac{\log\rbracket{2^n-2}-\log{\rbracket{\frac{P_{1, \text{err}}}{K}}}}{2s}} = \\
    &=\sqrt{\frac{\log\rbracket{2^n-2}-\log{\rbracket{P_{1, \text{err}}}+\log{\rbracket{K}}}}{2s}} \numberthis
\end{align*}
with additional term $\log{K}$ under the square root. 
In other words, we effectively lower the probability of error occurring in estimation of particular marginal distributions by a factor of $K$, which ensures that simultaneous estimation of $K$ marginal distributions has the precision from Eq.~\eqref{app:eq:statistical_bound_K} with probability not lower than the initial $P_{\text{1,err}}$ (by the virtue of union bound).

It is instructive to look at this from the perspective of sampling complexity. 
Let's say that we would like to increase number of samples 
\begin{align*}
    s \rightarrow \tilde{s} \eqqcolon s\rbracket{1+f_{\text{oh}}}
\end{align*}
in such a way, that upper bound remains fixed, i.e.,
\begin{align*}
    \epsilon^*\rbracket{n,\tilde{s},P_{1, \text{err}},K} = \epsilon^*\rbracket{n,s,P_{1, \text{err}},1} \ .
\end{align*}
Simple calculations show that in this case we would need to (multiplicatively) increase the number of samples $s$ by the overhead equal to
\begin{align*}
    1+f_{\text{oh}}=1+C\log{\rbracket{K}} 
\end{align*}
where parameter $C$ is equal to
\begin{align}
    \rbracket{\log{\rbracket{\frac{2^n-2}{P_{1,\text{err}}}}}}^{-1} \ .
\end{align}
Hence we see that for fixed dimension (fixed number of outcomes) simultaneous estimation of $K$ Hamiltonian terms leads to sampling overhead logarithmic in $K$.

Note that if one fixes the initial sampling size $s$, probability of error $P_{\text{1,err}}$ and number of simultaneously estimated marginals $K$, the sampling overhead is, perhaps counterintuitively, decaying in the number of outcomes. 
This dependence is roughly linear in the number of outcomes, hence it is exponential in the number of qubits.

\subsection{Proofs of energy error bounds \label{sec:app:marginals_energy}}
Here we will provide proofs of worst-case error bounds for energy estimation given in the main text. 
Note that in the following discussion, for clarity of notation we won't add a special index to marginal probability distributions indicating that they are marginals. 
This should be clear from the context.
On the contrary, local Hamiltonian terms, noise matrices, and correction matrices will have an additional label ($\alpha$) indicating their locality.

\subsubsection{Approximation errors}
We will start by providing proof of Eq.~\eqref{eq:additive_bound} which gives additive bound for possible deviations of error-mitigated energy from ideal one when the source of deviations is an approximation used in constructing noise model (recall Eq.~\eqref{eq:marginal_lambda_average}). 
Consider estimation of the local term of Hamiltonian $H_{\alpha}$. 
Recall that locality of the terms means that $H_{\alpha}$ acts non-trivially only on the subset of qubits.
Now let us denote by $\p$ the marginal probability distribution on the qubits belonging to the relevant subset.
Then expected value of the local term can be written as
\begin{align}\label{app:eq:local_term}
    \tbracket{H_{\alpha}} = \sum_{t} \mathbf{\lambda}^{\alpha}_t p_{t} = \braket{\mathbf{\lambda}^{\alpha}}{\p} ,
\end{align}
where by $\mathbf{\lambda}^{\alpha}$ we denoted vector of eigenvalues  of local term $H_{\alpha}$, and we used convenient braket notation to denote scalar product.
For example, if $H_{\alpha} = \sigma^z \otimes \sigma^z$, then $\lambda^{\alpha} = \rbracket{1,-1,-1,1}$.
Now we want to consider two different estimators of the same local term -- one from ideal distribution $\p^{\text{ideal}}$ and second from the error-mitigated noisy distribution 
$\p^{\text{corr}}=C_{av} \p^{noisy}$ (recall Proposition~\ref{lem:error_mitigation}).
Now we want to upper bound the difference between energy estimators based on those two marginal distributions.
From Eq.~\eqref{app:eq:local_term} we can write
\begin{align}\label{app:eq:proof_approx}
    |\tbracket{H^{\text{corr}}_{\alpha}} - \tbracket{H^{\text{ideal}}_{\alpha}}| = |\bra{\mathbf{\lambda}}\rbracket{\ket{\p^{\text{corr}}}-\ket{\p^{\text{ideal}}}}| \leq \underbrace{\max_{t}|\lambda_t|}_{||H_{\alpha}||} \bra{\iden} \underbrace{|\rbracket{\ket{\p^{\text{corr}}}-\ket{\p^{\text{ideal}}}}|}_{2\text{TVD}\rbracket{\p^{\text{corr}},\p^{\text{ideal}}}} \leq 2 ||H_{\alpha}|| \delta_{\alpha} \ ,
\end{align}
where $\delta^{\alpha}$ is approximation error defined in Eq.~\eqref{eq:approx_error_definition}, $\bra{\iden}$ is a vector of ones and the last inequality follows from Lemma~\ref{lem:error_mitigation}.
Factor two comes from the fact that Total-Variation Distance is defined with $\frac{1}{2}$ factor.
The additive error bound from Eq.~\eqref{eq:additive_bound} is just a multiple application of triangle inequality
\begin{align}\label{app:eq:additive:approx}
    |\tbracket{\H^{\text{corr}}}-\tbracket{\H}|   = |\sum_{\alpha}\tbracket{H^{\text{corr}}_{\alpha}} - \tbracket{H^{\text{ideal}}_{\alpha}}| \leq \sum_{\alpha} |\tbracket{H^{\text{corr}}_{\alpha}} - \tbracket{H^{\text{ideal}}_{\alpha}}| \leq 2\sum_{\alpha}\delta^{\alpha}||H_{\alpha}|| \ .
\end{align}

\subsubsection{Statistical errors}
Now we will proceed to prove Eq.~\eqref{eq:additive_bound} which bounds the effects of statistical noise on the error-mitigation. 
First, let us assume that the noise matrix acting on the marginal of interest is known exactly (not approximately as in the above derivations).
In that case, in the lack of presence of the statistical errors, the correction is done exactly.
However, in reality the statistics are finite, and the \textit{estimator} of the noisy marginal distribution  $\p^{noisy}$ from which we sample on imperfect detector can be formally written as 
\begin{align}\label{app:eq:noisy_estimation}
     \underbrace{\p^{noisy}}_{\text{true}} \rightarrow \underbrace{\p^{noisy} + \mathbf{\epsilon}}_{estimated}\ .
\end{align}
Now we want to bound the error in the energy estimated after acting by correction matrix $\Lambda_{\alpha}^{-1}$ on the estimated distribution above.
In analogy to Eq.~\eqref{app:eq:proof_approx} we can write
\begin{align*}\label{app:eq:proof_stats}
        |{H^{\text{corr, est}}_{\alpha}} - \tbracket{H^{\text{ideal}}_{\alpha}}| &= |\bra{\mathbf{\lambda}}\rbracket{\underbrace{\Lambda_{\alpha}^{-1}\ket{\p^{noisy}}}_{\ket{\p^{\text{ideal}}}}+\Lambda_{\alpha}^{-1}\ket{\mathbf{\epsilon}}-\ket{\p^{\text{ideal}}}}| =|\braket{\mathbf{\lambda}}{\Lambda_{\alpha}^{-1}\ket{\mathbf{\epsilon}}}| \leq\\
        &\leq ||H_{\alpha}||\  |\bra{\iden}{\Lambda_{\alpha}^{-1}\ket{\mathbf{\epsilon}}}| \leq ||H_{\alpha}||\ || \Lambda_{\alpha}^{-1}||_{1\rightarrow1}\ ||\ket{\mathbf{\epsilon}}||_1\leq \\
        &\leq 2 ||H_{\alpha}||\ || \Lambda_{\alpha}^{-1}||_{1\rightarrow1} \epsilon^* \ ,        \numberthis
\end{align*}
where first inequality follows from definition of operator norm, second inequality follows from a definition of any induced operator norm and the last inequality follows from analysis in previous section with $\epsilon^*$ being bound given by Eq.~\eqref{app:eq:statistical_bound_K}.
In analogy to Eq.~\eqref{app:eq:additive:approx} by applying multiple times triangle inequality we obtain Eq.~\eqref{eq:additive_bound}.

\subsubsection{Approximation and statistical errors}

Now we will combine two previous bounds by using a triangle inequality.
Let us denote by $\p^{\text{noisy, est}}$ the estimator of noisy probability distribution $\p^{\text{noisy}}$ on the marginal $\alpha$ (hence the RHS of Eq.~\eqref{app:eq:noisy_estimation}), and by $C_{av}^{\alpha}$ the average correction matrix used to correct that marginal.
We want to bound the distance between corrected estimator of noisy distribution $C^{\alpha}_{av}\p^{\text{noisy, est}}$ and the ideal distribution $\p^{\text{ideal}}$.
For particular marginal we have triangle inequality for Total-Variation Distance
\begin{align*}\label{app:eq:bound_approx_stats}
    \text{TVD}\rbracket{C^{\alpha}_{av}\p^{\text{noisy, est}},\p^{\text{ideal}}} &\leq \underbrace{\text{TVD}\rbracket{C_{av}\p^{\text{noisy}},C_{av}\p^{\text{noisy, est}}}}_{\leq 2||C_{av}||_{1\rightarrow1}\epsilon^*}+\underbrace{\text{TVD}\rbracket{C_{av}\p^{\text{noisy}},\p^{\text{ideal}}}}_{2\delta_{alpha}}\\
    &\leq    2||C^{\alpha}_{av}||_{1\rightarrow1}\epsilon^*+ 2\delta_{\alpha}    \numberthis \ ,
\end{align*}
where $\epsilon^*$ is statistical bound given by Eq.~\eqref{app:eq:statistical_bound_K}, and $\delta_{\alpha}$ is approximation error defined in Eq.~\eqref{eq:approx_error_definition}.
First underlined inequality follows from properties of induced operator norm and from statistical errors bound proved in the previous section (this step is analogous to the one in the second line of Eq.~\eqref{app:eq:proof_stats}).
The second underlined inequality follows from Proposition~\ref{lem:error_mitigation}.
To translate the above result to expected values of local Hamiltonian one just repeats the reasoning given in proofs of Eq.~\eqref{app:eq:additive:approx} and Eq.~\eqref{app:eq:proof_stats} for marginal distributions of the form $C^{\alpha}_{av}\p^{\text{noisy, est}}$ for which we have bounds of the form given by Eq.~\eqref{app:eq:bound_approx_stats}.

To finish this section, let us note that the reasoning presented in this section is analogous to the one given in Ref.~\cite{Maciejewski2020mitigation} where analysis of effects of statistical noise and non-classical noise on the error-mitigation performed on global distributions was presented.

\section{Details of Diagonal Detector Overlapping Tomography\label{sec:app:ddot}}

In this section, we give more details regarding our noise characterization procedure using DDOT.
We start by providing efficient way of construction of DDOT circuits in Section.~\ref{sec:app:ddot_construction}. 
In Appendix~\ref{sec:app:ddot_inference} we show how to use results of DDOT to infer the correlations structure of readout noise in a device.
Then we discuss the construction of Diagonal Detector Overlapping Tomography circuits which are balanced (Appendix~\ref{sec:app:ddot_balanced}) and perfect (Appendix~~\ref{sec:app:ddot_perfect}), together with proofs fro scaling of required number of random circuits. 
In Appendix~~\ref{sec:app:ddot_heuristic}, we discuss heuristic procedures of making DDOT circuits more balanced. 
Finally, in Appendix~~\ref{sec:app:ddot_overestimation} we explain the effect of overestimating correlations which can happen if the implemented DDOT collection is not balanced.

	\subsection{Construction of Diagonal Detector Overlapping Tomography circuits\label{sec:app:ddot_construction}}

	Here we will present in detail our building block for efficient characterization of correlations in a measurement device: Diagonal Detector Overlapping Tomography (DDOT). 
	In analogy to Quantum Overlapping Tomography (QOT) \cite{Cotler2020quantum}, a collection of $\rbracket{N,k}$ DDOT circuits allows one to reconstruct the noise matrices for the readout process of any group of $k$ qubits in an $N$-qubit device. 

	More precisely, we define a $\rbracket{N,k}$ perfect collection of DDOT circuits as a collection of $N$-qubit quantum circuits constructed only from $\iden$ and $X$ gates, with the property that for every subset of qubits of the size $k$, each of the computational basis states on that subset is prepared at least once.
	For example, we will call a collection of circuits $\rbracket{N,2}$ perfect if for all pairs of qubits, each of the states from  $\cbracket{\ket{00},\ket{01},\ket{10},\ket{11}}$ is prepared at least once in the whole collection.
    
    One way to construct DDOT collection is to follow Ref.~\cite{Cotler2020quantum} and make use of the notion of hash functions (here by a hash function we mean every function $\sbracket{N} \rightarrow \sbracket{k}$ with $k<N$).
	To construct $\rbracket{N,k}$ perfect collection of DDOT circuits using hash functions one can use Algorithm~\ref{alg:random_generation}, which encapsulates the idea of Quantum Overlapping Tomography from Ref.~\cite{Cotler2020quantum} translated to the construction of DDOT circuits.
	Specifically, this method corresponds to QOT with two ``bases'' which are preparation of state $\ket{0}$ or $\ket{1}$.
	Each hash function assigns each qubit a label from $\sbracket{k}$.
	For a given function, qubits are assigned to $k$ disjoint batches based on the value of the function value. For fixed assignment of batches, qubits belonging to a batch are initialized in the same state ($\ket{0}$ or $\ket{1}$), independently from the qubits in other batches.
    In this way, $2^k$ circuits are specified, which independently implement all computational basis states on all the batches.
    For example, in the case of $k=2$ and $N=6$ qubits, some specific hash functions could result in initializing the following states -- $\cbracket{000000,000111,111000,111111}$, which implement all two-qubit computational basis states on the pairs of qubits from left and right parts of the register (in this example the two batches are $\cbracket{Q0,Q1,Q2}$ and $\cbracket{Q3,Q4,Q5}$).
	The DDOT circuits collection constructed in this way is perfect if the underlying collection of hash functions is perfect.
	This means that for \textit{each} $k$-qubit subset in $N$-qubit device, there exists at least one hash function in the collection which assigns each qubit from that subset a distinct number from $\sbracket{k}$.
	Note that if a given function is indeed injective on some $k$-qubit subset, this means that qubits from those subsets belong to distinct batches and therefore all computational-basis states will be implemented on them.
	Hence if this holds for all subsets, the collection is perfect.
	In Algorithm~\ref{alg:random_generation} we generate \textit{random} hash functions to create a DDOT collection, therefore there is no deterministic guarantee that the collection constructed in this way will indeed be perfect.
	However, it follows directly from arguments in \cite{Cotler2020quantum} (specifically, Section III) that if we use Algorithm~\ref{alg:random_generation} to generate DDOT collection, then if we wish the collection to be perfect with probability at least $1-\delta$, the needed number circuits $\kappa$ is at least of the order 
	\begin{align}\label{eq:scaling_hash_functions}
	    \kappa> {\rbracket{2e}^k\rbracket{\log\rbracket{\noq}+\frac{1}{k}\log\rbracket{\frac{1}{\delta}}}}\ ,
	\end{align}
	see Appendix~\ref{sec:app:ddot_perfect} for detailed bounds.

	On the other hand, one can also consider construction which uses random circuits without referring to the notion of hash functions.
	Generation of random bitstrings and using them as definitions of quantum circuits is used to construct DDOT collection in Algorithm~\ref{alg:random_generation_circuits}.
	Similarly to Algorithm~\ref{alg:random_generation_circuits}, since we use randomness to generate the collection, there is no deterministic guarantee that the resulting collection will be perfect.
	However, in Appendix~\ref{sec:app:ddot_perfect} we show that if we wish the collection to be perfect with probability at least $1-\delta$, then the needed number of circuits $\kappa$ is at least of the order
	\begin{align}\label{eq:scaling_random_circuits}
	    \kappa > {2^{k}\rbracket{k\log\rbracket{2N}+\log{\frac{1}{\delta}}}} \ ,
	\end{align}
    which looks similar to Eq.~ \eqref{eq:scaling_hash_functions}, however in practice exhibits better scaling due to specific factors (see Appendix~\ref{sec:app:ddot_perfect} for more detailed bounds).  
	Hence we expect that for higher system sizes the random circuits algorithm should perform better in practice compared to the one which uses hash functions.
	Of course, when the collection with desired properties has been constructed, it does not matter what method was used to create it.
	
	As mentioned in the main text, during the implementation of DDOT, different circuits will in general be sampled a different number of times, this may cause some correlations to be overestimated (this effect can be reduced by proper post-processing of the data -- see Appendix~\ref{sec:app:ddot_overestimation}). 
Hence it is beneficial for the perfect collection of DDOT circuits to be \textit{approximately balanced}. 
	Here by 'balanced' we mean that all $k$-qubit computational basis states are sampled the same number of times, and by 'approximately balanced' that they are sampled approximately the same number of times (see related notion for hash functions, for example in Ref~\cite{alon2008balanced}). 
	For the construction using random circuits, we prove in Appendix~\ref{sec:app:ddot_balanced} that one can expect that with high probability the collection will also have this property with the required number of circuits scaling similarly to Eq.~\eqref{eq:scaling_random_circuits}.

    To conclude, let us point out that for a given pair of numbers $\rbracket{N,k}$ it suffices to generate the DDOT collection only once, and it can be used in the design of future experiments. 
    We will make a number of pre-generated collections publicly available in our GitHub repository QREM (Quantum Readout Errors Mitigation) \cite{qrem}.
    We note that the described techniques are suitable for noise characterization not only for the noise model we proposed but also for other models with bounded locality of correlations, for example with two-qubit correlations considered in Ref.~\cite{Bravyi2020mitigating}.

	\subsection{Inferring the structure of clusters and neighborhoods\label{sec:app:ddot_inference}}

After the implementation of a $\rbracket{N,k}$ perfect collection of DDOT circuits, one has potential access to a lot of information about the measurement noise in a device. 
Here we provide a method of using data from such implementation to infer the correlations in measurement noise.
As a starting point, let us note that one can use the output of the DDOT circuits to construct all the possible two-qubit noise matrices $\Lambda_{X_i X_j | Y_i Y_j}$. 
	In particular, whenever $k\geq2$, there are subsets of circuits that implement all computational-basis states for each two-qubit subsystem.
	This allows to gather all two-qubit marginal probability distributions of the form
	\begin{align}
	    \cbracket{p\rbracket{X_iX_j|Y_iY_j}}_{i,j}\, 
	\end{align}
	where $X_i$ ($X_j$) is the measured state of $i$th ($j$th) qubit, and the $Y_i$ ($Y_j$) is the input state of the corresponding qubit (i.e., the state that is supposed to be implemented by the quantum circuit from the DDOT collection).
	In the following Example~\ref{ex:generic_2q_lambda} we give explicitly calculated single-qubit noise-matrices for exemplary pair of qubits $i=1$ and $j=2$.
	\begin{example}\label{ex:generic_2q_lambda}
		For two-qubit set $S=\cbracket{1,2}$, assuming no dependence on the state of neighbors, generic left-stochastic map acting on its detector has a form
		\begin{align}
		\Lambda^{S} = \begin{pmatrix}
		p\rbracket{00|00} & \quad p\rbracket{00|01} &\quad p\rbracket{00|10} &\quad p\rbracket{00|11}\vspace{0.1cm}\\
		p\rbracket{01|00} &\quad p\rbracket{01|01} &\quad p\rbracket{01|10} &\quad p\rbracket{01|11}\vspace{0.1cm}\\
		p\rbracket{10|00} &\quad p\rbracket{10|01} &\quad p\rbracket{10|10} &\quad p\rbracket{10|11}\vspace{0.1cm}\\
		p\rbracket{11|00} &\quad p\rbracket{11|01} &\quad p\rbracket{11|10} &\quad p\rbracket{11|11}
		\end{pmatrix}\ .
		\end{align}
		It follows that single-qubit noise matrices from Eq.~\eqref{eq:correlations_pair} for first qubit have form
		\begin{align}
		\Lambda^{Y_2 = '0'} &= \frac{1}{2}\begin{pmatrix}
		p\rbracket{00|00}+p\rbracket{01|00}   &\quad p\rbracket{00|10}+p\rbracket{01|10} \vspace{0.1cm} \\
		p\rbracket{10|00} +p\rbracket{11|00}  &\quad p\rbracket{10|10}+p\rbracket{11|10}
		\end{pmatrix}\ ,  \\
		\Lambda^{Y_2 = '1'} &= \frac{1}{2}\begin{pmatrix}
		p\rbracket{00|01}+p\rbracket{01|01}   &\quad p\rbracket{00|11}+p\rbracket{01|11} \vspace{0.1cm} \\
		p\rbracket{10|01} +p\rbracket{11|01}  &\quad p\rbracket{10|11}+p\rbracket{11|11}
		\end{pmatrix}\, \\
		\end{align}
		while for second qubit they are
		\begin{align}
		\Lambda^{Y_1 = '0'} &= \frac{1}{2}\begin{pmatrix}
		p\rbracket{00|00}+p\rbracket{10|00}   &\quad p\rbracket{00|01}+p\rbracket{10|01} \vspace{0.1cm} \\
		p\rbracket{01|00} +p\rbracket{11|00}  &\quad p\rbracket{01|01}+p\rbracket{11|01}
		\end{pmatrix}\ ,  \\
		\Lambda^{Y_1 = '1'} &= \frac{1}{2}\begin{pmatrix}
		p\rbracket{00|10}+p\rbracket{10|10}   &\quad p\rbracket{00|11}+p\rbracket{10|11} \vspace{0.1cm} \\
		p\rbracket{01|10} +p\rbracket{11|10}  &\quad p\rbracket{01|11}+p\rbracket{11|11}
		\end{pmatrix} \ .
		\end{align}
	\end{example}
	Note that the idea is to simply fix the state of 'neighbouring' qubit to be either $'0'$ or $'1'$ and calculate according conditional marginals.

Now, we propose to use the information about the two-qubit noise matrices to calculate the correlations (Eq.~\eqref{eq:correlations_pair}) between a given pair of qubits. For the $i$-th qubit, it's the dependence from the $j$-th qubit can be calculated by constructing two single-qubit matrices on $i$-th qubit -- the first one with the condition that $j$-th qubit was initialized in $\ket{0}$ state, and the second in $\ket{1}$ state.

Those matrices can be used to calculate the parameter $c_{j\rightarrow i}$ as the norm of a difference of those matrices (Eq.~\eqref{eq:correlations_pair}).
We then propose to infer the structure of the readout correlations according to the magnitude of the parameters $c_{j\rightarrow i}$ as follows. 
One specifies threshold parameters $\delta_{\text{clust}}$ and  $\delta_{\text{neighb}}$ for the level of correlations between qubits, and then assigns qubit $j$ to the neighborhood or to the cluster of qubit $i$, if the parameter $c_{j\rightarrow i}$ is greater than the respective threshold.
In general, we advise to set those thresholds to be higher than the likely effects of statistical deviations.
We outline the whole procedure in Algorithm~\ref{alg:clusters}, which takes as input the conditional single-qubit noise matrices $\cbracket{\Lambda^{Y_j}_{Q_i}}_{i\neq j}$ together with a set of thresholds, and outputs the structure of the clusters and neighborhoods in a device.

Finally, let us note that the above-described inference of correlations from two-qubit marginal distributions works under the assumption that the correlations do not vanish under taking marginals over other (than given pair) qubits.
This does not need to be true in practice, and one can consider generalizations of Algorithm~\ref{alg:clusters}.
The analogous set of single-qubit matrices depending on states of $t$ neighboring qubits would be created in a fully analogous manner to that of Example~\ref{ex:generic_2q_lambda} but now one would need to fix the state of $t$ qubits.
Note that if one implemented DDOT $\rbracket{N,k}$ collection, the data to create such matrices for $t=k-1$ is available from the experiments.
In the future, we intend to investigate more elaborate methods of inferring correlations structure using DDOT, and we will accordingly expand our repository \cite{qrem}.

\begin{@twocolumnfalse}
\begin{algorithm*}
\small
 \caption{Generation of a perfect collection of $\rbracket{N,k}$ DDOT quantum circuits using random hash functions\label{alg:random_generation}$\qquad \ $}
  \RaggedRight\textbf{Input}:\\
  \RaggedRight $L$: number of hash functions\\
  \RaggedRight $N$: number of all qubits\\
  \RaggedRight $k$: size of the marginal\\
\begin{enumerate}
    \settowidth\tablen{$\quad$}
    \item Start collection by creating two circuits which prepare all qubits in $\ket{0}$ state and in $\ket{1}$ state.
    \item Generate $L$ random hash functions $\cbracket{f_1,\dots,f_L}$, i.e., random mappings $f_l: \sbracket{N}\rightarrow \sbracket{k}$.
	\item Define $k$-qubit sub-register as a set of all bit-strings $\cbracket{\X^1,\dots,\X^{2^k}}$ of length $k$.
	\item \textbf{For} each function $f_l$, \textbf{do}:\\
    \tabbox{$\quad$}  \textbf{For} each bitstring $\X^i$ in $k$-qubit register, \textbf{do}:\\
    \tabbox{$\quad$}\tabbox{$\quad$} Define string \textbf{Y} of length $N$ in the following way:
    \begin{align*}
        \rbracket{\textbf{Y}}_{\textbf{j}} =  \rbracket{\X^i}_{f_l\rbracket{\textbf{j}}}
        \end{align*}
    \tabbox{$\quad$}\tabbox{$\quad$}\mbox{ Save \textbf{Y} as definition of one quantum circuit in the} DDOT collection -- each symbol '0' corresponds to identity \tabbox{$\qquad$ $\quad$}\tabbox{$\qquad$ $\quad$} gate, each '1' symbol corresponds to NOT gate. 
    \item Check if generated set of $s=2+L(2^k-2)$ bitstrings $\cbracket{\Y}$ contains all combinations of $k$-length subsets of symbols '0' and '1'. 
    \\\textbf{If} yes:\\
        \tabbox{$\quad$}family $\cbracket{\Y}$ defines perfect DDOT collection,
    \\\textbf{If} not:\\
        \tabbox{$\quad$}set $L\rightarrow L+1$, generate new random hash function, add it to the collection and perform steps 2, 3 and 4.
\end{enumerate}
\end{algorithm*}
\begin{algorithm*}
\small
  \caption{\centering Generation of $\rbracket{N,k}$ perfect collection of DDOT quantum circuits using random circuits\label{alg:random_generation_circuits}$\qquad\ \qquad \ \qquad   $}
  \RaggedRight\textbf{Input}:\\
  \RaggedRight $s$: number of circuits\\
  \RaggedRight $N$: number of all qubits\\
  \RaggedRight $k$: size of the marginal\\
\begin{enumerate}
    \settowidth\tablen{$\quad$}
    \item Start collection by creating two circuits which prepare all qubits in $\ket{0}$ state and in $\ket{1}$ state.
    \item Generate $s$ random bitstrings of size $N$. 
    Each bitstring is a definition of quantum circuits  -- symbol '0' corresponds to identity gate, and symbol '1' symbol corresponds to NOT gate. 

    \item Check if generated set of $s$ bitstrings $\cbracket{\Y}$ contains all combinations of $k$-length subsets of symbols '0' and '1'. 
    \\\textbf{If} yes:\\
        \tabbox{$\quad$}family $\cbracket{\Y}$ defines perfect DDOT collection,
    \\\textbf{If} not:\\
        \tabbox{$\quad$}set $s\rightarrow s+1$, generate new random bitstring, add it to the collection and check new collection.
\end{enumerate}
\end{algorithm*}
\begin{algorithm*}
\small
 \caption{Assignment of qubits to clusters and neighborhoods\label{alg:clusters}}
\RaggedRight\textbf{Input}:\\
\RaggedRight $\S_{\text{pairs}}$: set of all pair indices without repetitions of indices\\
\RaggedRight $\delta_{\text{clust}}$: threshold of correlations for assignment to clusters\\
\RaggedRight $\delta_{\text{neighb}}$: threshold of correlations for assignment to neighborhoods\\
\vspace{0.15cm}
\RaggedRight $\cbracket{\Lambda_{X_i | Y_i }^{Y_j}}_{i,j}$: collection of single-qubit noise matrices depending on state of single neighbour for all pairs.\\
\vspace{0.15cm}
\begin{algorithmic}
\For {$(i,j)$ in $\S_{\text{pairs}}$}
	    \State Calculate $c_{j\rightarrow i}$ and $c_{i\rightarrow j}$  as 
	    \begin{align*}
	        c_{j\rightarrow i} = \frac{1}{2}||\Lambda_{Q_i }^{Y_j = '0'}-\Lambda_{Q_i }^{Y_j = '1'}||_{1\rightarrow1}\ , \quad
	        c_{i\rightarrow j} = \frac{1}{2}||\Lambda_{Q_j }^{Y_i = '0'}-\Lambda_{Q_j }^{Y_i = '1'}||_{1\rightarrow1}\ ,
	    \end{align*}
	     \State where $Y_j$ denotes input state of qubit $j$ (see Eq.~\eqref{eq:correlations_pair}).  
	    \State \ 
	    \If{$ c_{j\rightarrow i}>\delta_{\text{clust}}$ \textbf{or} $ c_{i\rightarrow j}>\delta_{\text{clust}}$}
	        \State Assign qubits $i$ and $j$ to the same cluster.
	  \State \ 
	    \Else
	        \If{$ c_{j\rightarrow i}>\delta_{\text{neighb}}$}
	           \State Assign qubit $j$ to the neighborhood of qubit $i$.
	        \EndIf
	        \If{$ c_{i\rightarrow j}>\delta_{\text{neighb}}$}
	           \State Assign qubit $i$ to the neighborhood of qubit $j$.
	        \EndIf
	  \EndIf
	 \EndFor
\end{algorithmic} 
\end{algorithm*}
\end{@twocolumnfalse}

Once a model for the above structure is obtained, one can use the rest of the data obtained in the DDOT procedure to reconstruct the cluster noise matrices as a function of the state of the neighbors and consequently construct a global noise model (Eq.~\eqref{eq:noise_model_correlated}), as well as the correction matrices for the marginals (Eqs.~\eqref{eq:marginal_lambda_average},\eqref{eq:marginal_correction_average}).
From the definition of a perfect $\rbracket{N,k}$ collection, it follows that one can reconstruct only cluster noise matrices involving a number $t = |\cluster_{\cindex}|+ |\N_{\cindex}|$ of qubits which does not exceed $k$.
If it happens that in fact $t>k$, we propose to implement one of two following solutions:
 \begin{enumerate}
     \item Neglect correlations between some qubits in order to enforce that $t=k$.
     For example, in Algorithm~\ref{alg:clusters} if the number of qubits assigned to the neighborhood of some cluster exceeds the limit, one could decide to \textit{not assign} to that neighborhood the qubits with the lowest values of $c_{j\rightarrow i}$ parameters.
     This will result in an imperfect model, which might nevertheless be accurate enough for the purposes of error mitigation.
     \item Refine the noise model by performing additional experiments of standard Diagonal Detector Tomography on a chosen subset.
 \end{enumerate}

We note that for the refinement of the noise model, there also exists the alternative possibility of constructing a \emph{restricted} DDOT collection that implements the characterization on a \emph{specific} set of $t$-length subsets (as opposed to all such subsets). 
We will analyze this problem in more detail in future works, as well as during the development of our GitHub repository \cite{qrem}.

To finish this subsection, we note that Algorithm~\ref{alg:clusters}, while being straightforward and efficient, is not a flawless method. 
For example, it might happen that the noise on qubit $i$ highly depends on the joint state of $j$ and $l$, but this dependence is much lower if one considers qubit $j$ and qubit $l$ separately. In that case, Algorithm~\ref{alg:clusters} might not assign the three qubits to the same cluster/neighborhood, even though they are correlated.
A way to overcome the above limitation is to consider s generalizations of Algorithm~\ref{alg:clusters} that makes use of three-qubit parameters ``$c_{jl\rightarrow i}$'', hence requiring a $k\geq 3$ DDOT procedure.
This might be particularly useful for refining dependencies between clusters which were reported by the original Algorithm~\ref{alg:clusters}.

\subsection{Constructing balanced DDOT collections\label{sec:app:ddot_balanced}}

Interestingly, the tools utilized in the previous sections for the analysis of statistical deviations can be used to estimate the probability that a given collection of DDOT circuits will be approximately balanced.
Recall that a perfect collection of   $\rbracket{N,k}$ DDOT circuits is a collection of $N$-qubit quantum circuits consisting of $\iden$ and $X$ gates with a property that for each $k$-qubit subset every $k$-qubit computational basis state is implemented at least once in the whole collection.
If the collection is \textit{balanced}, it means that additionally each basis state is sampled the same number of times.

Consider the randomized construction of such collection -- take circuits which are uniformly random combinations of $0$ and $1$ symbols ('$0$' corresponding to the identity gate and '$1$' to the NOT gate). 
This can be constructed efficiently since it suffices to choose each of the $N$ bits at random independently.
Such construction can be viewed as sampling from a $2^N$-dimensional uniform distribution (corresponding to all possible $N$-bit strings describing possible circuits).
The $k$-bit marginal distributions obtained from this distribution correspond to the distribution of local circuits on $k$-qubit subsets (i.e., the noise characterization by implementing computational-basis states on $k$-qubit subsets).
Having this perspective, we can formulate the problem of approximate balancing of DDOT family in the following way: What is the probability that, when sampling from the global $2^N$-dimensional uniform distribution, all of the $k$-bit marginals (corresponding to the subsets of interest in DDOT) will be at most $\epsilon$-distant from the uniform distributions (on $2^k$-dimensional space)? 

Now we can use tools of statistical analysis presented in the previous sections.
Our single marginal distribution has $2^k$ outcomes, hence from Eq.~\eqref{app:eq:probability_statistical}
we get that the probability of a single marginal being \emph{at least} $\epsilon$-distant in TVD is bounded by
\begin{align}
  P_{1} \leq 2^{2^k}\exp\rbracket{-2s\epsilon^2} \ ,
\end{align}
where $s$ is here the number of random circuits in the collection (viewed as samples from the global uniform distribution of bitstrings).
Since there is $\binom{N}{k}$ number of $k$-bit marginals, applying the union bound (as in the derivations in the previous sections) gives the upper bound for the probability $P_{N,k}$ that at least one of the $\binom{N}{k}$ marginals is at least $\epsilon$-distant from uniform distributions as
\begin{align}\label{app:eq:circuits_balanced}
P_{N,k}\leq \binom{N}{k}2^{2^k}\exp\rbracket{-2s\epsilon^2} \ .
\end{align}

Hence with probability \textit{at least} $1-\binom{N}{k}2^{2^k}\exp\rbracket{-2s\epsilon^2}$ all marginals are \textit{at most} $\epsilon$-distant from uniform distribution.
Note that $\epsilon=0$ corresponds to a perfectly balanced family of circuits, and small but non-zero $\epsilon$ will correspond to the approximately balanced family.

Let us now choose some parameter $\delta$ as an upper bound for probability Eq.~\eqref{app:eq:circuits_balanced} and find the bound for the required number of random circuits $s$ (viewed as samples from a global uniform distribution) which are needed to obtain a family for which each marginal is distant from the uniform distribution by at most $\epsilon$.
After basic manipulations of Eq.~\eqref{app:eq:circuits_balanced} we obtain that 
\begin{align}
s \geq \frac{2^k\log{2}+\log{\binom{N}{k}}+\log\rbracket{\frac{1}{\delta}}}{2\epsilon^2} \approx \frac{2^k+k\log{N}+\log{\frac{1}{\delta}}}{2\epsilon^2} \ .
\end{align}

\subsection{The efficiency of the random construction of DDOT collection\label{sec:app:ddot_perfect}}
\subsubsection{Random circuits}
Adopting the perspective from the previous section, we can in a simple manner tackle the problem of bounding the required number of random circuits which are needed to obtain a perfect collection of $\binom{N}{k}$ DDOT circuits.
Consider randomly sampling $s$ number of bit-string of length $N$.
Now for a fixed $k$-element subset, the probability of a particular $k$-element combination \emph{not appearing} is $1-\frac{1}{2^k}$. 
Hence after $s$ samples, there is $\rbracket{1-\frac{1}{2^k}}^s$ probability that this particular combination did not appear.
Since there are $2^k\binom{N}{k}$ combinations of interest (i.e., $2^k$ small $k$-length bitstrings for all $\binom{N}{k}$ subsets), we can use the union bound to obtain
\begin{align}\label{eq:app:perfect_collection_probability}
    2^k\binom{N}{k} \rbracket{1-\frac{1}{2^k}}^s \approx  2^k\binom{N}{k} \exp\rbracket{-\frac{s}{2^k}}
\end{align}
as the upper bound for the probability that \emph{at least} one $k$-length bit-string did not appear after $s$ samples.
This means that with probability of at least $1-2^k\binom{N}{k} \exp\rbracket{-\frac{s}{2^k}}$, all of the $k$-length bitstrings appeared. 

Let us now choose some parameter $\delta$ as the upper bound for Eq.~\eqref{eq:app:perfect_collection_probability} and calculate the resulting bound for the number of samples (i.e., random circuits). 
After basic manipulations we obtain
\begin{align}
s > 2^k\rbracket{k \log{2}+\log{\binom{N}{k}}+\log{\frac{1}{\delta}}} \approx 2^{k}\rbracket{k\log{2N}+\log{\frac{1}{\delta}}} \ .
\end{align}

\subsubsection{Random hash functions}
Using simple combinatorial arguments, in Ref.~\cite{Cotler2020quantum} it is shown that the probability of the collection of $\binom{N}{k}$ \textit{random hash functions} not being perfect (for our purposes perfect hash function collection means that it can be used to construct a perfect DDOT collection) is upper bounded by
\begin{align}\label{app:eq:hash_functions_cotler}
    \binom{N}{k}\rbracket{1-\frac{k!}{k^k}}^{L} \ ,
\end{align}
where $L$ is the number of generated random hash functions.
Then the authors simplify the use of the above fact to derive a bound on the needed number of hash functions for the collection to be perfect.
Let us repeat the derivation from Ref.~\cite{Cotler2020quantum} with paying attention to specific factors.
From Algorithm~\ref{alg:random_generation} it follows that the number of circuits $s$ obtained from a given hash function collection is equal to $s = 2+\rbracket{2^{k}-2}L$.
Now by choosing the parameter $\delta$ as upper bound on Eq.~\eqref{app:eq:hash_functions_cotler}, we obtain that required number of circuits is lower bounded by 
\begin{align*}\label{app:eq:hash_functions_scaling}
    s &\geq 2+\rbracket{2^k-2} \frac{-\log\binom{N}{k}+\log{\delta}}{\log{\rbracket{1-\sqrt{2\pi k}\ e^{-k}}} }  \\
    &\ge 2+
    \rbracket{2^k-2} \frac{-\log\binom{N}{k}+\log{\delta}}{-\sqrt{2\pi k}\ e^{-k} - 2 \pi k e^{-2k}} \  \\
    &\ge 2+ \rbracket{2^k-2} \frac{-\log\binom{N}{k}+\log{\delta}}{-\sqrt{2\pi}\ k\ e^{-k}}  \\
    &\approx 2^k e^{k}  \rbracket{\frac{1}{k}\log{\frac{1}{\delta}}+\log{N}}  \ , \numberthis
\end{align*}
where we used that $\log (1-x) \ge -x -x^2$. Let us note that Eq.~\eqref{app:eq:hash_functions_scaling} is a reiterated result from Ref.~\cite{Cotler2020quantum}.

In above derivations we assumed that $\rbracket{-\log{\binom{N}{k}}+\log\delta}<0$.
, then we utilized the fact that 
\begin{align}
   \sqrt{2\pi}\ k^{k+\frac{1}{2}}e^{-k}\ \leq k! \leq   e\ k^{k+\frac{1}{2}}e^{-k} ,
\end{align}
and we used approximation $\binom{N}{k}\approx N^k$.

\subsection{Heuristic balancing of DDOT collection\label{sec:app:ddot_heuristic}}
After generating a perfect DDOT collection, one can be interested in making it more balanced. 
There are various possible figures of merit that can be used to quantify the ``balancing'' of the family. 
For example, in previous sections, we used TVD between a uniform distribution and generated sample, when viewing obtained circuits in the collection as samples from the uniform distribution.
Another possibility is to calculate a number of appearances of each marginal term in the whole collection (for example, for 2-qubit subsets the number that each of $00$, $01$, $10$, and $11$ appeared, for every 2-qubit subset), and then take an empirical standard deviation $\sigma_{n,k}$ of this quantity. 
A perfectly balanced family would have 0 standard deviation defined in that way.

Now we discuss a simple heuristic method to improve balancing. 
The starting point is a perfect $\rbracket{N,k}$ DDOT collection. 
Now we apply the following steps in a loop.
\begin{enumerate}
    \item Calculate a number of appearances of all $k$-qubit terms in the collection.
    \item Find the set of non-overlapping $k$-bit marginals which appear in the collection the least number of times compared to the whole population. 
    If $k$ does not divide $N$, choose $\lfloor\frac{N}{k}\rfloor$ subsets and add random gates to the remaining bits. 
    \item Add circuit which implements those least-appearing marginals to the collection.
\end{enumerate}
For example, say that in the last step of the above procedure for $k=2$ and $N=5$ we added circuit
\begin{align}
    01100 \ .
\end{align}
This might mean that we found that the least appearing marginal is state $01$ on qubits $Q0$ and $Q1$, the marginal $10$ on qubits $Q2$ and $Q3$ appeared the same or second-least number of times in the whole collection, while the $0$ on last qubit $Q4$ was random.
In this way, we are adding missing marginal terms ``by hand'' at each step of the loop.
Clearly, this procedure is heuristic and it is not guaranteed to succeed since by adding certain circuits that implement desired marginals we also implement other marginals (in the example above, we add, e.g., marginal $00$ on qubits $Q0$ and $Q3$).
If it happens that those other marginals are in the opposite ``tail'' of the whole distribution (i.e., they are the most abundant ones), then the collection will not become more balanced (it can actually become less balanced).
However, on average it is more likely that by doing so we add marginals that are closer to the ``average marginals'' (i.e., those which appear a number of times close to the mean appearance number).

Clearly, a lot of practical refinements of the described method are possible -- for example, in the second step of the method, instead of adding the least appearing marginals one can focus on maximizing the number of low-appearance marginals. 
In practice, we found that the described method without modifications usually reduced $\sigma_{N,k}$ with a growing number of circuits added in the loop.

To conclude, let us note that the special case of the above technique can be also used to add circuits to the non-perfect collection of DDOT circuits in order to make it closer to being perfect.
Namely, if the collection is not perfect, this means that it does not implement some computational-basis states on some $k$-qubit subsets.
Hence in that case the above procedure would report ``the least appearing marginals'' to be those which are missing in the collection and it would keep adding them in a loop until there are no missing terms -- a strategy which might turn out to better then adding random circuits if the number of missing terms is not too big.

\subsection{Overestimation of correlations\label{sec:app:ddot_overestimation}}

In the main text, we explained that when considering the implementation of DDOT circuits, some two-qubit correlations might be overestimated due to the fact the collection of circuits is not balanced. 
To understand this effect, let us now consider the following illustrative (but rather unrealistic) example.

\subsubsection{Explanation of the effect and post-processing strategy}
Say that we want marginal probability distribution on qubits $Q0$ and $Q2$ when input state was $\ket{00}$.
Then if noise is \textit{uncorrelated} it does not matter whether global input state on three qubits was $\ket{000}$ or $\ket{010}$ ($\ket{001}$ or $\ket{011}$).
However, when there are correlations, it might happen that those two distributions \textit{will be different} depending on the input state of $Q1$.
When we marginalize over $Q1$ we forgot about ``where the data came from''.
Normally, we would just add the marginal probability from circuits implementing both global states and then normalized it.
But if global state $\ket{000}$ was implemented a different number of times than $\ket{010}$, it can cause that effectively correlations which are caused by $Q1$ are wrongly identified as correlations between $Q0$ and $Q2$, because some global probability distributions contribute to the marginal with higher weights (i.e., are effectively counted more times when calculating marginal distributions).
As mentioned in the main text, the natural way to reduce such effects is to create a collection that is balanced, hence it samples from all two-qubit states the same number of times. 
The other thing one can do is to perform post-processing of the experimental data in such a way, that all contributions to the given marginal are weighted by the inverse of a number of times they were implemented. Note that they come from different global distributions, so what is important is the number of times the given marginal state was implemented together with some specific state on all the other qubits.
Importantly, this method is not perfect because for big systems the ``state on all the other qubits'' will likely be different each time anyway (this is due to the fact that collection of DDOT circuits will be random, hence it becomes quite unlikely to obtain two times the same bitstring if the number of qubits is high).
Another thing one can do is to change the weights of the given contribution to the marginal distribution depending on the state of some particular subset of qubits in order to assess whether inferred correlations were correct.
Specifically, one might perform a recursive procedure in which the structure of clusters and neighborhoods inferred from non-post-processed data is validated on particular subsets using this type of post-processing.
There are certainly a lot of practical possibilities to improve the post-processing scheme and we intend to investigate them in future research.

\subsubsection{Illustration of the effect}
To have some idea how big a described effect can be, let us consider implementation of the following collection of circuits on three qubits $\cbracket{Q_0,Q_1,Q_2}$
\begin{align*}\label{eq:example_collection}
\{000,  \textcolor{red}{0}0\textcolor{red}{1}, \textcolor{red}{0}0\textcolor{red}{1}, 010,  \textcolor{red}{0}1\textcolor{red}{1}, 100, 101 ,110 , 111 \} \ , \numberthis
\end{align*}
where each $0$ corresponds to implementation of identity gate and each $1$ to the NOT gate.
Note that considered collection is clearly abundant for 3-qubit characterization since it contains more than $2^3=8$ circuits.
One can think of the above circuits as a part of a collection of DDOT circuits on a higher number of qubits, while we only look at a specific triple of qubits.
Now let us assume that that the qubits $Q_0$ and $Q_2$ are completely \emph{uncorrelated} in terms of readout noise.
On the other hand, $Q_1$ and $Q_2$ are highly correlated in a following way
\begin{align}
    Q_1 &= \ket{0} \implies \text{do nothing}\ ,\\
    Q_1 &= \ket{1} \implies \text{apply bitflip to $Q_2$} \ .
\end{align}
Now if one takes outcomes from circuits from the collection, add them up, normalize and consider resulting estimators of probabilities of obtaining different outcomes after implementation of collection from Eq.~\eqref{eq:example_collection} \emph{marginalized} over qubit 1, it follows from direct computation that the noise matrix on $Q_0$ and $Q_2$ is of the form
	\begin{align}
		\Lambda^{Q_0Q_2} = \begin{pmatrix}
		p\rbracket{00|00} & \quad p\rbracket{00|01} &\quad p\rbracket{00|10} &\quad p\rbracket{00|11}\vspace{0.1cm}\\
		p\rbracket{01|00} &\quad p\rbracket{01|01} &\quad p\rbracket{01|10} &\quad p\rbracket{01|11}\vspace{0.1cm}\\
		p\rbracket{10|00} &\quad p\rbracket{10|01} &\quad p\rbracket{10|10} &\quad p\rbracket{10|11}\vspace{0.1cm}\\
		p\rbracket{11|00} &\quad p\rbracket{11|01} &\quad p\rbracket{11|10} &\quad p\rbracket{11|11}
		\end{pmatrix} =   \begin{pmatrix}
		\frac{1}{2}& \quad \textcolor{red}{\frac{1}{3}} &\quad 0 &\quad 0\vspace{0.1cm}\\
		\frac{1}{2} &\quad \textcolor{red}{\frac{2}{3}} &\quad 0 &\quad 0\vspace{0.1cm}\\
		0 &\quad 0 &\quad \frac{1}{2}&\quad \frac{1}{2}\vspace{0.1cm}\\
	    0 &\quad 0 &\quad \frac{1}{3} &\quad \frac{1}{2}
		\end{pmatrix} \ , 
		\end{align}
		from which we obtain that the noise matrices on $Q_0$ depending on state of $Q_2$ are
		\begin{align}
		\Lambda^{Y_2 = '0'} &= \begin{pmatrix}
		p\rbracket{00|00}+p\rbracket{01|00}   &\quad p\rbracket{00|10}+p\rbracket{01|10} \vspace{0.1cm} \\
		p\rbracket{10|00} +p\rbracket{11|00}  &\quad p\rbracket{10|10}+p\rbracket{11|10}
		\end{pmatrix}\ = \begin{pmatrix}
		1  &\quad 0 \vspace{0.1cm} \\
		0  &\quad 1
		\end{pmatrix}\ \\
		\Lambda^{Y_2 = '1'} &= \begin{pmatrix}
		p\rbracket{00|01}+p\rbracket{01|01}   &\quad p\rbracket{00|11}+p\rbracket{01|11} \vspace{0.1cm} \\
		p\rbracket{10|01} +p\rbracket{11|01}  &\quad p\rbracket{10|11}+p\rbracket{11|11}
		\end{pmatrix}\ = \begin{pmatrix}
		1   &\quad 0\vspace{0.1cm} \\
		0 &\quad 1
		\end{pmatrix}\, \\
		\end{align}
		and on $Q_2$ depending on state of $Q_0$
		\begin{align}
		\Lambda^{Y_0 = '0'} &=\begin{pmatrix}
		p\rbracket{00|00}+p\rbracket{10|00}   &\quad p\rbracket{00|01}+p\rbracket{10|01} \vspace{0.1cm} \\
		p\rbracket{01|00} +p\rbracket{11|00}  &\quad p\rbracket{01|01}+p\rbracket{11|01}
		\end{pmatrix}\ = \begin{pmatrix}
		\frac{1}{2}  &\quad \frac{1}{2}\vspace{0.1cm} \\
		\frac{1}{2}  &\quad \frac{1}{2}
		\end{pmatrix} \\
		\Lambda^{Y_0 = '1'} &= \begin{pmatrix}
		p\rbracket{00|10}+p\rbracket{10|10}   &\quad p\rbracket{00|11}+p\rbracket{10|11} \vspace{0.1cm} \\
		p\rbracket{01|10} +p\rbracket{11|10}  &\quad p\rbracket{01|11}+p\rbracket{11|11}
		\end{pmatrix} \ =\begin{pmatrix}
		 \textcolor{red}{\frac{1}{3}}   &\quad \frac{1}{2}  \vspace{0.1cm} \\
		 \textcolor{red}{\frac{2}{3}}  &\quad \frac{1}{2} 
		\end{pmatrix} \ .
		\end{align}
		The above matrices give correlation factors (Eq.~\eqref{eq:correlations_pair}
        \begin{align}
            c_{2\rightarrow 0} &= 0\\
            c_{0\rightarrow 2} &= \frac{1}{6} \ .
        \end{align}		
        Hence the Algorithm~\ref{alg:clusters} would report that the noise on $Q_2$ significantly depends on the state of $Q_0$ even though physically there is no dependence.
        The correlations between $Q_1$ and $Q_2$ give rise to false correlations between $Q_0$ and $Q_2$ after marginalizing over $Q_1$.
        This is solely due to the fact that in Eq.~\eqref{eq:example_collection} different three-qubit states are sampled different numbers of times, which gives some contributions to the marginals have higher weight when constructing effective noise matrix on $Q_0$ and $Q_2$.

\section{Noise characterization scheme overview}\label{sec:app:characterization}

Here we provide a step-by-step description of our noise characterization procedure. 
We note that stages 0 and 1 were only briefly mentioned in the main text, and they correspond to the verification of the undertaken assumptions: stage 0 verifies the quality of the single-qubit gates, and stage 1 the assumption about classical nature of the noise in the measurement device.
Stages 2 and 3 describe the proper characterization scheme of Diagonal Detector Overlapping Tomography which was discussed in detail in the main part of the work and in Appendix~\ref{sec:app:ddot}.

\subsection{Stage 0 -- single-qubit gate fidelities}
	To begin, let us note that of our characterization procedure relies on the assumption of perfect state preparation.
	However, in practice, this assumption might be significantly violated. 
	Therefore we propose not to use qubits with single-qubit gate infidelities above some threshold -- in our experiments we arbitrarily chose this threshold to be $0.01$.
	
	In experiments on IBM's \textit{Melbourne} backend, the single-qubit gate fidelities were good enough (fidelities above $99\%$) to use all of the qubits. 
	In the case of Rigetti's \textit{Aspen-8}, we discarded  8 qubits which had fidelities below $98\%$, while still using qubits 5 qubits which had fidelity in range $\sbracket{98\%,99\%}$.
	\subsection{Stage 1 -- assessing classical form of the noise}
	In order to perform simultaneous estimation of single-qubit detectors with an overcomplete operator basis, one needs to implement 6 different circuits -- each implementing eigenstate of a different Pauli matrix on every qubit at the same time.
	After this implementation, one needs to post-process data to obtain marginal single-qubit distributions and use standard detector tomography algorithms, for example, those described in \cite{Fiurasek2001} (and implemented in Python in online repository \cite{qrem}).

	Having reconstructed POVMs describing each single-qubit detector in a device, one can assess the classicality of the noise using methods described in Ref.~\cite{Maciejewski2020mitigation} -- 
	for the sake of completeness, we will recall the main notions of that procedure here.
	First, we will need a notion of distance between quantum measurements.
	Such distances are usually related to the probability distributions that those measurements generate via Born's rule.
    Recall from Eq.~\eqref{eq:TVD} that the Total Variation Distance between probability distributions $\p$ and $\mathbf{q}$ as L1 norm of difference of those vectors
	\begin{align}\label{eq:appendix_TVD}
	\text{D}_{\text{TV}}\rbracket{\p,\mathbf{q}} = \frac{1}{2}||\p-\mathbf{q}||_{1} = \frac{1}{2} \sum_i |p_i-q_i| \ .
	\end{align}
	Now, the distance between quantum measurements related to TVD is the operational distance  $\text{D}_{\text{op}}$ defined for two POVMs $\M$ and $\P$ as \cite{Puchala2018}
	\begin{align}\label{eq:appendix_operational_distace}
	\text{D}_{\text{op}}\rbracket{\M,\P} = \max \text{D}_{\text{TV}}\rbracket{\p_{\M},\mathbf{q}_{\P}} \ ,
	\end{align}
	where $\p_{\M}$ (or $\mathbf{q}_{\P}$) is a probability distribution generated by measurement $\M$ (or $\P$) via Born's rule, and the maximization goes over all quantum states. 
	Hence, the operational distance between two quantum measurements is the worst-case distance between probability distributions they can generate.
	Now, to quantify readout noise one can simply calculate\footnote{See Refs. \cite{Puchala2018,Maciejewski2020mitigation} on practical calculation of RHS of Eq.~\eqref{eq:appendix_operational_distace}.} operational distance for reconstructed POVM $\M$ for each qubit and the ideal measurement $\P$. 
	Then, to quantify coherent part of the noise, one can make the following decomposition \cite{Maciejewski2020mitigation}
	\begin{align}\label{eq:appendix_measurement_decomposition}
	\M = \underbrace{\Lambda\P}_{\text{classical part}}+\underbrace{\Delta}_{\text{coherent part}} \ .
	\end{align}
	For the ideal measurement $\P$ being the standard measurement in the computational basis, this decomposition is straightforward -- the classical part of the noise is contained in the diagonal part of the measurement operators, while off-diagonal terms are a coherent part.
	The magnitude of the coherent part of the noise can be quantified as $		\text{D}_{\text{op}}\rbracket{\M,\Lambda\P}$.
	The assumption of fully classical noise leads to discarding the coherent part $\Delta$ in the Eq.~\eqref{eq:appendix_measurement_decomposition} and performing error-mitigation \textit{as if} POVM $\Lambda \P$ was exact description of the detector. 
	This leads to the propagation of coherent errors under error-mitigation, and it can be quantified via $||\Lambda^{-1}||_{1\rightarrow1}\text{D}_{\text{op}}\rbracket{\M,\Lambda\P}$.
	Following guidelines in Ref.~\cite{Maciejewski2020mitigation}, we propose to \emph{discard} every qubit which fulfills the following inequality
	\begin{align}\label{app:eq:rule_of_thumb}
	||\Lambda^{-1}||_{1\rightarrow1}\text{D}_{\text{op}}\rbracket{\M,\Lambda\P} \geq \text{D}_{\text{op}}\rbracket{\M,\P} \ .
	\end{align}
	Fortunately, in experiments on both IBM's and Rigetti's machines, this step did not lead us to discard any qubits -- the noise in those devices remains highly classical, as indicated in previous experiments \cite{Maciejewski2020mitigation,Chen2019detector}.
	
	Before proceeding further, let us note that while assessing classicality of the noise via single-qubit QDTs we make the following implicit assumption -- the coherent part of the noise does not scale significantly with growing system size (by this we mean that it is at most additive in the number of qubits).
	Furthermore, in using the rule Eq.~\eqref{app:eq:rule_of_thumb} we disregard effects of statistical deviations.

    \subsection{Stage 2 -- Diagonal Detector Overlapping Tomography}
	
	The main idea of DDOT was described in the main text. 
	The practical generation of relevant circuits was described in Appendix~\ref{sec:app:ddot_construction} and summarized in Algorithm~\ref{alg:random_generation} (using random hash functions) and in Algorithm~\ref{alg:random_generation_circuits} (using random circuits).

    \subsection{Stage 3 -- inferring correlations structure}
	The procedure of  inferring correlations (i.e., structure of clusters and neighbourhoods) was described in detail in Appendix~\ref{sec:app:ddot_inference} and summarized in Algorithm~\ref{alg:clusters}.

\section{Sample complexity of energy estimation\label{sec:app:sampling} }

In this section, we give some derivations related to the estimation of expectation values of local Hamiltonian terms on various quantum states.
We start by discussing correlations in random states in Appendix~\ref{sec:app:sampling_states}. 
Then in Appendix~\ref{sec:app:sampling_proof2}, we prove Proposition~2 from the main text, which concerns states appearing in the QAOA algorithms.
Finally, in Appendix~\ref{sec:app:sampling_covariances} we analyze what happens with the covariances of local Hamiltonian terms if the uncorrelated measurement noise and its mitigation are present.

\subsection{Local correlations in random states\label{sec:app:sampling_states}}

\begin{proposition}\label{prop:covMIXEDBOUND}
 Let $\ket{\psi}$ be a pure state on $(\mathbb{C}^2)^{\ot N}$, for  a subset $\gamma\subset[N]$ of qubits we denote by $\rho_\gamma$ and $\id_\gamma$ the marginal of $\ketbra{\psi}{\psi}$ corresponding to $\gamma$ and the maximally mixed state on $\gamma$, respectively.  Let $H_\alpha,H_\beta$ be local hamiltonians term that act on disjoint substets of qubits. We then have
 \begin{equation}
     \mathrm{Cov}(H_A,H_B) \leq 3 \| H_A\| \|H_B\| \|\rho_{\alpha\cup \beta } - \id_{\alpha\cup\beta} \|_1 .
 \end{equation}
\end{proposition}
\begin{proof}
A simple algebra gives
\begin{align}
\mathrm{Cov}(H_\alpha,H_\beta) =  \tr\left(H_\alpha \ot H_\beta ( \rho_{\alpha\cup\beta} - \rho_{\alpha}\otimes  \rho_\beta) \right) \nonumber \\
   =  \tr\left(H_\alpha\ot  H_\beta (  \Delta_{\alpha\cup\beta} -\Delta_{\alpha}\ot\rho_\beta +\id_\alpha \ot \Delta_\beta) \right)\  ,
 \end{align}  
 where $\Delta_\gamma=\rho_\gamma -\id_\gamma$. We now apply the well-known inequality 
 \begin{equation}\label{eq:simpleINEQ}
     \tr( X) \leq \|A\| \|X\|_1\ ,
 \end{equation}
 for $A=H_\alpha \ot H_\beta$ and $X=\Delta_{\alpha\cup\beta} -\Delta_{\alpha}\ot\rho_\beta -\id_\alpha \ot \Delta_\beta$. The 1-norm can be upper bounded as follows:
 \begin{align}\label{eq:DeltBound}
\|    \Delta_{\alpha\cup\beta} -\Delta_{\alpha}\ot\rho_\beta +\id_\alpha \ot \Delta_\beta \|_1 \leq  \|\Delta_{\alpha\cup\beta} \|_1 + \|\Delta_{\alpha} \|_1 +\|\Delta_{\beta} \|_1 \leq 3 \|\Delta_{\alpha\cup\beta} \|_1\ ,
 \end{align}
 where we used the following properties of 1-norm: triangle inequality, multiplicativity: $\|A\ot B\|_1 = \|A\|_1 \|B\|_1  $, and data-processing inequality \cite{Nielsen2010}. We conclude the proof by inserting \eqref{eq:DeltBound} into \eqref{eq:simpleINEQ} and using $\|H_\alpha \ot H_\beta\|=\| H_\alpha\| \|H_\beta\|$. 
\end{proof}

\subsection{Proof of Proposition~\ref{lem:random_graphs}\label{sec:app:sampling_proof2}}

Consider Hamiltonian with connectivity given by Erd\"os-R\'enyi random graph with $N$ nodes and $K$ edges with average degree equal to $q=\frac{K}{N}$. 
We will show that 
if the number of levels  satisfies $p = c  \log(N)+1$ with $c \le \frac{1}{2 \log( 2 q/ \ln 2)} $, then the variance of the energy  $\mathrm{Var}\rbracket{\H}$  will scale as $\mathcal{O}(N^{2-x})$, with $x > 0$ depending on $c$.

 Let $G =(V, E)$ be a graph (with vertex set $V$ and edge set $E$), we denote by $B(i,r)$ the set of vertices that are in graph distance $r$ or less from vertex $i$. 
 Similarly, for an edge $\alpha=(i,j) \in E$,  we define 
 $C(\alpha, r)$ the set of vertices that are in graph distance $r$ or less away from $\alpha$.
 For $\alpha=(i,j)$ and $\beta=(v,w)$, we have that if $v \notin B(i, r+2)$, then $\beta \notin C(\alpha, r)$.

 The random Hamiltonian corresponding to a graph can be written as sum of 2-qubit terms (with also single qubit terms incorporated into to these) corresponding to the edges $E$ of the random graph $G(N,q)$, i.e., $H = \sum_{(i,j) \in E} H_{(i,j)}$.
 The variance is then  
 \begin{align} \label{eq:general_exp}
    \mathrm{Var}\rbracket{\H} = \sum_{\alpha,\beta \in E}\mathrm{Cov}\rbracket{H_{\alpha},H_{\beta}} \ .
\end{align}

To bound this quantity we will utilize the following two facts.
First, we notice that than any non-zero term in Eq.~\eqref{eq:general_exp} can be bounded by
\begin{align}
    \mathrm{Cov}\rbracket{H_{\alpha},H_{\beta}} \leq \mathrm{Var}\rbracket{H_{\alpha}} \ \mathrm{Var}\rbracket{H_{\beta}}  \leq ||H_{\alpha}||\ ||H_{\beta}||  ,
\end{align}
where we used known covariance-variance inequality together with Popoviciu's inequality.

Second, we reiterate an important observation from Ref.~\cite{farhi2020quantum} about these types of QAOAs: Consider two operators $O_1$ and $O_2$ acting non-trivially only on the sets of qubits  $A_1 \subset V$ and $A_2 \subset V$, respectively. If $U$ is a unitary corresponding to a $p$-level QAOA (with any parameter setting), the set of nodes $A_1$ and $A_2$ is in graph distance at least $2p$ distance from each other and  $\ket{\psi}$ is product state, then 
\begin{equation}\label{app:eq:farhi_result}
    \bra{\psi} U^\dagger O_1 O_2 U \ket{\psi} = \bra{\psi} U^\dagger O_1 U \ket{\psi} \bra{\psi} U^\dagger O_2 U \ket{\psi}.
\end{equation}

With these two ingredients, we can bound the variance as
\begin{align} 
    \mathrm{Var}\rbracket{\H} &= \sum_{\alpha,\beta \in E}\mathrm{Cov}\rbracket{H_{\alpha},H_{\beta}} \nonumber \\
    &= \sum_{\alpha \in E} \sum_{\beta \in C(\alpha, 2p)} \mathrm{Cov}\rbracket{H_{\alpha},H_{\beta}} \nonumber \\
    &\le 
    \sum_{\alpha \in E} \; \; \sum_{\beta \in C(\alpha, 2p)} ||H_{\alpha}||\ ||H_{\beta}||  \ .
\end{align}
An even more rough upper bound can be given by
\begin{align} \label{app:eq:rough_bound}
    \mathrm{Var}\rbracket{\H} &\le 
    f_{\H} \sum_{\alpha \in E} \max_{\alpha \in E} |\mathrm{C}(\alpha,2p)|  = f_{\H}\ K\  \max_{\alpha \in E} |\mathrm{C}(\alpha,2p)|  \ ,
\end{align}
with setting
\begin{align}
        f_{\H} = \max_{\alpha,\beta} ||H_{\alpha}||\ ||H_{\beta}|| \
\end{align}

Now we can turn to the concrete case of Hamiltonians corresponding to graph sampled from the  Erd\"os-R\'enyi graphs with average degree $q$. For such graphs the {\it Neighborhood Size Theorem}  states the following \cite{farhi2020quantum}:\\
For any 
\begin{equation}
    r <   \frac{  w\log(N)}{ 4 \log(2q/ \ln(2)},
\end{equation}
where $0 <w >1$ , there exists a constants $a>0$ and $A < 1$ such that
\begin{equation}
    \mathrm{Prob}\Big[\max_{i \in V} |\mathrm{B}(i,r)| \ge N^{A/2}\Big] \le e^{-N^{a/2}}.
\end{equation}
To be more specific, we can give the expression of $a$ and $A$ in terms of $r$ and $N$:
\begin{align}
   A &= w \, \frac{(2 + |\log_{2q}(\ln(2))|) }{(1+ |\log_{2q}(\ln(2))|)}   ,\; \; \\
   a &= \frac{w}{3(1 + |\log_{2q}(\ln(2))|}  \ .
\end{align}
The proof of the above comes from  the proof of Neighborhood Size Theorem given in \cite{farhi2020quantum}, page 13. 
Specifically our inequalities correspond to setting $s=1$.  

Now note that since for any $\alpha =(i,j) \in E$ and $\beta = (v,w) \in E$, we have that if $\beta \in \mathrm{C}(\alpha,r)$ then $v \in \mathrm{B}(i,r+2)$. 
This immediately implies that generally for  $\alpha =(i,j) \in E$ we have $|\mathrm{C}(\alpha,r)| \le  |\mathrm{B}(i,r+2)|^2$, and thus also 
\begin{equation}\label{app:eq:concentration_qaoa}
    \mathrm{Prob}\Big[\max_{\alpha \in E} |\mathrm{C}(\alpha,r-2)| \ge
    N^{{A}}\Big] \le 
    \mathrm{Prob}\Big[\max_{\alpha \in E} |\mathrm{B}(\alpha,r)|^2 \ge
    N^{{A}}\Big] =
    \mathrm{Prob}\Big[\max_{\alpha \in E} |\mathrm{B}(\alpha,r)| \ge
    N^{{A/2}}\Big] \le
    e^{-N^{{a/2}}}.
\end{equation}
If we now choose $2p=r-2$, and thus 
\begin{equation}
    p <   \frac{  w\log(N)}{ {8} \log(2q/ \ln(2))} - 1,
\end{equation}
then by combining Eq.~\eqref{app:eq:concentration_qaoa} with bound in Eq.~\eqref{app:eq:rough_bound}, we get that with probability at least $1-e^{-N^{{a/2}}}$ variance of the Hamiltonian is bounded by
\begin{align} \label{app:eq:final_bound}
    \mathrm{Var}\rbracket{\H} &\le 
    f_{\H}\ K\  \max_{\alpha \in E} |\mathrm{C}(\alpha,r)| \leq f_{\H} \ K\ N^{{A}} = f_{\H} \ q\ N^{{A}+1}
\end{align}
which is statement of Proposition~\ref{lem:random_graphs}.

\subsection{Covariances of local terms in presence of uncorrelated readout noise\label{sec:app:sampling_covariances}}

\subsubsection{Presence of uncorrelated readout noise}
As a starting point, let us assume that we have a state with correlations bounded in trace norm 
\begin{align}\label{app:eq:marginal_state_bounded}
    ||\rho_{\alpha \beta}-\rho_{\alpha}\otimes\rho_{\beta}||_1 \leq \epsilon_{\alpha\beta} \ ,
\end{align}
where $\rho_{\alpha \beta}$ is marginal quantum state on subsystems $\alpha$ and $\beta$ which is close to a product state of those subsystems.
The special instance $\epsilon=0$ corresponds to the product state case.
We are interested in what happens with the covariances of local Hamiltonian terms if the measurement is affected by uncorrelated classical noise of the form $\Lambda = \bigotimes_{i}\Lambda_{Q_i}$, where $\Lambda_{Q_i}$ is noise matrix acting on qubit $Q_i$ (see Eq.~\eqref{eq:approx_error_definition}).
Similarly, as in the main text, the Hamiltonian we consider is classical (i.e., diagonal), and local terms can be decomposed into sums of products of Pauli $\sigma_z$ terms. 
Let us denote by $\mathbf{p}^{\text{noisy}}_{\alpha} = \Lambda_{\alpha} \mathbf{p}^{\text{ideal}}_{\alpha}$ a noisy marginal distribution on the qubit subset $\alpha$ corresponding to local Hamiltonian term $H_{\alpha}$.
Note that in the uncorrelated noise model $\Lambda_{\alpha} = \bigotimes_{Q_i \in \alpha}\Lambda_{Q_i}$.
Similarly to Eq.~\eqref{app:eq:local_term}, we can write the expectation value of the noisy local Hamiltonian term as 
\begin{align}\label{app:eqq:noisy_local_term}
    \tbracket{\tilde{H}_{\alpha}}_{\rho_{\alpha}} \coloneqq \braket{\mathbf{\lambda}_{\alpha}}{\mathbf{p}^{\text{noisy}}_{\alpha}}= \braket{\mathbf{\lambda}_{\alpha}}{\Lambda_{\alpha}\mathbf{p}^{\text{ideal}}_{\alpha}} 
    \ ,
\end{align}
where $\mathbf{\lambda}^{\alpha}$ is a vectorized spectrum of term $H_{\alpha}$, the subscript $\rho_{\alpha}$ indicates the marginal state on which the expectation value is calculated and we used a convenient braket notation to denote scalar product.

We are now interested in bounding the covariance $\mathrm{Cov}\rbracket{H^{\text{noisy}}_{\alpha},H^{\text{noisy}}_{\beta}}$ between two local Hamiltonian terms when measured in a quantum state from Eq.~\eqref{app:eq:marginal_state_bounded}.
Let us decompose a marginal quantum state as
\begin{align*}
    \rho_{\alpha\beta} &= \rho_{\alpha}\otimes \rho_{\beta} + \Delta_{\alpha\beta} \\
    \Delta_{\alpha\beta} &\coloneqq \rho_{\alpha\beta}- \rho_{\alpha}\otimes \rho_{\beta} \ . \numberthis
\end{align*}
Now can write
\begin{align*}\label{app:eq:convariances_local_noise}
    \mathrm{Cov}\rbracket{\tilde{H}_{\alpha},\tilde{H}_{\beta}} &=    \tbracket{\tilde{H}_{\alpha\beta}}_{\rho_{\alpha\beta}}-   \tbracket{\tilde{H}_{\alpha}}_{\rho_{\alpha}}   \tbracket{\tilde{H}_{\beta}}_{\rho_{\beta}} =\\
    &=\tbracket{\tilde{H}_{\alpha\beta}}_{\Delta{\alpha\beta}}+\tbracket{\tilde{H}_{\alpha\beta}}_{\rho_{\alpha}\otimes\rho_{\beta}}-   \tbracket{\tilde{H}_{\alpha}}_{\rho_{\alpha}}   \tbracket{\tilde{H}_{\beta}}_{\rho_{\beta}} = \\
    &=\underbrace{\braket{\mathbf{\lambda}_{\alpha\beta}}{\Lambda_{\alpha\beta}\mathbf{p}^{\text{ideal}}_{\Delta_{\alpha\beta}}}}_{\text{correlated part}} + 
    \underbrace{\braket{\mathbf{\lambda}_{\alpha\beta}}{\Lambda_{\alpha\otimes\beta}\mathbf{p}^{\text{ideal}}_{\alpha\beta}}-\braket{\mathbf{\lambda}_{\alpha}}{\Lambda_{\alpha}\mathbf{p}^{\text{ideal}}_{\alpha}}\braket{\mathbf{\lambda}_{\beta}}{\Lambda_{\alpha}\mathbf{p}^{\text{ideal}}_{\beta}}}_{\text{uncorrelated part}}
    \ . \numberthis
\end{align*}

In the above we slightly abused the notation -- in general, $\alpha$ and $\beta$ can overlap, in which case one needs to insert proper identities and accordingly redefine product state $\rho_{\alpha} \otimes \rho_{\beta}$ (together with bound in Eq.~\eqref{app:eq:marginal_state_bounded} which will now correspond to the different, more refined division of qubits) and corresponding noise matrices.
We also denoted by $\mathbf{p}^{\text{ideal}}_{\Delta_{\alpha\beta}}$ a formal vector given by diagonal elements of $\Delta_{\alpha\beta}$ (this corresponds to measurement in computational basis).
Now note that in the last line of Eq.~\eqref{app:eq:convariances_local_noise} the part underlined as ``uncorrelated part'' contains terms without any correlations except those which can appear if $\alpha$ and $\beta$ overlap. 
In particular, if $\alpha \cap \beta = \emptyset$, then since the noise is uncorrelated, we have $\bra{\lambda_{\alpha\beta}}=\bra{\lambda_{\alpha}}\otimes\bra{\lambda_{\beta}}$ and  $\Lambda_{\alpha\otimes\beta}\mathbf{p}^{\text{ideal}}_{\alpha \otimes \beta} = \rbracket{\Lambda_{\alpha}\mathbf{p}^{\text{ideal}}_{\alpha}} \otimes \rbracket{\Lambda_{\beta}\mathbf{p}^{\text{ideal}}_{\beta}}$, therefore then this part is equal 0.
Otherwise, it gives a non-zero contribution, which however would be present even without any measurement noise.

The only part which adds non-trivial correlations is therefore $\braket{\mathbf{\lambda}_{\alpha\beta}}{\Lambda_{\alpha\beta}\mathbf{p}^{\text{ideal}}_{\Delta_{\alpha\beta}}}$ which using Eq.~\eqref{app:eq:marginal_state_bounded} and elementary transformations can be bounded as
\begin{align}
    |\braket{\mathbf{\lambda}_{\alpha\beta}}{\Lambda_{\alpha\beta}\mathbf{p}^{\text{ideal}}_{\Delta_{\alpha\beta}}}| \leq ||{H}_{\alpha}{H}_{\beta}||\ ||\Lambda_{\alpha\beta}||_{1\rightarrow 1}\
    ||\Delta_{\alpha\beta}||_{1} \leq ||{H}_{\alpha}{H}_{\beta}||\ \ \epsilon_{\alpha\beta} \ .
\end{align}
In the last inequality we made use of the following facts.
If $\alpha$ and $\beta$ do not overlap, we have $||\Lambda_{\alpha\beta}||_{1\rightarrow1} = ||\Lambda_{\alpha}\otimes \Lambda_{\beta}||_{1\rightarrow1} =  ||\Lambda_{\alpha}||_{1\rightarrow1}\ ||\Lambda_{\beta}||_{1\rightarrow1} = 1$. 
In case $\alpha$ and $\beta$  overlap, the non-overlapping parts will simply give stochastic matrices (and they have $1\rightarrow 1$ norm equal to $1$), while overlapping parts will be squared. 
We therefore obtain $||\Lambda_{\alpha\beta}||_{1\rightarrow 1} = \prod_{i \in \alpha \cap \beta}||\Lambda^2_{i}||_{1\rightarrow1} \leq \prod_{i \in \alpha \cap \beta}||\Lambda_{i}||^2_{1\rightarrow1} =1$, where we used submultiplicity of norm.  
Of course, if the noise model is known, the $||\Lambda_{\alpha\beta}||_{1\rightarrow 1}$ could be also calculated explicitly.

From the above, it follows that if $\epsilon_{\alpha\beta}$ is small, then under uncorrelated measurement noise the covariances between local terms will be, unsurprisingly, small as well.

\subsubsection{Effects of error-mitigation}
To include error-mitigation on the above considerations let us note that since error-mitigation is operation performed classically, it can be incorported into the spectrum of the Hamiltonian by defining spectrum $\bra{\lambda_{\alpha}\corr_{\alpha}^{\dagger}}$ of ``error-mitigated Hamiltonian'' term $H_{\alpha}^{\text{corr}}$ which energy is estimated on the noisy probability distribution $\ket{\mathbf{p_{\alpha}^{\text{noisy}}}}$.
Here we define a dual of correction matrix $\corr_{\alpha}^{\dag}$ acting on the spectrum of the local Hamiltonian.
Note that if the correction matrix is exact (i.e., not approximate), it immediately follows that expected value of such defined Hamiltonian coincides with the true value of the energy
\begin{align}
    \tbracket{H_{\alpha}^{\text{corr}}} = \braket{\lambda_{\alpha}\corr_{\alpha}^{\dagger}}{\mathbf{p_{\alpha}^{\text{noisy}}}} =    \braket{\lambda_{\alpha}\corr_{\alpha}^{\dagger}}{\Lambda_{\alpha}\mathbf{p_{\alpha}^{\text{ideal}}}}  = \braket{\lambda_{\alpha}}{\underbrace{\corr_{\alpha}\Lambda_{\alpha}}_{\iden}\mathbf{p_{\alpha}^{\text{ideal}}}} = \braket{\lambda_{\alpha}}{\mathbf{p_{\alpha}^{\text{ideal}}}} = \tbracket{H_{\alpha}} \ .
\end{align}
Using this perspective we can derive the bounds on covariances between error-mitigated local Hamiltonian terms in a manner fully analogous to previous derivations.
The ``correlated part''  this time is bounded as
\begin{align}
   | \braket{\lambda_{\alpha\beta}\corr_{\alpha\beta}^{\dagger}}{\Lambda_{\alpha\beta}\mathbf{p_{\Delta_{\alpha\beta}}^{\text{ideal}}}} | \leq ||{H}_{\alpha}{H}_{\beta}||\  ||\corr_{\alpha\beta}^{\dag}||_{1\rightarrow 1} ||\Lambda_{\alpha\beta}||_{1\rightarrow 1}\
    ||\Delta_{\alpha\beta}||_{1} \leq ||{H}_{\alpha}{H}_{\beta}||\  ||\corr_{\alpha\beta}^{\dag}||_{1\rightarrow 1} \ 
    \epsilon_{\alpha\beta} \ .
\end{align}
The rest of discussion is identical to the previous analysis of the effects of uncorrelated noise, the difference being that now we have additional terms coming from the duals of correction matrices.
Note that those duals in special cases can have huge values -- indeed, if the noise matrix is barely invertible in the first place, one can expect that error-mitigation will highly increase the uncertainty in the estimation of energy.

\section{Details of numerical simulations and additional experimental data}\label{sec:app:misc}

In this short section, we provide some additional information on numerical simulations and experiments.
This includes detailed discussion of methods used to simulate QAOA in Appendix~\ref{sec:app:misc_QAOA}, and additional experimental results on characterization of correlations in Appendix~\ref{sec:app:misc_experiments}.

\subsection{Simulation of the Quantum Approximate Optimization Algorithm \label{sec:app:misc_QAOA}}

\begin{table}[!t]
\begin{center}
\begin{tabular}{|l|l|l|l|l|l|}
\hline
\textbf{parameter}      & $\alpha$ & $\gamma$ & A   & a    & c    \\ \hline
\textbf{starting value} & 0.602    & 0.101    & 200 & 0.06 & 0.12 \\ \hline
\end{tabular}
\caption{\label{app:tab:qaoa_params}Starting values of hyperparameters used in optimization. The meaning of the parameters is in agreement with standard conventions (see for example Refs.~\cite{Spall1998ANOO,cade2019strategies}).}
\end{center}
\end{table}

Here we provide some details of the performed numerical simulations. 
Simulation of QAOA was performed on the system of $N=8$ qubits.
In the main text, we described how the algorithm works, however, let us now repeat it for the sake of completeness.
In standard implementation \cite{farhi2014quantum}, one initializes quantum system to be in $\ket{+}^{\otimes \noq}$ state, where $\ket{+}=\frac{1}{\sqrt{2}}\rbracket{\ket{0}+\ket{1}}$.
Then $p$-layer QAOA is performed via implementation of unitaries of the form
\begin{align}
U_p = \Pi_{j}^p \exp\rbracket{-i\ \alpha_j \H_{D}}\exp\rbracket{-i\ \beta_j \H_{O}}, \ 
\end{align}
where $\H_{D}$ is driver Hamiltonian, which we take to be
\begin{align}
    \H_{D} = \sum_{k}^{\noq} \sigma_{x}^{k}, \ 
\end{align}
and $\H_{O}$ is objective Hamiltonian that one wishes to optimize (i.e., to find approximation for its ground state energy), and $\cbracket{\alpha_j},\cbracket{\beta_j}$ are the angles to-be-optimized.
The quantum state after $p$-th layer is
\begin{align}
    \ket{\psi_p} = U_p \ket{+}^{\otimes \noq} \ ,
\end{align}
and the function which is passed to classical optimizer is the estimator of the expected value $\bra{\psi_p}\H_{O}\ket{\psi_p}$ (note that this makes those estimators to effectively be a function of parameters $\cbracket{\alpha_j},\cbracket{\beta_j}$).
The estimator is obtained by sampling from the distribution defined by the measurement of $\ket{\psi_p}$ in the computational basis, taking the relevant marginals, and calculating the expected value of $\H_{O}$ using values of those estimated marginals.

Theoretically one could optimize $p$-layer QAOA by simultaneous optimization of all $2p$ angles.
However, this is hard in practice, since the number of optimized parameters increases the complexity of classical optimization.
In our optimizations, we therefore modified the optimization to be divided into steps in the following way.
In each step of optimization, we optimized over the set of $6$ angles (i.e., 3 QAOA layers). 
Then the input to the next step was the optimized state obtained in the previous steps.
For example, input to second optimization step was the quantum state $U^{\rbracket{1}}_3\ket{+}^{\otimes \noq}=U_3 \ket{+}^{\otimes \noq}$, where we used superscript to denote optimization step.
Then input to third optimization step was the state 
$U^{\rbracket{2}}_3U^{\rbracket{1}}_3\ket{+}^{\otimes \noq}= U_6 \ket{+}^{\otimes \noq}$, etc.

In summary, for $p$-layer QAOA, the optimization was effectively divided into $\frac{p}{3}$ steps, and in each step, we optimized over $6$ angles, i.e., $3$ QAOA layers. 
Each step used $2*800=1600$ function evaluations (the factor $2$ comes from two evaluations needed for gradient evaluation which was done 800 times) plus a single final function evaluation.
We further performed each such procedure additional time and chose the better run (out of 2).
Each energy estimator was obtained using $10^4$ energy measurements.
Therefore the total number of function evaluations was  $\approx\frac{p}{3}\times 3.2 \times 10^7$.

As a classical optimizer, we used Simultaneous perturbation stochastic approximation (SPSA) \cite{Spall1998ANOO,cade2019strategies,montanaro2020compressed,Kandala2017}.
Then with a growing number of optimization steps, we gradually changed the parameters to (heuristically) make the optimization more adaptive.
Parameters $\alpha$ and $\gamma$ were not changed, while other parameters were changed according to prescription
\begin{align*}
 a_p &= a_0\ 0.9^{p} \ , \\ 
 c_p &= c_0\ 0.9^{p} \ , \\ 
  A_p &= A_0\ 1.1^{p} \ , \numberthis
\end{align*}
where $0$ subscript denotes starting values of parameters given in Table~\ref{app:tab:qaoa_params}. 

We note that the results presented in Fig.~\ref{fig:QAOA_estimators} exhibit rather poor convergence of the algorithm in a sense that adding more layers above $p=9$ changed the resulting energies only slightly. 
This can be explained by the fact that we did not perform an exhaustive search over hyperparameters, but can also be a manifestation of the recently reported fact that QAOA might have problems with reaching global minima for relatively complicated Hamiltonians (like high-density MAX-2-SAT used in our work) \cite{Akshay2020reachability}.
Clearly, in the context of our work, only the comparison of noisy and noise-mitigated optimization to the noiseless run was of significance.

\begin{figure}[!h]
\begin{center}
      \includegraphics[width=0.98\textwidth]{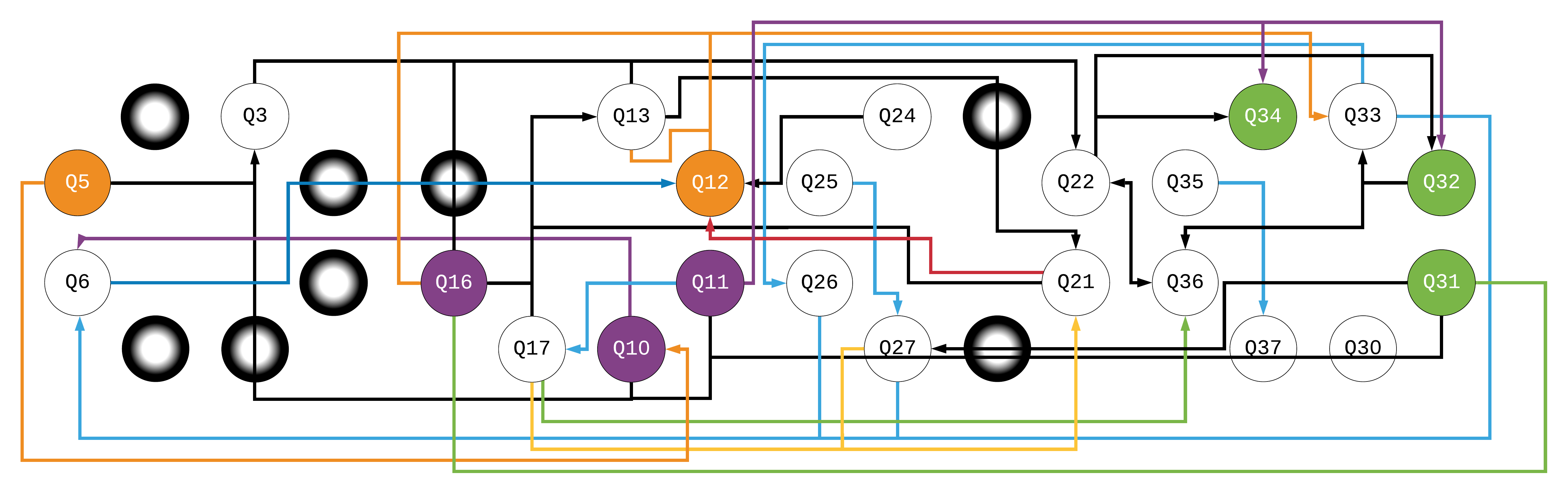}
\caption{\label{fig:full_picture_rigetti} Depiction of correlations in Rigetti device. In Fig.~\ref{fig:correlations_models} we presented the above image splited into two parts for clarity.}
\end{center}
 \end{figure}
 \begin{figure}[!h]
\begin{center}
\includegraphics[width=0.475\textwidth]{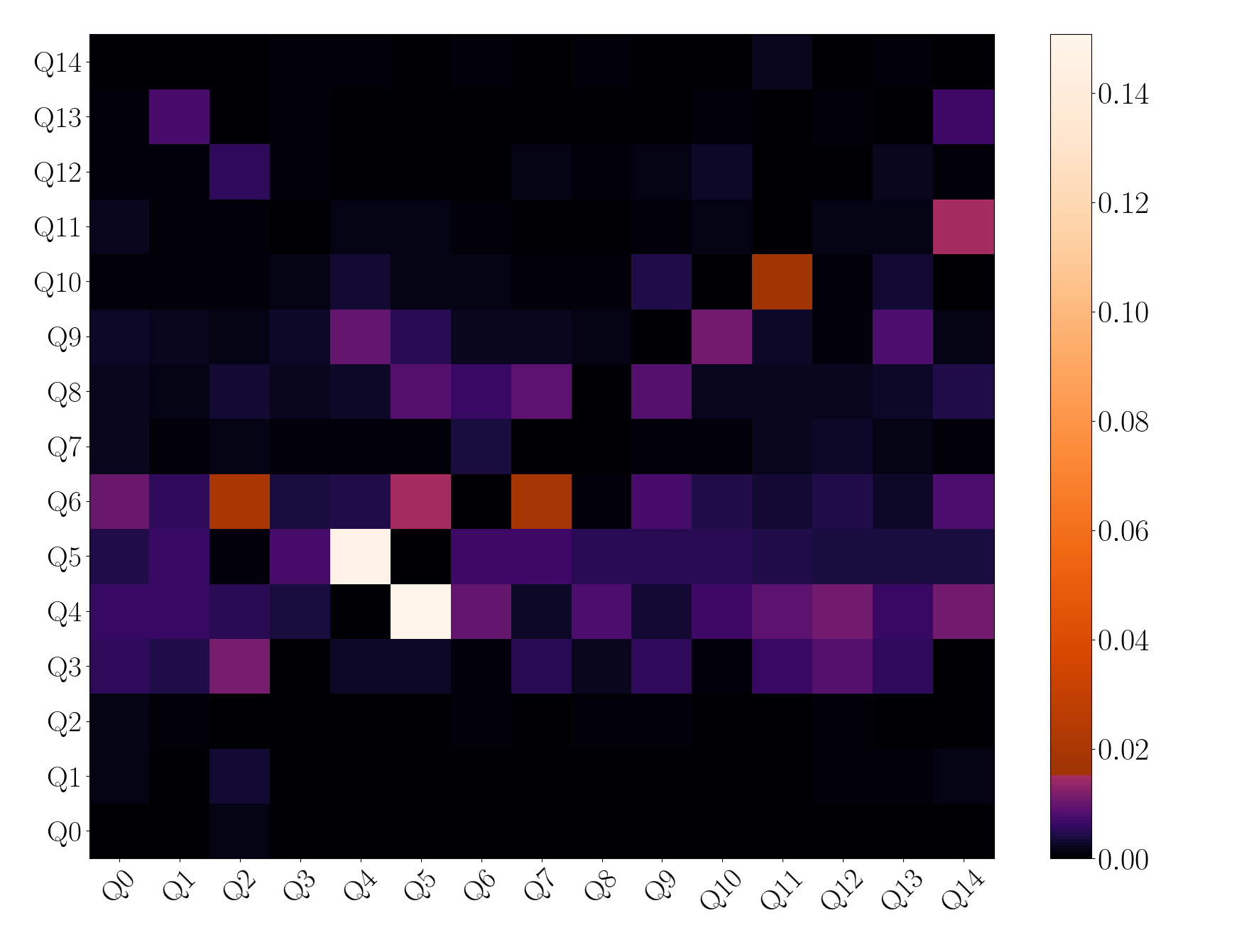} $\quad$
\includegraphics[width=0.475\textwidth]{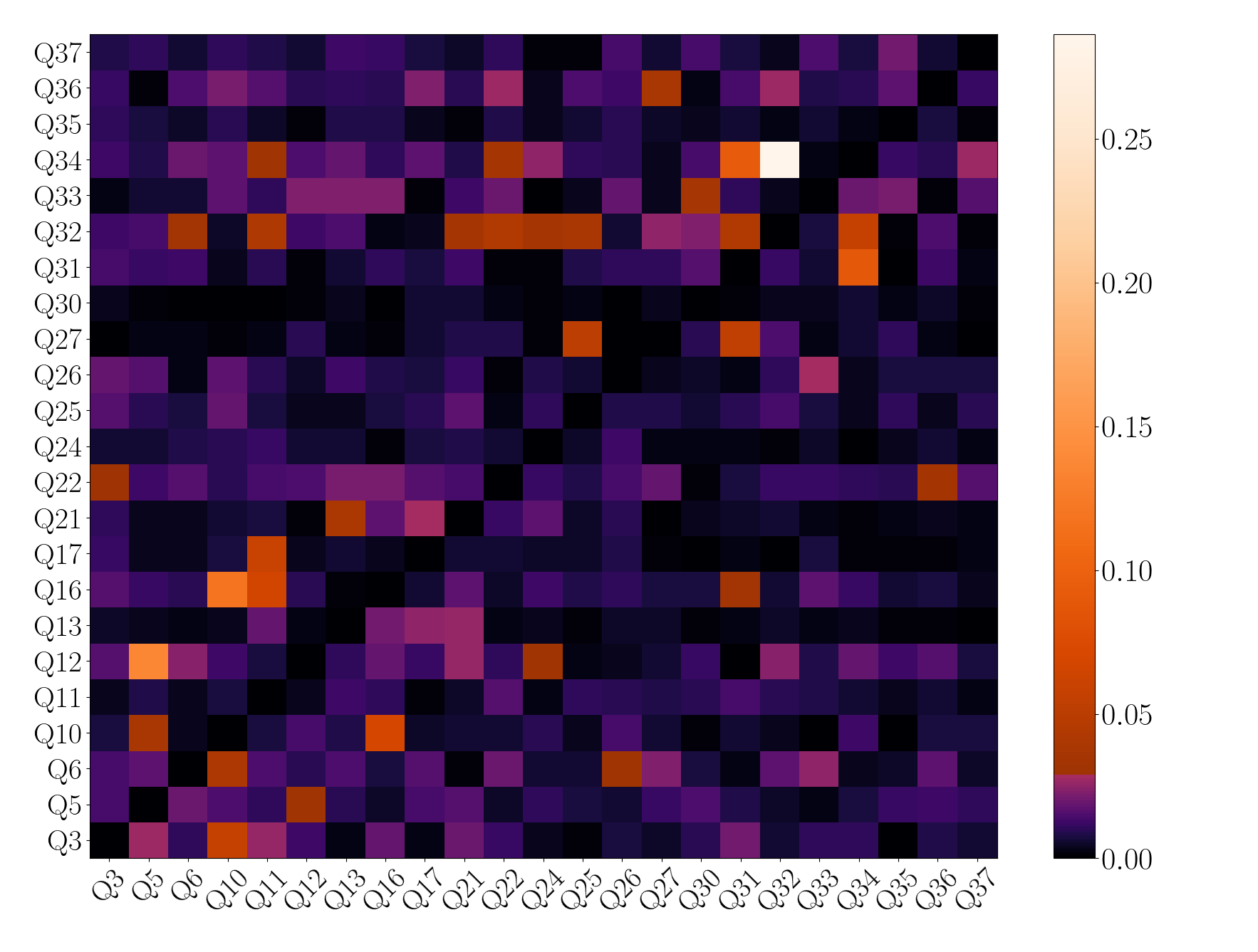}
\caption{\label{fig:heatmaps}Heatmap of the correlations (Eq.~\eqref{eq:correlations_pair}) in IBM's 15-qubit \textit{Melbourne} device (left) and a 23-qubit subset of Rigetti's \textit{Aspen-8} device (right). 
The convention is that ``row is affected by column'', i.e., the measurement noise on the qubit with the label given by row index depends on the state of the qubit with the label given by column index, and the magnitude of the dependence is indicated by colors.}
\end{center}
\end{figure}

\subsection{Additional experimental data\label{sec:app:misc_experiments}}
Here we present additional data concerning correlations reconstructed in DDOT characterization of IBM's 15-qubit \textit{Melbourne} device and a 23-qubit subset of Rigetti's \textit{Aspen-8} device.
The full depiction of correlations in Rigetti's device is presented in Fig.~\ref{fig:full_picture_rigetti}.
The heatmaps showing the reported correlations in both devices are presented in Fig.~\ref{fig:heatmaps}.
In Fig.~\ref{fig:table_parameters_count} we show how many parameters are needed to describe various noise models, and table in Fig.~\ref{fig:table_processing_time} shows how much time data-processing took. 
The data was processed using laptop with 32GB DDR4 RAM (speed 2667MT/s) and Intel(R) Core(TM) i7-9750H CPU @ 2.60GHz. 
We note that no multi-threading was implemented -- in principle, one likely can reduce the run-time further be exploiting parallel calculations.
We intend to optimize the code used for data-processing during development of our online repository \cite{qrem}.

\begin{figure}[!t]
\begin{center}
      \includegraphics[width=0.98\textwidth]{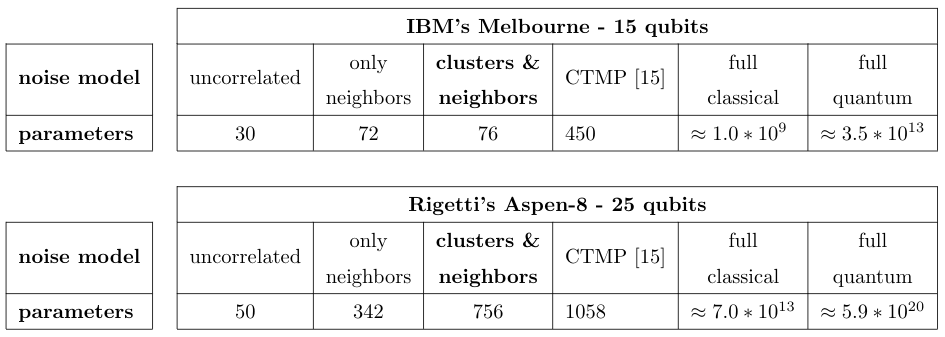}
\caption{\label{fig:table_parameters_count} 
Total number of parameters needed to describe a noise model. 
The uncorrelated noise model is tensor product of single-qubit stochastic noise matrices.
The "only neighbors" model corresponds to considering only single-qubit (trivial) clusters and their neighborhoods estimated in our experiments.
The "clusters \& neighbors" model corresponds to our full noise model containing both non-trivial clusters and their neighborhoods.
CTMP is a number of parameters needed to describe 2-local classical noise model from Ref.~\cite{Bravyi2020mitigating} \textit{without any assumptions} on the correlations structure. 
We note that combining our DDOT characterization with CTMP model could reduce the number of parameters in CTMP by pointing to  negligible correlations that can be disregarded. 
For comparison, the "full classical" noise model refers to generic stochastic map, and "full quantum" to a generic d-outcome POVM.}
\end{center}
\end{figure}

\begin{figure}[!t]
\begin{center}
      \includegraphics[width=0.98\textwidth]{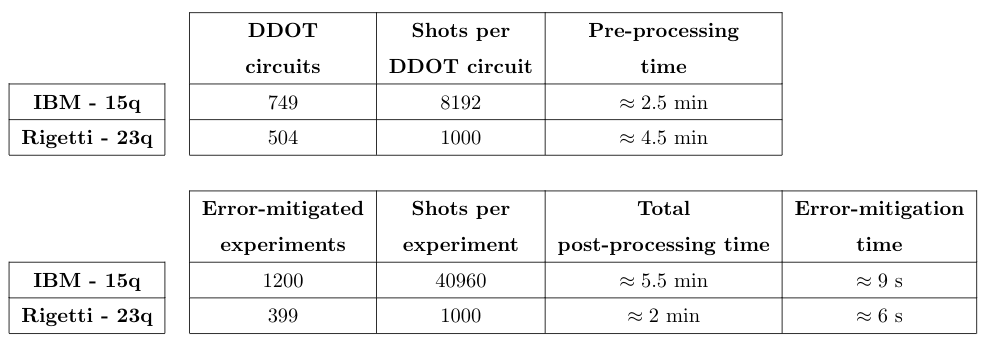}
\caption{\label{fig:table_processing_time} 
Time of data processing. 
In first table, the "pre-processing time" includes calculation of marginal noise matrices on 2-qubit subsets, calculation of pairwise correlations, reconstruction of the noise model and calculation of inverse noise matrices needed for corrections of all possible two-qubit marginal probability distributions (note that this in general includes also higher-dimensional corrections, as explained in Fig.~\ref{fig:example_marginal_subset}).
In the second table, the "Total post-processing time" includes both calculation of marginal distributions needed to estimate energies of investigated 2-local Hamiltonians and performed error-mitigation on them.
The "Error-mitigation time" shows only the latter. 
}
\end{center}
 \end{figure}

\end{document}